\documentclass[lettersize,journal]{IEEEtran}

\usepackage{graphicx}

\usepackage{amsmath, amssymb, amsfonts, amsthm}
\usepackage{bm}

\usepackage{hyperref}
\usepackage{cleveref} 
\usepackage{stmaryrd}
\usepackage{subcaption}
\usepackage{makecell}
\usepackage{ulem}
\usepackage{fontspec}
\usepackage{suffix}
\usepackage{algorithm}
\usepackage{algpseudocode}

\usepackage{mdframed}

\theoremstyle{definition}
\newmdtheoremenv[
  topline=false,
  bottomline=false,
  leftline=false,
  rightline=false,
  linewidth=0.6pt,
  skipabove=0pt,
  skipbelow=6pt,
  leftmargin=0pt,
  rightmargin=0pt,
  innerleftmargin=0pt,
  innerrightmargin=0pt
]{definition}{Definition}

\newtheorem{theorem}{Theorem}[section]

\newtheorem{lemma}{Lemma}
\newtheorem{proposition}{Proposition}
\newtheorem{stipulate}{Stipulate}

\newtheorem{postulate}{Postulate}

\theoremstyle{remark}
\newtheorem*{remark}{Remark}

\crefformat{equation}{(#2#1#3)}
\crefrangeformat{equation}{(#3#1#4)–(#5#2#6)}

\usepackage{enumitem}
\usepackage{needspace}

\newlist{rsq}{enumerate}{2}
\setlist[rsq,1]{
  label=\textbf{RQ\arabic*},
  ref=RQ\arabic*,
  leftmargin=*,
}
\setlist[rsq,2]{
  label=\textbf{RQ\arabic{rsqi}.\alph*},
  ref=RQ\arabic{rsqi}.\alph*,
  leftmargin=2em,
}
\crefname{rsqi}{RQ}{RQs}
\crefname{rsqii}{RQ}{RQs}
\AtBeginEnvironment{rsq}{\Needspace{3\baselineskip}}

\newlist{cq}{enumerate}{2}
\setlist[cq,1]{
  label=\textbf{CQ\arabic*},
  ref=CQ\arabic*,
  leftmargin=*,
}
\setlist[cq,2]{
  label=\textbf{CQ\arabic{cqi}.\alph*},
  ref=CQ\arabic{cqi}.\alph*,
  leftmargin=2em,
}
\crefname{cqi}{CQ}{CQs}
\crefname{cqii}{CQ}{CQs}
\AtBeginEnvironment{cq}{\Needspace{3\baselineskip}}

\newlist{nf}{enumerate}{2}
\setlist[nf,1]{
  label=\textbf{NF\arabic*},
  ref=NF\arabic*,
  leftmargin=*,
}
\setlist[nf,2]{
  label=\textbf{NF\arabic{nfi}.\alph*},
  ref=NF\arabic{nfi}.\alph*,
  leftmargin=2em,
}
\crefname{nfi}{NF}{NFs}
\crefname{nfii}{NF}{NFs}
\AtBeginEnvironment{nf}{\Needspace{3\baselineskip}}

\usepackage{color}

\usepackage{array}
\usepackage{textcomp}
\usepackage{stfloats}
\usepackage{url}
\usepackage{verbatim}
\usepackage{cite}
\usepackage{booktabs}
\usepackage{enumitem}
\usepackage{mathtools}
\usepackage{mathrsfs}
\usepackage{cancel}

\usepackage{listings}

\lstdefinestyle{sparql}{
    basicstyle=\ttfamily\footnotesize,
    breakatwhitespace=false,         
    breaklines=true,                 
    captionpos=b,                    
    keepspaces=true,                 
    showspaces=false,                
    showstringspaces=false,
    showtabs=false,                  
    tabsize=2,
    numbers=left,
    xleftmargin=3em,
    framexleftmargin=2.5em
}
\lstdefinestyle{turtle}{
    basicstyle=\ttfamily\footnotesize,
    breakatwhitespace=false,         
    breaklines=true,                 
    captionpos=b,                    
    keepspaces=true,                 
    showspaces=false,                
    showstringspaces=false,
    showtabs=false,                  
    tabsize=2
}

\newcommand{\sidenote}[1]{}
\newcommand{\future}[1]{\textcolor{blue}{#1}}

\makeatletter
\newcommand{\mypar}{%
  \@ifstar{\mypar@star}{\mypar@nostar}%
}
\newcommand{\mypar@star}[1]{\textbf{#1}}
\newcommand{\mypar@nostar}[1]{\textbf{#1.}}
\makeatother

\newcommand{\qmark}{?{\!}}

\newcommand{\xor}{\oplus}

\newcommand{\component}{\mathcal{C}}

\newcommand{\resource}{\mathcal{R}}
\newcommand{\expectation}{\mathscr{E}}
\newcommand{\fault}{\mathscr{F}}
\newcommand{\dee}{\mathscr{D}}

\newcommand{\capability}{\mathscr{C}}
\newcommand{\ointerface}{\mathcal{O}}
\newcommand{\iinterface}{\mathcal{I}}

\newcommand{\channel}{\vec{\mathcal{C}}}

\newcommand{\function}{\mathbb{F}}

\newcommand{\fdep}{\dee^A}
\newcommand{\fdepb}{\dee^B}
\newcommand{\fdepc}{\dee^C}
\newcommand{\fdgnamefull}{{\text{Capability Interaction Graph}}} 
\newcommand{\fdgname}{{\text{CIG}}} 
\newcommand{\fdgnamearticle}{{\text{a}}} 
\newcommand{\fdg}{\mathcal{G}}
\newcommand{\fdginstance}{\Gamma}
\newcommand{\faultcreation}{\lightning}

\newcommand{\powsymbol}{\mathscr{P}}
\newcommand{\pow}[1]{\powsymbol(#1)}
\newcommand{\preq}{\psi}

\newcommand{\id}{\iota}

\newcommand{\faultinstance}[1]{\hbox{\sout{$#1$}}}
\newcommand{\channelcons}{\mathrm{cons}}
\newcommand{\channelprod}{\mathrm{prod}}

\newcommand{\quaindcomp}{\varpi}
\newcommand{\quaindchannel}{\varsigma}

\newcommand{\providedby}{\pi}
\newcommand{\provides}{\pi}

\newcommand{\capid}{\id}
\newcommand{\faultid}{\id}
\newcommand{\compid}{\id}
\newcommand{\chanid}{\id}
\newcommand{\interfaceid}{\id}

\newcommand{\mytext}[1]{\textrm{#1}}
\newcommand{\ufo}[1]{\mytext{#1}} 

\newcommand{\err}{\varepsilon}
\newcommand{\consdv}{\delta}

\begin{document}

\title{Ontology-Grounded {\fdgnamefull}s: From Knowledge Graphs to Fault Trees}

\author{Manzi Aimé Ntagengerwa,
Georgiana Caltais,
Mariëlle Stoelinga, University of Twente, The Netherlands}

\maketitle

\begin{abstract}

The development of Cyber-Physical Systems (CPSs) is inherently multidisciplinary, involving expertise from domains such as software engineering, electrical engineering, and mechatronics.
throughout the lifecycle of the system, from design to deployment. 
Ensuring system reliability in Cyber-Physical Systems (CPSs)
requires the identification and analysis of potential failures and their cascading effects. However, reliability modeling remains a challenging and error-prone activity, as it often depends on tacit expert knowledge, incomplete documentation of failure modes, and limited consideration of interactions between subsystems.

To address these challenges, this paper introduce the {\fdgnamefull} ({\fdgname}), an ontology-driven representation of CPS architectures grounded in the Unified Foundational Ontology (UFO). Due to its graph-based structure, {\fdgnamearticle} {\fdgname} is naturally represented as a knowledge graph (KG), enabling the explicit capture of functional dependencies and system semantics.

Building upon this representation, we propose an automated synthesis algorithm for generating Fault Trees (FTs) directly from {\fdgname}s encoded as knowledge graphs. Fault Tree Analysis provides an effective mechanism for evaluating critical failure properties, including failure propagation paths and minimal cut set sets. 
Our approach reduces this complexity by leveraging  {\fdgname}s and knowledge graphs. We  provide a common semantic representation across engineering domains and support the automated generation of reliability models.

\end{abstract}

\section{Introduction}

\newcommand{\myparagraph}[1]{\medskip {\it #1}}

Cyber-Physical Systems (CPSs) are engineered through the interaction of multiple disciplines, including software engineering, electrical engineering, mechanics, and control. During early design stages, engineers must reason about system reliability despite incomplete knowledge, heterogeneous models, and evolving architectures. Yet many reliability analyses are still performed only after substantial design decisions have already been fixed, when changes become costly and difficult to implement\cite{incose2023incose}.

\paragraph*{Fault tree analysis and knowledge graphs}
Fault Tree Analysis (FTA) is one of the most widely used techniques for analyzing system reliability and safety. Fault trees (FTs) provide a structured representation of how lower-level faults combine into system-level failures, enabling analyses such as minimal cut sets, failure propagation paths, and criticality assessment. Despite the maturity of FTA algorithms, constructing fault trees remains largely manual. In practice, FT development requires extensive coordination between domain experts and depends heavily on tacit engineering knowledge. As a consequence, fault trees are expensive to maintain, difficult to integrate across subsystem boundaries, and often unavailable during the early stages of CPS design.

Model-Based Systems Engineering (MBSE) and knowledge graphs (KGs) offer an opportunity to address this limitation. Modern CPS development already produces large amounts of structured architectural information describing components, interfaces, and interactions. Knowledge graphs are particularly attractive because they provide a unified semantic representation capable of integrating heterogeneous engineering artifacts under a common ontology. However, existing approaches still lack a principled intermediate representation that connects architectural knowledge to fault propagation semantics in a formally grounded way.

\paragraph*{{\fdgnamefull}s} This paper introduces the {\fdgnamefull} ({\fdgname}), a novel formal model for representing functional dependencies in CPS architectures. The {\fdgname} is the central contribution of this work. It provides an ontologically grounded representation of components, capabilities, channels, interfaces, functions, faults, and errors, together with the dependency structure that governs fault propagation. Unlike traditional FT models, {\fdgname}s explicitly capture the distinction between:

\begin{itemize}
    \item Faults and errors,
    \item Events and situations, and
    \item Fault creation and fault activation.
\end{itemize}

These distinctions are essential for accurately modeling how failures emerge and propagate in cyber-physical architectures, yet they are largely absent from conventional fault tree formalisms.

The {\fdgname} serves as an intermediate semantic layer between knowledge graphs and synthesized fault trees. Starting from a minimally structured CPS architecture represented in a KG, we automatically construct {\fdgnamearticle} {\fdgname} that captures the system’s functional dependency structure. Fault trees are then synthesized directly from the {\fdgname} semantics. This approach enables the automated generation of meaningful FT models under minimal modeling assumptions while preserving traceability to the underlying system architecture.

A key aspect of our approach is its ontological foundation. The {\fdgname} model and its supporting ontology are grounded in the Unified Foundational Ontology (UFO), providing precise semantics for notions such as capability manifestation, participation, causality, and error propagation. At the same time, the ontology remains intentionally lightweight so that it can absorb information from heterogeneous engineering models and remain applicable in early design stages where detailed behavioral information may not yet exist.

Thus, this paper connects foundational ontology, knowledge-graph reasoning, and fault-tree analysis into a unified framework for modeling and analyzing failures in Cyber-Physical Systems.







\autoref{fig:methodology} summarizes the information flow. Systems architecture descriptions, conforming to the proposed ontology, are represented as a knowledge graph. These KGs are then queried using SPARQL to obtain {\fdgnamearticle} {\fdgname}, and automatically transformed into fault trees suitable for standard FTA techniques.

\begin{figure}
    \centering
    \includegraphics[width=\linewidth]{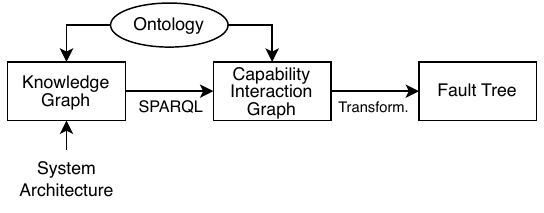}
    \caption{Proposed pipeline for the automatic synthesis of fault trees from a knowledge graph.}
    \label{fig:methodology}
\end{figure}

\subsection{Overview of the Approach}

To derive fault trees from architectural knowledge represented as a knowledge graph, we proceed through the following  transformations. Each introduces a new level of abstraction, adding the information required for reliability analysis.

\begin{enumerate}
    \item \textbf{Representation of the system architecture.}
    We begin with a description of the system architecture consisting of components, channels, and their interactions. The ontology provides explicit semantics for concepts such as components, capabilities, functions, interfaces, faults, and errors.

    At this stage, the model describes {\it what the system is} and {\it how its parts are connected}, but it does not yet describe failure propagation.

    This information is later represented as a knowledge graph conforming to an ontology grounded in UFO. 

    \item \textbf{Introduce failure semantics.}

    We then define how faults, fault activations, and errors are represented within the {\fdgname}. Faults are associated with capabilities, while errors arise when expected capabilities fail to manifest. Because capabilities depend on resources produced elsewhere in the graph, an error in one location may propagate to dependent capabilities.

    This step transforms the {\fdgname} from a purely architectural model into a model that explains {\it how failures can spread through the system}.

    \item \textbf{Extract the functional dependency structure.}

    The {\fdgname} identifies the capabilities present in the system and the dependencies between them. In particular, it captures how capabilities consume and produce resources, and how resources are exchanged through channels and interfaces.

    The {\fdgname} therefore makes explicit {\it which functions depend on which other functions}.

    \item \textbf{Characterize the causes of errors.}
    Using the failure semantics, we derive an error structure function for each capability. This function recursively describes the conditions under which the capability becomes erroneous. A capability may fail because it is directly affected by a fault, or because one of its required resources are unavailable due to failures elsewhere in the {\fdgname}.

    The resulting structure provides a complete causal explanation of the capability's failure.

    \item \textbf{{\fdgname}s from KGs.}
    Now that the structure and semantics of {\fdgname}s have been defined, we show how well-formed KGs (those that conform to the ontology) can be queried using SPARQL to construct {\fdgname}s.
    
    \item \textbf{Synthesize a fault tree.}
    Finally, we translate the {\fdgname} into a fault tree using its error semantics. The capability under analysis becomes the top-level event, direct and indirect causes are recursively expanded, and logical relationships are represented by fault-tree gates.

    Through induction we show that the resulting fault tree is semantically grounded in the original architecture and can be readily analyzed using standard Fault Tree Analysis techniques such as minimal cut-set computation and criticality analysis.
\end{enumerate}

In summary, the development of the paper follows the chain 
$
\text{{\fdgname}}
\;\rightarrow\;
\text{Failure Semantics}
\;\rightarrow\;
\text{Error Structure Function}
\;\rightarrow\;
\text{KG}
\;\rightarrow\;
\text{FT}.
$

Each step enriches the previous representation with additional semantics, ultimately transforming architectural knowledge into a form suitable for reliability analysis.

\paragraph{Running example}
\autoref{fig:print_head} shows a subsystem of a production printer. There are three components: the power supply unit (PSU), the ink reservoir, and the print head. The PSU has the capability to provide electricity, the ink reservoir can provide ink, and the print head can deposit ink and self-clean. All of these capabilities, except for the self-cleaning of the print head, are manifested as \textit{functions} in the presented situation. The output of the PSU and the ink reservoir reach the print head through a wire and a tube respectively.

\begin{figure}
    \centering
    \includegraphics[width=\linewidth]{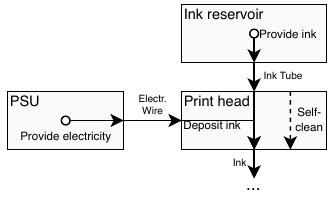}
    \caption{A simplified diagram of a production print head, illustrating its components, capabilities and interfaces.}
    \label{fig:print_head}
\end{figure}

\paragraph{Organization of the sections}
In \autoref{sec:background}, we introduce explain the concepts of knowledge graphs, formal ontology and fault trees.
In \autoref{sec:ontology}, we propose a novel ontology that describes the concepts of components, capabilities and component interactions in greater detail.
In \autoref{sec:fdg_formalization} we combine these ontological concepts to form the formal {\fdgnamefull} ({\fdgname}).
\hyperref[sec:dynamics_and_events]{Section~\ref*{sec:dynamics_and_events}} describes the dynamic evolution of {\fdgname}s over time through events, and in \autoref{sec:fault_creation_error_fault_activation} we specifically consider fault creation and activation, and error behavior.
\hyperref[sec:fdep]{Section~\ref*{sec:fdep}} formalizes the notion of functional dependency.
We introduce the mechanism of error propagation in \autoref{sec:error_propagation}, and in \autoref{sec:structure_function} we give the formal interpretation of error semantics in {\fdgname}s.
In \autoref{sec:kg_query_fdg_construction} we show how {\fdgname}s are constructed from well-formed KGs.
In \autoref{sec:transformation} we give the semantic mapping between {\fdgname}s and FTs.
In \autoref{sec:transformation_properties}, we mathematically prove completeness of the {\fdgname}-FT transformation w.r.t. preservation of error semantics.
Finally, \autoref{sec:related_work} discusses related work.


\section{Background}\label{sec:background}

\begin{figure}[ht!]
\centering
\begin{minipage}[t]{.23\textwidth}
    \includegraphics[width=\linewidth]{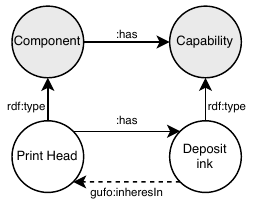}
    \caption{A small KG that instantiates the ontology based on the running example.}
    \label{fig:kg_example}
\end{minipage}{\ }
\begin{minipage}[t]{.23\textwidth}
    \includegraphics[width=\linewidth]{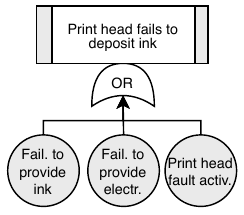}
    \caption{A simple FT that models the failure modes of the running example.}
    \label{fig:ft_example}
\end{minipage}
\end{figure}

\subsection{Knowledge Graphs}\label{sec:background_kg}
Knowledge graphs (KGs) are a flexible knowledge-base framework that allow for the integration of knowledge from various sources, and the reasoning over known facts to infer new knowledge. KGs typically contain knowledge over a particular domain. The concepts in this domain, and how they relate to each other, can be encoded explicitly with an ontology. Rule-based fact inference then reasons over the instances of the ontology. A prominent example of a KG is Wikidata\cite{Vrandecic2014}, the knowledge source of the online encyclopedia Wikipedia.

\autoref{fig:kg_example} shows an example of a KG,
formalizing a fragment of the diagram shown in \autoref{fig:print_head}. The gray elements show the relevant parts of the ontology. White elements are instances, and the black lines between the white elements are factual relationships between instances.

We use gUFO~\cite{almeida2026gufogentlefoundationalontology}, a lightweight OWL implementation of the Unified Foundational Ontology \cite{ufo_a}, to represent UFO concepts within RDF knowledge graphs. gUFO provides machine-readable semantics for notions such as objects, events, dispositions, and relators, enabling ontology-based reasoning over CPS architectures.

The dashed \texttt{gufo:inheresIn} association between the PSU and the capability to \textit{produce electricity} is one we can infer: 
Imagine an inference rule that states
\begin{multline}
\forall \alpha, f. {\ } \mytext{Component}(\alpha), \mytext{Capability}(f) \implies\\
\mytext{has}(\alpha, f) \iff \mytext{inheresIn}(f, \alpha)
\end{multline}
I.e. if a component has a capability, then we say that the capability inheres in that component and vice versa.
Applying this rule on the KG in \autoref{fig:kg_example}, we insert a new fact $(f, \texttt{gufo:inheresIn}, \alpha)$ that states \textit{``The capability to produce electricity inheres in the PSU component''}.

In this work, we use a GraphDB instance as our KG for its gUFO support. GraphDB is an open standard and it has a straight-forward connection to predicate logic. GraphDB implements a Resource Description Framework (RDF) triple-store\cite{w3c_rdf}, which is an open W3C standard. All RDF triples are of the form $(\textit{subject}, \textit{predicate}, \textit{object})$, where each entity is identified by its uniform resource identifier (URI)\cite{uri_rfc}. All facts in the KG have the form $xPy$ and can be described in predicate logic as a binary predicate $P(x, y)$. RDF triple-stores are queried using SPARQL\cite{w3c_sparql}.

SPARQL (SPARQL Protocol and RDF Query Language) is the standard query language for retrieving and manipulating data stored in RDF graphs. SPARQL enables users to express graph-pattern queries to match triples. The language supports core operations such as selection, filtering, aggregation, optional matching, unions, and federated queries across distributed data sources. Through these capabilities, SPARQL allows both flexible exploration of heterogeneous linked data and precise extraction of structured relationships.

\subsection{Fault Trees and Fault Tree Analysis}
Fault trees (FTs) are risk models that enable fault tree analysis (FTA)\cite{vesely1981fault}. There exist many variations of FTs\cite{Ruijters2015}. In this work, we focus on the well-known Static FT (formally defined in \autoref{sec:ft_def}) that allows, among other things, for the modeling of common cause failures\cite{jones2012common}. This makes FTs directed, acyclic graphs (DAGs) and not strictly trees, contrary to what their name suggests.

FTs model all foreseen failure modes that result in the particular top-level event (TLE) of interest. The leaves of the trees, so-called basic events (BEs) are stochastically independent events that are not decomposed further. The conditions under which these BEs propagate to result in the TLE are modeled by use of \texttt{AND} / \texttt{OR} logic gates.

\autoref{fig:ft_example} shows an example of a simple fault tree based on the diagram in \autoref{fig:print_head}. The TLE is a situation in which the print head fails to deposit ink. The failure of the PSU to provide electricity, the failure of the ink reservoir to provide ink, and the fault activation of the print head component are the basic events. In natural language, the FT expresses that ``The print head fails to deposit ink when it is faulty, or when either the PSU or the ink reservoir fail to deliver.''

Fault Tree Analysis (FTA) is a widely used deductive safety and reliability engineering technique that calculates metrics over FTs. FTA enables qualitative reasoning about system vulnerabilities by identifying minimal cut sets: the smallest combinations of basic events sufficient to cause the top event. Qualitative analysis also supports ranking critical contributors, revealing single points of failure, exposing common-cause dependencies, and improving understanding of fault propagation paths. Because of FTA's traceable logic, it remains valuable during system design, hazard assessment, and certification activities.

In addition to these structural insights, FTA also supports quantitative evaluation when failure data are available. Probabilities or failure rates can be assigned to basic events and propagated through the tree to estimate the likelihood or frequency of the top event, as well as importance measures for individual contributors. However, in many early design and model-based engineering contexts, qualitative metrics are especially useful because they can be derived without precise statistical data and still guide architectural decisions, redundancy strategies, and risk reduction efforts.

\subsection{The Unified Foundational Ontology (UFO)}

The Unified Foundational Ontology (UFO)~\cite{ufo_a} is a foundational ontology developed to provide precise semantic distinctions for conceptual modeling. UFO has been widely adopted in conceptual engineering and enterprise modeling because it offers a rigorous treatment of objects, events, dispositions and relations. The ontology is accompanied by the modeling language OntoUML~\cite{ufo_a}, which provides graphical constructs grounded in UFO semantics.

UFO distinguishes between endurants and perdurants. {\it Endurants} are entities that are wholly present whenever they exist, such as physical components, channels, or interfaces. {\it Perdurants} are entities that unfold in time and possess temporal parts, such as functions or errors. In our setting, CPS components, channels and capabilities are modeled as endurants, whereas functions, fault activations, and errors are modeled as events (i.e., perdurants).

A central notion in UFO is that of moments. {\it Moments} are dependent entities that cannot exist without another entity, called their bearer. In particular, intrinsic moments inhere in exactly one bearer. An intuitive example is a headache -- a person's headache cannot exist without the particular person (and only \textit{that} person) in whom it inheres~\cite{ufo_b}. Formally, UFO provides the relation
\[
\mathit{inheresIn}(m,o)
\]
to denote that moment $m$ inheres in object $o$. In this work, capabilities and faults are modeled as intrinsic moments that inhere in components or channels.

UFO further characterizes dispositions, which are moments that may manifest under certain triggering conditions. A capability such as ``providing electricity'' is therefore not itself an event, but a disposition that can be manifested through an event. UFO relates manifestations and dispositions through the relation
\[
\mathit{manifests}(e,d)
\]
meaning that event $e$ is the manifestation of disposition $d$. In our model, functions are manifestations of capabilities.

Another important distinction concerns situations. A situation represents a snapshot of the world at a particular time point.
Events transform one situation into another.
In this work, we introduce a syntax that facilitates the description of the preconditions and postconditions of events:
\[
\Gamma \xrightarrow{e} \Gamma'
\]
where event $e$ begins in situation $\Gamma$ and brings about situation $\Gamma'$. This treatment allows us to formalize fault creation, fault activation, and error propagation as temporally ordered events.

UFO also provides a formal treatment of causality. An event may bring about a situation that subsequently triggers another event, thereby establishing a causal chain. This distinction is particularly important in reliability modeling because it separates:
\begin{itemize}
    \item the creation of a fault,
    \item the activation of a fault under operational conditions,
    \item the resulting error state.
\end{itemize}
Traditional fault tree approaches often conflate these notions, whereas UFO enables them to be represented explicitly.

Finally, UFO distinguishes between objects, roles, relators, and events. We exploit these distinctions in the proposed ontology:
\begin{itemize}
    \item components and channels are modeled as objects,
    \item capabilities and faults are modeled as dispositions,
    \item producer and consumer are modeled as roles,
    \item interfaces are modeled as relators mediating interactions, and
    \item functions, errors, and fault activations are modeled as events.
\end{itemize}

Grounding the proposed {\fdgname} model in UFO provides precise semantics for participation, manifestation, dependency, and causality, while remaining lightweight enough for early-stage CPS architecture modeling.

\section{The Ontological Foundations of {\fdgnamefull}s}\label{sec:ontology}



We use OntoUML to formally visualize the elements of {\fdgname}s in \autoref{fig:ontology}.
We argue that this conceptualization is small enough to capture knowledge in the earliest stages of CPS design, yet generic enough to capture concepts across modeling languages and engineering domains, and sufficiently rich for the synthesis of meaningful fault trees.

Note that not all boxes in this diagram need to be provided by a user instantiating {\fdgnamearticle} {\fdgname}. The only input we strictly require (for the purpose of any \textit{meaningful} fault tree synthesis) is a set of components, channels, and a set of labeled edges that connect them. All the other concepts can implicitly be instantiated under a fixed set of assumptions.


\subsection{Formal notation}
Central to the {\fdgname} model is UFO's notion of situations. In essence, a situation describes a system under study at a single point in time. By convention, we here denote a situation $\fdginstance$.







Furthermore, we use some standard mathematical notations for our formal definitions. We introduce them here for convenience. Let $\mytext{proj}_i$ be the function that maps a tuple to its $i^{\mytext{th}}$ element. I.e. let $T$ be a tuple, such that
\[
T = \left( x_1, \dots, x_n \right)
\]
Then
\[
\mytext{proj}_i {\ } T = x_i
\]

Furthermore, we use the common FOL syntactic shorthand for quantifying over members of a set:
\begin{align*}
    \forall x \in S. {\ } \phi (x) &\equiv \forall x. \bigl( S(x) \implies \phi (x) \bigr)\\
    \exists x \in S. {\ } \phi (x) &\equiv \exists x. \bigl( S(x) \wedge \phi (x) \bigr)
\end{align*}

Lastly, let $\powsymbol$ denote the finite powerset (as we only work with finite sets).

\subsection{Concepts and formal definitions}\label{sec:reference_ontology}

Endurants (such as objects) are \textit{present in} a situation. E.g. the apple that fell on Sir Isaac Newton's head was present at the moment of impact. That same apple is not present now. We will introduce the elements of {\fdgnamearticle} {\fdgname} in the context of a particular situation. It suffices to know that this means \textit{the elements of the system that are present at that particular time point}.
In this section, we introduce sets of individuals that ``are present''. For instance, we collect components in the set $\component$ and we indicate the situation in which these objects are present by means of a subscript. I.e. the set $\component_\fdginstance$ collects exactly those object which are components and that are present in the situation $\fdginstance$.

In contrast, perdurants (events) \textit{exist throughout} a time interval. E.g. the event of that particular apple falling began in the situation where it had just detached from the stem, and ended in the situation that it reached Sir Newton's cranium (all the while accumulating temporal parts). Section \ref{sec:dynamics_and_events} formalizes this notion in the context of {\fdgname}s.

The bold-print terms introduced below refer to the endurant concepts shown in \autoref{fig:ontology}. We use the running example in \autoref{fig:print_head} to help the reader in building an intuition of these concepts.

\begin{figure*}[h!]
\centering
\includegraphics[width=\linewidth]{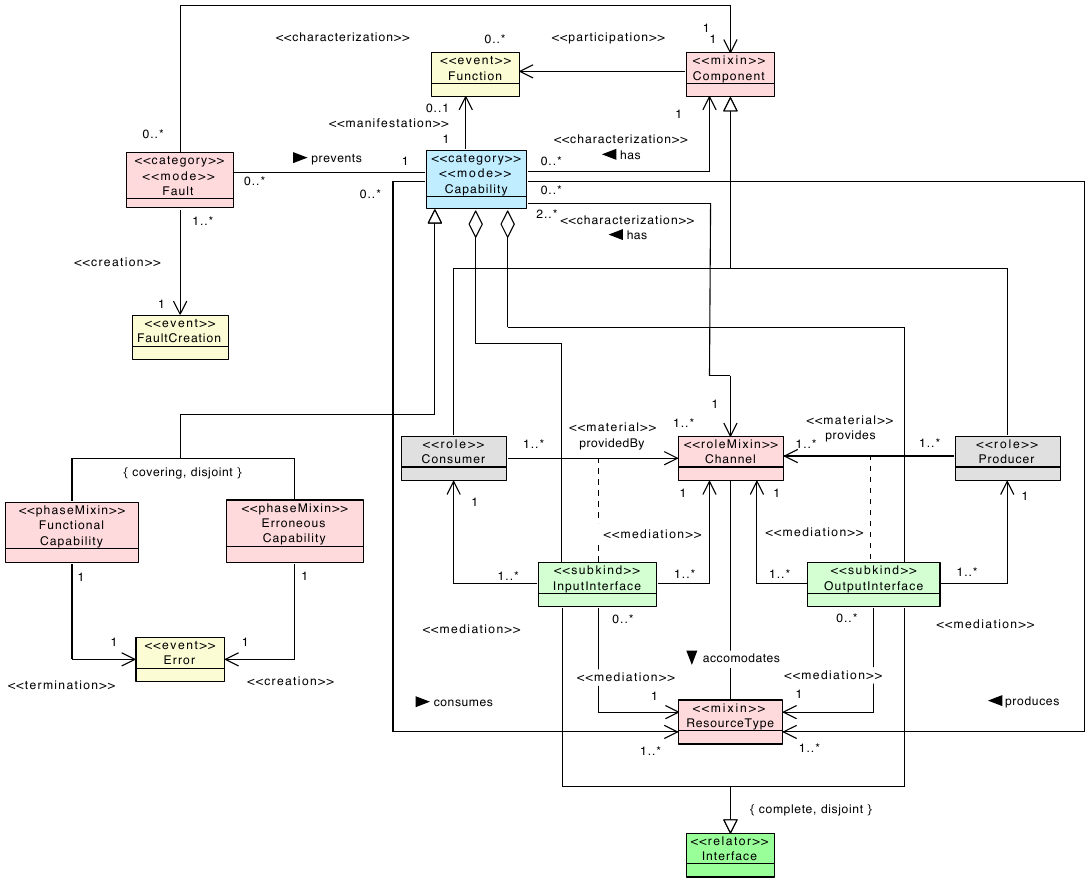}
\caption{Our proposed ontology of CPS architecture, introducing the notions of components, capabilities, functions, channels, etc.}\label{fig:ontology}
\end{figure*}



\textbf{Resource Types} are \textit{higher-order types} that denote the parameter and output types of a function. When instantiating the model, instances of \textit{Resource Type} (e.g. \textit{paper} or \textit{electricity}) are types themselves. To be precise, the instances of resource types are \textit{substantial types} (types whose instances are founded on matter).
The assertion that a type is a resource type is given by the unary predicate $\textit{ResourceType}: \textit{Individual} \rightarrow \{ \top, \bot \}$. This predicate comes from the ontology. Looking at \autoref{fig:print_head}, we find that $\textit{ResourceType}(\texttt{Ink}) = \top$, and $\textit{ResourceType}(\texttt{PSU}) = \bot$.

Resource types generalize any form of energy in the sense of \cite{feynman1965}, such as materials, fluids or signals. This broad definition covers many of the fluids, gases, control signals, etc. that are involved in the design of a CPS architecture.
The only meaningful semantics of resource types are conveyed by their relationships to ontologically richer concepts like capabilities and channels. \autoref{fig:print_head} shows two resource types; \textit{electricity} and \textit{ink}.

\begin{definition}[Resource Types]
    The set of resource types, denoted by the set $\resource = \left\{ r {\ } | {\ } \mytext{ResourceType}(r) \right\}$ represents a higher-order type.
\end{definition}
Elements of $\resource$ are sortal \textit{types} themselves. We do not consider individuals that instantiate resource types in this work. For instance, $\textit{Page} \in \resource$ is a valid resource type, but any particular sheet of paper (e.g. \textit{page 42 of my copy of ``A Room of One's Own''}) is not.
As a result, resource types are purely symbolic elements in $\resource$, whose symbol \textit{is} the (nominal) identity of the type. I.e., the elements $\mytext{Page} \in \resource$ and $\mytext{Ink} \in \resource$ are distinct because their symbols are distinct. Likewise, the elements $\mytext{Ink} \in \resource_\fdginstance$ and $\mytext{Ink} \in \resource_{\fdginstance'}$ are the same type, because their symbols are identical. For this reason, we generally do not subscript $\resource$ with a situation.

\textbf{Capabilities} are \textit{dispositions} (or more generally; \textit{modes}) that describe the possible behavior of a component. They inhere in exactly \underline{one} component. Dispositions are intrinsic moments that are existentially dependent on their bearer. They are particular, i.e. one capability cannot inhere in two distinct bearers.

A capability can be manifested as a function. We distinguish the \textit{Functional} and the \textit{Erroneous} phases of capabilities. A capability is in the erroneous phase only during the existence of an error over that capability (and otherwise it is in the functional phase). The generalization set of these phases is complete and disjoint (components are either functional or erroneous, but not both).
The assertion that an individual is a capability is given by the unary predicate $\textit{Capability}_\fdginstance: \textit{Individual} \rightarrow \{ \top, \bot \}$. In \autoref{fig:print_head}, we have $\textit{Capability}(\texttt{Provide electricity}) = \top$ and $\textit{Capability}(\texttt{Ink}) = \bot$.

The assertion that a capability is in the functional or erroneous phase is given by $\textit{FunctionalCapability}_\fdginstance: \textit{Individual} \rightarrow \{ \top, \bot \}$ and $\textit{ErroneousCapability}_\fdginstance: \textit{Individual} \rightarrow \{ \top, \bot \}$ respectively.


In \autoref{fig:print_head}, the capabilities are to \textit{provide electricity}, to \textit{provide ink}, and to \textit{deposit ink} and \textit{self-clean}. They inhere in the PSU, ink reservoir and print head respectively. The capability to \textit{self-clean} is not manifested in the presented situation.

\begin{definition}[Capabilities]
    Capabilities $\capability_\fdginstance \subseteq \mytext{CapabilityId}_\fdginstance \times \pow{\resource} \times \pow{\resource}$ is a set of tuples, where the first element ranges over $\mytext{CapabilityId}_\fdginstance$.
\end{definition}
We define the bijective function $\capid_\fdginstance$ that maps the elements in $\capability_\fdginstance$ to an individual in the ontology.
\[
\capid_\fdginstance : \capability_\fdginstance \rightarrow \mytext{CapabilityId}_\fdginstance
\]
 Its inverse is denoted $\capid_\fdginstance^{-1}$.

The second and third element of the tuple respectively denote the resource types its bearing component consumes and produces when the capability is manifested.

Capabilities inhere in a unique object (\textbf{D1}, \autoref{tab:ufo_b_axioms}):
\begin{equation*}
    \forall f \in \capability_\fdginstance. {\ } \ufo{inheresIn}_\fdginstance(\capid_\fdginstance(f), \beta(\capid_\fdginstance(f)))
\end{equation*}

\textbf{Functions} are \textit{atomic events} that are the manifestation of a capability. The bearing component of that capability participates in the function event. UFO-B\cite{ufo_b} defines that the begin point of this manifestation is the end point of the activation of the capability it manifests. The conditions for this activations must therefore be met at this begin point, as well as for the entire duration of the function. Such conditions are given by the capability's \textit{consumes} relation.

A function effects a transformation of resources, describing the behavior of a component. In this work, the mechanics of such transformations are not explicitly modeled.
The assertion that an individual is a function is given by the unary predicate $\textit{Function}_\fdginstance: \textit{Individual} \rightarrow \{ \top, \bot \}$.

\begin{definition}[Functions]\label{def:functions}
    The set $\function_\fdginstance = \mytext{FunctionId}_\fdginstance$ is a set of function identifiers. Elements of this set denote which capabilities are currently manifested as a function event.
    
\end{definition}

    The set $\function_\fdginstance$ is exactly the set of symbols that represent a function event in $\fdginstance$:
    \begin{equation*}
        \function_\fdginstance = \left\{ \bar{f} {\ } | {\ } \mytext{Function}_\fdginstance(\bar{f}) \right\}
    \end{equation*}
    
    The set $\mytext{FunctionId}_\fdginstance$ is derived from the set of capability identifiers $\textit{CapabilityId}_\fdginstance$ by a partial bijection
    \[
    \varrho_\fdginstance : \mytext{CapabilityId}_\fdginstance \rightarrow \function_\fdginstance
    \]
    
    \noindent and
    \begin{multline}
        \text{dom} \varrho_\fdginstance =\\
        \{ c {\ } | {\ } \mytext{Capability}_\fdginstance(c) \wedge \exists e : \mytext{Event}. {\ } \ufo{manifests}(e, c) \}
    \end{multline}
    \noindent and
    \begin{multline}
    \forall f \in \text{dom} \varrho_\fdginstance. \forall \bar{f} \in \function_\fdginstance.\\
    \varrho_\fdginstance(f) = \bar{f} \iff \ufo{manifests}(\bar{f}, f)
    \end{multline}

    Where $\text{dom} \varrho_\fdginstance$ denotes the domain of $\varrho_\fdginstance$. Furthermore, since $\varrho_\fdginstance$ is a partial bijection, \textit{if} a function manifests a capability, it is the \textit{unique} manifestation of that capability in $\fdginstance$.

    Let $\varrho^{-1}_\fdginstance : \textit{FunctionId}_\fdginstance \rightarrow \textit{CapabilityId}_\fdginstance$ be the inverse function of $\varrho$ that returns the unique capability that a given function manifests (\textbf{D2}, \autoref{tab:ufo_b_axioms}):
    \begin{equation*}
        \forall \bar{f} \in \function_\fdginstance. {\ } \ufo{manifests}(\bar{f}, \varrho^{-1}_\fdginstance(\bar{f}))
    \end{equation*}

    Note that $\varrho^{-1}_\fdginstance$ is injective where $\text{dom} \varrho^{-1}_\fdginstance = \function_\fdginstance$.

    In our notation, we use symbol names to indicate the manifestation relationship between a capability and a function. Bar-notation indicates that $\bar{f}$ is a function that manifests the capability $f$.
    
    

    

\textbf{Faults} are \textit{modes} that describe the possible prevention of the manifestation of a component's capabilities. They inhere in exactly \underline{one} component.
The assertion that an individual is a fault is given by the unary predicate $\textit{Fault}: \textit{Individual} \rightarrow \{ \top, \bot \}$.

\begin{definition}[Faults]\label{def:faults}
Faults $\fault_\fdginstance \subseteq \mytext{FaultId}_\fdginstance \times \pow{\mytext{CapabilityId}_\fdginstance}$ are modes that inhere in a component and prevent the manifestation of (some of) that component's capabilities.
\end{definition}


We denote by $\faultid_\fdginstance : \fault_\fdginstance \rightarrow \textit{FaultId}_\fdginstance$ the bijective function that maps the elements in $\fault_\fdginstance$ to an individual in the ontology. Its inverse is denoted $\faultid_\fdginstance^{-1}$.

A fault prevents the manifestation of certain capabilities:
\begin{equation*}
    \forall {\,} \fdginstance : \fdg. \forall (\mytext{id}_{\faultinstance{f}}, C) \in \fault_\fdginstance. \forall \mytext{id}_f \in C. {\ } \mytext{id}_f \not \in \text{dom} \varrho_\fdginstance
\end{equation*}

A fault can only prevent the manifestation of capabilities that inhere in the bearer of that fault:
\begin{equation}\label{eq:faults_prevent_manifestation}
\forall {\,} \fdginstance : \fdg. \forall (\mytext{id}_{\faultinstance{f}}, C) \in \fault_\fdginstance. \forall f \in C. {\ } \beta(f) = \beta(\mytext{id}_{\faultinstance{f}})
\end{equation}

We may denote a fault with the striked-out symbol of any of the capabilities whose manifestation is prevented by that fault. I.e. $\faultinstance{f} \in \fault_\fdginstance$ denotes that the manifestation of capability $f \in \capability_\fdginstance$ is prevented.

The \textbf{Component} type is a \textit{mixin} that describe physical entities with one or more capabilities. We also distinguish the \textit{Consumer} and \textit{Producer} roles. The generalization set of these roles if not disjoint and not complete. These roles are respectively brought about by a material \textit{providedBy} and \textit{provides} relation to a channel.
The assertion that an individual is a component is given by the unary predicate $\textit{Component}: \textit{Individual} \rightarrow \{ \top, \bot \}$. For instance, in \autoref{fig:print_head}, we find that $\textit{Component}(\texttt{PSU}) = \top$ and $\textit{Component}(\texttt{Provide electricity}) = \bot$.

In this work, we do not make an ontological commitment to the parthood of components or the definition of systems. Fault tree synthesis and analysis does not necessarily depend on the mereology of systems, and to keep the ontology minimal we do not involve it.

We adopt an atomistic worldview\cite{varzi2021mereology} and define all components to be mereologically atomic, i.e. components do not have proper parts. A common refrain is \textit{``but what is a component to one, is a system to another''}. This is not necessarily incompatible with our perspective.
In fact, it is one of the key principles that allows for the modeling of an entire ``print head'' as a single component, constrained only by its interfaces to other other components.

In \autoref{fig:print_head}, the components are the \textit{PSU}, the \textit{ink reservoir}, and the \textit{print head}.

\begin{definition}[Components]\label{def:components}
    Components are tuples of the form
    \[
    \component_\fdginstance \subseteq \mytext{ComponentId}_\fdginstance \times \pow{\capability_\fdginstance} \times \pow{\function_\fdginstance} \times \pow{\fault_\fdginstance}
    \]
    where $\mytext{ComponentId}_\fdginstance$ is a set of component identifiers.
\end{definition}

\noindent Let $\alpha = (\mytext{id}_\alpha, \capability^\alpha_\fdginstance, \function^\alpha_\fdginstance, \fault^\alpha_\fdginstance)$ be a component. Then:
\begin{itemize}
    \item[] $\capability^\alpha_\fdginstance \in \pow{\capability_\fdginstance}$ is the set of capabilities that inhere in $\alpha$, and;
    \item[] $\function^\alpha_\fdginstance \in \pow{\function_\fdginstance}$ is the set of functions that manifest a capability of $\alpha$, and;
    \item[] $\fault^\alpha_\fdginstance \in \pow{\fault_\fdginstance}$ is the set of faults that inhere in $\alpha$.
\end{itemize}

\noindent We define the bijective function
\[
\compid_\fdginstance : \component_\fdginstance \rightarrow \mytext{ComponentId}_\fdginstance
\]
that maps components to an instance of $\textit{Component}$ in the situation $\fdginstance$, and that preserves the mixed identity principles provided by the specializations of the $\textit{Component}$ mixin across different situations.
Its inverse is denoted $\compid_\fdginstance^{-1}$.

We denote by $\capability^\alpha_\fdginstance$ the set of capabilities that inhere in the component $\alpha \in \component_\fdginstance$ in situation $\fdginstance$:
\begin{equation*}
    \capability^\alpha_\fdginstance = \left\{ f {\ } | {\ } f \in \capability_\fdginstance. {\ } \compid_\fdginstance(\alpha) = \beta(\capid_\fdginstance(f)) \right\}
\end{equation*}

We denote by $\function^\alpha_\fdginstance$ the set of functions that manifest a capability of the component $\alpha \in \component_\fdginstance$ in the situation $\fdginstance$:
\begin{equation*}
    \function^\alpha_\fdginstance = \left\{ \bar{f} {\ } | {\ } \bar{f} \in \function_\fdginstance. {\ } \compid_\fdginstance(\alpha) = \beta_\fdginstance(\varrho^{-1}_\fdginstance(\bar{f})) \right\}
\end{equation*}

We denote by $\fault^\alpha_\fdginstance \subseteq \fault_\fdginstance$ the set of faults that inhere in the component $\alpha \in \component_\fdginstance$ in the situation $\fdginstance$:
\begin{equation*}
    \fault^\alpha_\fdginstance = \left\{ \faultinstance{f} {\ } | {\ } \compid_\fdginstance(\alpha) = \beta_\fdginstance(\faultinstance{f}) \right\}
\end{equation*}

\textbf{Channels} are \textit{role mixins} that interface with a component. The channel role is grounded in the association to a resource type that the channel is said to accommodate. Channels can transport only the resource types they accommodate. The electrical wire that connects the \texttt{PSU} and the \texttt{Print head} in \autoref{fig:print_head} is an example of a channel.
The transport of each resource type is made up of two \textit{distinct} capabilities: a pure consumption capability to \textit{consume} that resource, and a pure production capability to \textit{produce} that resource.
Since we only discuss resource types in this work, we do not associate a quantity with this capacity. The predicate $\textit{Channel}: \textit{Individual} \rightarrow \{ \top, \bot \}$ asserts whether an individual is a channel.

\begin{definition}[Channels]\label{def:channel}
    Channels $\channel_\fdginstance \subseteq \mytext{ChannelId}_\fdginstance \times \pow{\resource} \times \pow{\capability_\fdginstance} \times \pow{\function_\fdginstance} \times \pow{\fault_\fdginstance}$ are tuples that relate an individual channel to the resource types it accommodates, where $\mytext{ChannelId}_\fdginstance$ is a set of channel identifiers.
\end{definition}
Let the function $\chanid_\fdginstance$ denote the bijective function that maps the elements in $\channel$ to an individual in the ontology.
\[
\chanid_\fdginstance : \channel_\fdginstance \rightarrow \mytext{ChannelId}_\fdginstance
\]
Its inverse is denoted $\chanid_\fdginstance^{-1}$.

We denote by $\capability^c_\fdginstance$ the set of capabilities that inhere in the channel $c \in \channel_\fdginstance$ in situation $\fdginstance$:
\begin{equation*}
    \capability^c_\fdginstance = \left\{ f {\ } | {\ } f \in \capability_\fdginstance. {\ } \chanid_\fdginstance(c) = \beta(\capid_\fdginstance(f)) \right\}
\end{equation*}

We denote by $\function^c_\fdginstance$ the set of functions that manifest a capability of the channel $c \in \channel_\fdginstance$ in situation $\fdginstance$:
\begin{equation*}
    \function^c_\fdginstance = \left\{ \bar{f} {\ } | \bar{f} \in \function_\fdginstance. {\ } \capid_\fdginstance(c) = \beta_\fdginstance(\varrho^{-1}_\fdginstance(\bar{f})) \right\}
\end{equation*}

We denote by $\fault^c_\fdginstance \subseteq \fault_\fdginstance$ the set of faults that inhere in the channel $c \in \channel_\fdginstance$ in the situation $\fdginstance$:
\begin{equation*}
    \fault^c_\fdginstance = \left\{ \faultinstance{f} {\ } | {\ } \chanid_\fdginstance(c) = \beta_\fdginstance(\faultinstance{f}) \right\}
\end{equation*}

Let the function $\channelcons_\fdginstance : \channel_\fdginstance \times \resource \rightarrow \capability_\fdginstance$ denote the unique input capability that allows a channel's consumption of a resource type. Let the function $\channelprod_\fdginstance : \channel_\fdginstance \times \resource \rightarrow \capability_\fdginstance$ denote the unique output capability that allows a channel's production of a resource type.

Let $c \in \channel_\fdginstance$ be a channel, and let $r \in \resource$ be a resource type that $c$ accommodates.
Let $h = \channelcons_\fdginstance(c, r)$. Then by definition, $h$ must be unique:
\begin{equation*}
     h \in \capability^c_\fdginstance \wedge \left( \neg \exists h' \in \capability^c_\fdginstance. {\ } h' = \channelcons_\fdginstance(c, r) \wedge h \neq h' \right)
\end{equation*}
Furthermore, $h$ must -- if it exists -- refer to the same capability across situations\sidenote{This is true but may be overkill? \future{In future work, we can change this to allow channels to dynamically switch internal paths :)}}:
\begin{multline*}
    \forall {\,} \fdginstance, \fdginstance' : \fdg. {\ } \left( h = \channelcons_\fdginstance(c, r) \right) \implies\\
    \left( \left( \exists h' \in \capability_{\fdginstance'}. {\ } \capid_\fdginstance(h) = \capid_{\fdginstance'}(h') \right) \implies h = \channelcons_{\fdginstance'}(c, r) \right)
\end{multline*}
The same definitions hold for $k = \channelprod_\fdginstance(c, r)$.

The production and consumption capabilities must actually consume and produce the correct respective resources. Let $c \in \channel_\fdginstance$ be a channel, and let $r \in \resource$ be a resource type that $c$ accommodates in the situation $\fdginstance$. Let $h = \channelcons_\fdginstance(c, r)$ and let $k = \channelprod_\fdginstance(c, r)$. Then:
\begin{equation*}
    \bigwedge_{\mathclap{x \in \left\{ h, k \right\}}}
    \mytext{consumes}_\fdginstance(x, r) \wedge \mytext{produces}_\fdginstance(x, r)
\end{equation*}

Furthermore, this resource type $r$ is unique.
I.e. there exists no $r' \in \resource$ s.t.:
\[
    r' \neq r \wedge \bigvee_{x \in \{ h, k \}} \bigl( \mytext{consumes}_\fdginstance(x, r') \vee \mytext{produces}_\fdginstance(x, r') \bigr)
\]

The consumption and production capabilities in a channel $c \in \channel_\fdginstance$ are distinct. Let $c \in \channel_\fdginstance$ be a channel, and let $r \in \resource$ be a resource type that $c$ accommodates in the situation $\fdginstance$. Then:
\begin{equation*}
    \capid_\fdginstance(\channelcons_\fdginstance(c, r)) \neq \capid_\fdginstance(\channelprod_\fdginstance(c, r))
\end{equation*}

\textbf{Interfaces} are \textit{relators} that describe an interaction between a component and a channel through some resource. They aggregate a component's and a channel's externally dependent modes that are founded by the foundation \textit{an}\footnote{The founding interface is not necessarily \textit{this} interface.} interface, resulting in the material relations \textit{providedBy} (InputInterfaces) and \textit{provides} (OutputInterfaces). I.e., the component $\alpha$ has the role \textit{consumer of resource $r$} iff there exists a channel $c$ such that an input interface $\textit{in}$ materializes a \textit{providedBy} relationship that mediates that resource $r$ through $c$. Similarly, the component $\alpha$ has the role \textit{producer of resource $r$} iff there exists a channel $c$ such that an output interface $\textit{out}$ materializes a \textit{provides} relationship that mediates that resource $r$ through $c$.

We can compose functions such that the output of one component's function serves as the input of another component's function. Such patterns describe components exchanging a resource through a channel. We sometimes call this pattern a \textit{resource exchange}. 


In \autoref{fig:print_head}, there is an exchange of \textit{electricity} between the PSU's \textit{provide electricity} function and the print head's \textit{deposit ink} function. There is also a resource exchange between the ink reservoir's \textit{provide ink} and the print head's \textit{deposit ink}.

\begin{definition}[Output interfaces]\label{def:output_interface}
    An output interface $\ointerface_\fdginstance \subseteq \mytext{ComponentId}_\fdginstance \times \mytext{ChannelId}_\fdginstance \times \resource$ is the relator that the \textit{provides} association between a producer component and a channel is derived from.
\end{definition}
    
An output interface relates a component, a channel and the resource type that the component provides to the channel.

Let $\quaindcomp^o_\fdginstance : \ointerface_\fdginstance \rightarrow \mathcal{P}(\mytext{CapabilityId}_\fdginstance)$ be the function that returns the set of capabilities that make up the component's qua-individual \textit{component-$\alpha$-qua-provider-of-$o$}, such that
\begin{multline*}
    \forall {\,} \fdginstance : \fdg. \forall o \in \ointerface_\fdginstance. \forall \mytext{id}_f \in \quaindcomp^o_\fdginstance(o). {\ } \beta(\mytext{id}_f) = \mytext{proj}_1(o)
\end{multline*}

Let $\quaindchannel^o_\fdginstance : \ointerface_\fdginstance \rightarrow \mytext{CapabilityId}_\fdginstance$ be the function that returns the unique capability that makes up the channels's qua-individual \textit{channel-$c$-qua-consumer-of-$o$}, such that
\begin{multline*}
    \forall {\,} \fdginstance : \fdg. \forall o \in \ointerface_\fdginstance.\\
    \quaindchannel^o_\fdginstance(o) = \capid_\fdginstance(\channelcons(\chanid^{-1}(\mytext{proj}_2(o)), \mytext{proj}_3(o)))
\end{multline*}


The identity of an output interface is derived from the sum of the qua-entities it aggregates, given by
\[
\interfaceid^o_\fdginstance : \ointerface_\fdginstance \rightarrow \mytext{InterfaceId}_\fdginstance
\]
Its inverse is given by $\interfaceid^{o, -1}_\fdginstance$.

Lastly, the output interface relator is existentially dependent on the manifestations of the capabilities that make up its qua-entities:
\begin{multline*}
    \forall {\,} \fdginstance : \fdg. \forall o \in \ointerface_\fdginstance. \forall \mytext{id}_f \in \left( \quaindcomp^o_\fdginstance(o) \cup \left\{ \quaindchannel^o_\fdginstance(o) \right\} \right).\\
    \exists \bar{f} \in \function_\fdginstance. {\ } \varrho_\fdginstance(\mytext{id}_f) = \bar{f}
\end{multline*}
    

\begin{definition}[Input interfaces]\label{def:input_interface}
    An input interface $\iinterface_\fdginstance \subseteq \mytext{ChannelId}_\fdginstance \times \mytext{ComponentId}_\fdginstance \times \resource$ is the relator that the \textit{providedBy} association between a channel and a consumer component is derived from.
\end{definition}

An input interface relates a channel, a component and the resource type that the channel provides to the component.

Let $\quaindcomp^i_\fdginstance : \iinterface_\fdginstance \rightarrow \pow{\mytext{CapabilityId}_\fdginstance}$ be the function that returns the set of capabilities that make up the component's qua-individual \textit{component-$\alpha$-qua-provided-by-$i$}, such that
\begin{multline*}
    \forall {\,} \fdginstance : \fdg. \forall i \in \iinterface_\fdginstance. \forall \mytext{id}_f \in \quaindcomp^i_\fdginstance(i). {\ } \beta(\mytext{id}_f) = \mytext{proj}_2(i)
\end{multline*}

Let $\quaindchannel^i_\fdginstance : \iinterface_\fdginstance \rightarrow \mytext{CapabilityId}_\fdginstance$ be the function that returns the unique capability that makes up the channel's qua-individual \textit{channel-$c$-qua-provider-of-$i$}, such that
\begin{multline*}
    \forall {\,} \fdginstance : \fdg. \forall i \in \iinterface_\fdginstance.\\
    \quaindchannel^i_\fdginstance(i) = \capid_\fdginstance(\channelprod(\chanid^{-1}(\mytext{proj}_1(i)), \mytext{proj}_3(i)))
\end{multline*}


The identity of an input interface is derived from the sum of the qua-entities it aggregates, given by
\[
\interfaceid^i_\fdginstance : \iinterface_\fdginstance \rightarrow \mytext{InterfaceId}_\fdginstance
\]
Its inverse is given by $\interfaceid^{i, -1}_\fdginstance$.

Lastly, like with output interfaces, the input interface relator is existentially dependent on the manifestations of the capabilities that make up its qua-entities:
\begin{multline*}
    \forall {\,} \fdginstance : \fdg. \forall i \in \iinterface_\fdginstance. \forall \mytext{id}_f \in \left( \quaindcomp^i_\fdginstance(i) \cup \left\{ \quaindchannel^i_\fdginstance(i) \right\} \right).\\
    \exists \bar{f} \in \function_\fdginstance. {\ } \varrho_\fdginstance(\mytext{id}_f) = \bar{f}
\end{multline*}

\mypar{Expectations}
Expectations describe what capabilities \textit{should} be manifested as functions. In this work, we do not explore the justification of \textit{why} a particular capability's manifestation should be expected, nor do we explicitly describe \textit{who} has such expectations.
E.g. we may expect the capability \textit{to provide electricity} to be manifested. If it is not -- i.e. if this expectation is not met -- we call this an \textit{error}.

\begin{definition}[Expectations]\label{def:expectation}
    A set of expectations $\expectation_\fdginstance \subseteq \mytext{CapabilityId}_\fdginstance$ is a subset of capability identifiers.
\end{definition}


We say a situation \textit{meets} an expectation if and only if the expected capability is manifested by a function in that situation:
\begin{equation*}
\forall {\,} \fdginstance : \fdg. \forall e \in \expectation_\fdginstance. {\ } \mytext{meets}(\fdginstance, e) \iff e \in \text{dom} \varrho_\fdginstance
\end{equation*}


Note that, given \autoref{def:expectation}, {\fdgnamearticle} {\fdgname} can still meet all expectations when the system ``does more than expected''.

\mypar{On the sets $\capability_\fdginstance$, $\function_\fdginstance$ and $\fault_\fdginstance$}

Capabilities inhere in objects. In this work, the only objects we distinguish are components and channels. Hence, all capabilities must inhere in either a component or a channel. The union of component capabilities and channel capabilities therefore equals the set of capabilities $\capability_\fdginstance$:
\begin{equation}\label{eq:capabilities_partition}
    \forall {\,} \fdginstance : \fdg. {\ } \capability_\fdginstance = \left( \bigcup_{\alpha \in \component_\fdginstance} \capability^\alpha_\fdginstance \right) \cup \left( \bigcup_{c \in \channel_\fdginstance} \capability^c_\fdginstance \right)
\end{equation}

Functions are manifestations of exactly one capability, which in turn inhere in exactly one object. Like with capabilities, the union of component functions and channel functions therefore equals the set of functions $\function_\fdginstance$:
\begin{equation}\label{eq:functions_partition}
    \forall {\,} \fdginstance : \fdg. {\ } \function_\fdginstance = \left( \bigcup_{\alpha \in \component_\fdginstance} \function^\alpha_\fdginstance \right) \cup \left( \bigcup_{c \in \channel_\fdginstance} \function^c_\fdginstance \right)
\end{equation}

Faults inhere in exactly one object. Like with capabilities, the union of component faults and channel faults therefore equals the set of faults $\fault_\fdginstance$:
\begin{equation}\label{eq:faults_partition}
    \forall {\,} \fdginstance : \fdg. {\ } \fault_\fdginstance = \left( \bigcup_{\alpha \in \component_\fdginstance} \fault^\alpha_\fdginstance \right) \cup \left( \bigcup_{c \in \channel_\fdginstance} \fault^c_\fdginstance \right)
\end{equation}

Recall that {\fdgname}s permit components that play the role of a channel. Therefore, the families of sets $\bigcup_{\alpha \in \component_\fdginstance} \capability^\alpha_\fdginstance$, $\bigcup_{c \in \channel_\fdginstance} \capability^c_\fdginstance$, and $\left\{ \bigcup_{\alpha \in \component_\fdginstance} \function^\alpha_\fdginstance, \bigcup_{c \in \channel_\fdginstance} \function^c_\fdginstance \right\}$, and $\left\{ \bigcup_{\alpha \in \component_\fdginstance} \fault^\alpha_\fdginstance, \bigcup_{c \in \channel_\fdginstance} \fault^c_\fdginstance \right\}$ are not \textit{necessarily} partitions of the sets $\capability_\fdginstance$, $\function_\fdginstance$ and $\fault_\fdginstance$ respectively.



\section{{\fdgnamefull}s}\label{sec:fdg_formalization}
In this section we combine the individual concepts of CPS architecture from the previous section into one large model: the {\fdgname}. We use this formalization in the next section, to show how these {\fdgname}s can dynamically change over time.

In this section we also present a graphical notation, turning {\fdgname}s into readable diagrams which will be useful for examples throughout the paper. \autoref{tab:glossary_fdg_endurants} provides a glossary of endurant {\fdgname} elements and functions.


\begin{table*}[]
    \centering
    \begin{tabular}{cl}
        \toprule
        \textbf{Notation} & \textbf{Description} \\
        \midrule
        $\fdginstance : \fdg$ & The factual situation $\fdginstance$ \\
        \midrule
        $\component_\fdginstance , \channel_\fdginstance, \dots$ & The set of components, channels, etc. in the situation $\fdginstance$ \\
        \midrule
        $\chanid_\fdginstance$ & The function that returns the ontological identity of a given {\fdgname} element \\
        \midrule
        $\varrho_\fdginstance (f)$ & The function that returns the Function that manifests a given capability manifests in the situation $\fdginstance$ \\
        \midrule
        $f, \bar{f}, \faultinstance{f}$ & A capability, a Function manifesting that capability, and a fault preventing the manifestation of that capability resp.\\
        \midrule
        $\channelcons_\fdginstance(c, r)$ & The function that returns the unique input capability that allows the channel $c$'s consumption of the resource type $r$ \\
        \midrule
        $\channelprod_\fdginstance(c, r)$ & The function that returns the unique output capability that allows the channel $c$'s production of the resource type $r$ \\
        \midrule
        $\quaindcomp^o_\fdginstance(o), \quaindcomp^i_\fdginstance(i)$ & The functions that return, for a given interface, the set of capabilities that make up the component's qua-individual \\
        \midrule
        $\quaindchannel^o_\fdginstance(o), \quaindchannel^i_\fdginstance(i)$ & The functions that return, for a given interface, the unique capability that makes up the channels's qua-individual \\
        \bottomrule
    \end{tabular}
    \caption{A glossary of endurant {\fdgname} elements and functions.}
    \label{tab:glossary_fdg_endurants}
\end{table*}

\subsection{Formal definition}
A {\fdgname} is a tuple composed of \textbf{capabilities} ($\capability$), \textbf{functions} ($\function$), \textbf{faults} ($\fault$), \textbf{components} ($\component$), \textbf{channels} ($\channel$), \textbf{output interfaces} ($\ointerface$), \textbf{input interfaces} ($\iinterface$) and \textbf{expectations} ($\expectation$). Such a tuple represents the state of a system, and is typically denoted $\gamma$. The type of {\fdgnamearticle} {\fdgname} is denoted $\fdg$, and we write $\gamma \in \fdg$ for instances of $\fdg$.

We base our notion of time points on that introduced in UFO-B\cite{ufo_b}. Simply put, the set $\textit{TimePoint}$ and the $\textit{precedes}$ relation make up a strict total order.
We may assign a time point to the state $\gamma$. We then write $\fdginstance : \fdg$.

Furthermore, we may denote by a subscript the {\fdgname} to which a set belongs. E.g. we mean by $\component_{\fdginstance}$ the set of components that belong to some {\fdgname} $\fdginstance : \fdg$.

Additionally, in the places where it is obvious, we do not write the particular situation in which predicates are said to hold. For instance, when we write
\[
\forall \alpha \in \component_\fdginstance. {\ } \textit{Component}( \compid(\alpha) )
\]
we mean that $\alpha$ is \textit{present at} the time point at which the situation $\fdginstance$ obtains, and that the predicate $\textit{Component}( \compid_\fdginstance(\alpha)$ holds in that situation. We describe time points and situations in more detail in \autoref{sec:dynamics_and_events}.

The formal definition of {\fdgname}s is grounded in UFO-A and UFO-B predicates and functions. For reference, the relevant portions of UFO-A are described in \autoref{tab:ufo_predicates}. \autoref{tab:ufo_b_axioms} lists the relevant portion of UFO-B.

\begin{table*}[]
\centering
\begin{tabular}{lll}
\toprule
\textbf{ID}                        & \textbf{Definition}                                     & \textbf{Description}                                                   \\
\midrule
\textbf{10}\cite{ufo_a} & $\beta : \mytext{IntrinsicMoment} \rightarrow \mytext{Individual}$ & The function that returns the bearer of a particular intrinsic moment. \\
\midrule
\textbf{a102}\cite{ufo_a_axioms} & $\mytext{manifests}(x, y) \implies \mytext{Perdurant}(x) \wedge \mytext{Endurant}(y) $ & Relates an event to the particular moment it is a manifestation of. \\
\bottomrule
\end{tabular}
\caption{The relevant portion of UFO, taken from \cite{ufo_a} and \cite{ufo_a_axioms}.}\label{tab:ufo_predicates}
\end{table*}


\begin{definition}\label{def:fdg}

Consider the set $\fdg = \pow{\capability} \times \pow{\function} \times \pow{\fault} \times \pow{\component} \times \pow{\channel} \times \pow{\ointerface} \times \pow{\iinterface} \times \pow{\expectation}$. A {\fdgname} $\gamma = (\capability_\fdginstance, \function_\fdginstance, \fault_\fdginstance, \component_\fdginstance, \channel_\fdginstance, \ointerface_\fdginstance, \iinterface_\fdginstance, \expectation_\fdginstance)$ is a tuple that is an element of the set $\fdg$, denoted $\gamma \in \fdg$.
\end{definition}

An factual {\fdgname} $\fdginstance$ (i.e. one that obtains at some time point) is a \textit{situation} in the sense of UFO-B\cite{ufo_b} and denoted $\fdginstance : \fdg$. In UFO-B, situations are akin to \textit{states of affairs} that describe a fragment of the world at a particular point in time. However, two situations are distinct if the time point at which they obtain are different -- even if the state of affairs they describe are identical.

The multiplicity function $\mu_{\fdg'} : \fdg \rightarrow \mathbb{N}$ defines the multiset $\fdg' = (\fdg, \mu_{\fdg'})$. Contrary to common sets, a multiset allows for more than one element of the same value. This is an important property, since in UFO two situations are distinct if they describe a different time points\cite{ufo_b}.

Then let $\textit{Occ}(\fdg') \subset \fdg \times \mathbb{N}$ be the set that relates {\fdgname}s to the multiplicity of their occurrences:
\[
\mytext{Occ}(\fdg') = \{ (\gamma, i) {\ } | {\ } \gamma \in \fdg. {\ } 0 \leq i \leq \mu_{\fdg'}(\gamma) \}
\]

We introduce the notation $\fdginstance$ to denote an occurrence $(\gamma, i) \in \mytext{Occ}(\fdg')$.
The notation $\fdginstance : \fdg$ is such that the multiplicity of the occurrence is at least one:
\[
\forall {\,} (\gamma, i) \in \fdg. {\ } \fdginstance : \fdg \iff \mu_{\fdg'}(\fdginstance) \geq 1
\]

Let the function $\tau : \textit{Occ}(\fdg') \rightarrow \textit{TimePoint}$ be the total function that relates occurrences of {\fdgnamearticle} {\fdgname} to a time point. By the definitions of $\textit{Occ}(\fdg')$ and $\tau$, all {\fdgname}s such that $\fdginstance : \fdg$ have obtained in a time point:
\begin{equation*}
\forall \fdginstance \in \fdg. {\ } \fdginstance : \fdg \iff \exists t \in \mytext{TimePoint}. {\ } \tau(\fdginstance) = t
\end{equation*}
In UFO, such situations are termed \textit{facts}.

Two {\fdgname}s $\fdginstance = (\gamma, i) \in \textit{Occ}(\fdg')$ and $\fdginstance' = (\gamma', j) \in \textit{Occ}(\fdg')$ are structurally equivalent when they describe the same state of affairs (regardless of their time point):
\begin{equation*}
    \fdginstance \equiv \fdginstance' \iff \gamma = \gamma'
\end{equation*}

Finally, we define that two {\fdgname}s $\fdginstance = (\gamma, i) \in \textit{Occ}(\fdg')$ and $\fdginstance' = (\gamma', j) \in \textit{Occ}(\fdg')$ are distinct if they obtain at a different time point:
\begin{equation*}
    \tau(\fdginstance) \neq \tau(\fdginstance') \implies i \neq j
\end{equation*}

In this way, the notation $\fdginstance : \fdg$ mirrors the UFO notion of factual situations, whereas $\fdginstance \in \fdg$ denotes the set of all world \textit{possibilia} (i.e. all the states that a system \textit{could} be in).

\subsection{On {\fdgname}s as objects}
{\fdgname}s describe a collection of interacting elements. To argue that these interacting elements can be the focus of a situation, we must argue that this collection itself has object-like properties. We call the elements that {\fdgnamearticle} {\fdgname} comprises \textit{endogenous} elements.
And we call the sum of such endogenous elements \textit{systems}. We do not make any ontological commitments to the nature of systems -- such as their mereology -- beyond the stipulate that a system is the object comprising all components, channels and interactions in {\fdgnamearticle} {\fdgname} $\gamma$, and that this object is precisely the focus of the situation $\fdginstance$.
We must now be precise: {\fdgnamearticle} {\fdgname} $\gamma$ is an \textbf{object}. When we discuss the dynamics and evolution of such objects, we describe their occurrences (e.g. $\fdginstance$) as \textbf{situations}.



\subsection{Formal graphical notation of {\fdgname}s}
A {\fdgname} can be interpreted as a graph by having its set of capabilities as vertices, and its interfaces as solid edges. The label of each edge is the resource type that is exchanged. We draw a rectangle around capabilities that have the same bearer (i.e. components). A dashed edge denotes the flow of a resource type between channel input and output capabilities. We overset a capability label with a bar (e.g. $\bar{f}$) to indicate that it is currently manifested as a function.
The running example shown in \autoref{fig:print_head} is more formally denoted in \autoref{fig:print_head_formal}.\sidenote{TODO: include the capability to self-clean!}

\begin{figure}[ht!]
\centering

\includegraphics[]{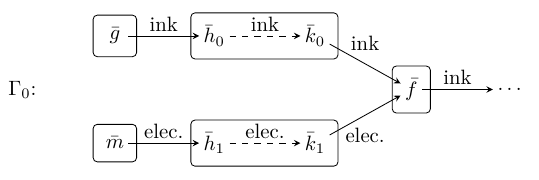}

\caption{A formal {\fdgname} diagram showing the running example (derived from \autoref{fig:print_head}).}
    \label{fig:print_head_formal}
\end{figure}

\subsection{Further model constraints}
To further eliminate unintended model instantiations, we provide additional integrity constraints.

\mypar{Individuation}
The elements of {\fdgname}s must have a surjective relation to the instances that the ontology in \autoref{sec:ontology} permits for concepts of the same name;
\begin{align*}
    \forall r \in \resource&. {\ } \mytext{ResourceType}(r)\\
    \forall f \in \capability_\fdginstance&. {\ } \mytext{Capability}(\capid_\fdginstance(f))\\
    \forall \bar{f} \in \function_\fdginstance&. {\ } \mytext{Function}(\bar{f}) \\
    \forall \hbox{\sout{$f$}} \in \fault_\fdginstance&. {\ } \mytext{Fault}(\faultid_\fdginstance(\hbox{\sout{$f$}}))\\
    \forall \alpha \in \component_\fdginstance&. {\ } \mytext{Component}( \compid_\fdginstance(\alpha) )\\
    \forall c \in \channel_\fdginstance&. {\ } \mytext{Channel}(\chanid_\fdginstance(c))\\
    \forall o \in \ointerface_\fdginstance&. {\ } \mytext{OutputInterface}(\interfaceid^o_\fdginstance(o))\\
    \forall i \in \iinterface_\fdginstance&. {\ } \mytext{InputInterface}(\interfaceid^i_\fdginstance(i))\\
    \forall e \in \expectation_\fdginstance&. {\ } \mytext{Capability}(\capid_\fdginstance(e))
\end{align*}

\mypar{Integrity of resource types}
The second and third element of a capability tuple denote the resource types that the capability, when manifested, consumes and produces respectively. Let $\fdginstance : \fdg$ and $(\mytext{id}_f, R_{in}, R_{out}) \in \capability_\fdginstance$:
\begin{multline}\label{eq:capability_tuple_constraints}
    \left( \forall r^\downarrow \in R_{in}. {\ } \mytext{consumes}(\beta(\mytext{id}_f), r^\downarrow) \right) \wedge\\
    \left( \forall r^\uparrow \in R_{out}. {\ } \mytext{produces}(\beta(\mytext{id}_f), r^\uparrow) \right)
\end{multline}

The second element of a channel tuple denotes the resource types that channel accommodates. Let $\fdginstance : \fdg$:
\[
\forall (\mytext{id}_c, R, \dots) \in \channel_\fdginstance. {\ } \forall r \in R. {\ } \mytext{accommodates}(\mytext{id}_c, r)
\]

The component capabilities mediated by interfaces are compatible with that interface w.r.t. the resource type they produce or consume (for output and input interfaces respectively). Let $\fdginstance : \fdg$:
\begin{equation*}
    \forall o \in \ointerface. \forall \mytext{id}_f \in \quaindcomp^o_\fdginstance(o). {\ } \mytext{proj}_3(o) \in \mytext{proj}_3(\capid^{-1}_\fdginstance(\mytext{id}_f))
\end{equation*}
and
\begin{equation*}
    \forall i \in \iinterface. \forall \mytext{id}_f \in \quaindcomp^i_\fdginstance(i). {\ } \mytext{proj}_3(i) \in \mytext{proj}_2(\capid^{-1}_\fdginstance(\mytext{id}_f))
\end{equation*}

\mypar{On the special case of transport capabilities}
The consumption capability of a channel cannot provide a channel, and the production capability of a channel cannot be provided by a channel. Let $\fdginstance : \fdg$ and $c \in \channel_\fdginstance$:
\begin{multline*}
    \bigl( \exists r \in \resource. {\ } \mytext{accommodates}(c, r) \bigr) \wedge \forall f \in \capability^c_\fdginstance. \\
    \left( f = \channelcons_\fdginstance(c, r) \implies \neg \exists o \in \ointerface. {\ } \capid_\fdginstance(f) = \quaindchannel^o_\fdginstance(o) \right) \wedge\\
    \left( f = \channelprod_\fdginstance(c, r) \implies \neg \exists i \in \iinterface. {\ } \capid_\fdginstance(f) = \quaindchannel^i_\fdginstance(i) \right)
\end{multline*}

\mypar{The bearer of transport capabilities}
The bearer of the consumption and production capabilities of a channel $c$, is $c$. Let $\fdginstance : \fdg$ be a situation, an $c \in \channel_\fdginstance$ be a channel. Let $r \in \resource$ be a resource that $c$ accommodates. Then:
\[
\beta(\channelcons_\fdginstance(c, r)) = \beta(\channelprod_\fdginstance(c, r)) = c
\]

\mypar{Manifestation of capabilities}
Let $\fdginstance : \fdg$. The following equation explicitly connects all members of $\function_\fdginstance$ to a unique member of $\capability_\fdginstance$:
\begin{equation*}
    \forall \bar{f} \in \function_\fdginstance. \exists !f \in \capability_\fdginstance.\\
    \bar{f} = \varrho_\fdginstance(\capid_\fdginstance(f)) \wedge \ufo{manifests}(\bar{f}, \capid_\fdginstance(f))
\end{equation*}





\section{Dynamics and system evolution}\label{sec:dynamics_and_events}
In this section, we extend the representation of the system \textit{at some point in time} (formalized in the previous section) with a representation of a system's dynamic evolution \textit{throughout} time.

To encode the dynamic aspects of a system, we can evolve {\fdgname}s through the UFO-B notion of events\cite{ufo_b}. We typically denote by $\fdginstance'$ the situation that is derived from the situation $\fdginstance$ by the occurrence (or \textit{application}) of some event. For reference, we also list the relevant portions of UFO-B in \autoref{tab:ufo_b_axioms}\sidenote{Update this with UFO-B*\cite{ufo_ab} axioms.}.


\begin{table*}[t]
\centering
\begin{tabular}{@{}ll@{}}
\textbf{Axiom} & \textbf{FOL formula} \\
\toprule
S1             & $\forall s : \mytext{Situation}, e : \mytext{Event}. {\ } \mytext{triggers}(s, e) \implies \mytext{obtainsIn}(s, \mytext{begin-point}(e))$ \\
\midrule
S2             & $\forall s : \mytext{Situation}, e : \mytext{Event}. {\ } \mytext{brings-about}(e, s) \implies \mytext{obtainsIn}(s, \mytext{end-point}(e))$ \\
\midrule
S3             & $\forall e : \mytext{Event}. \exists !s : \mytext{Situation}. {\ } \mytext{triggers}(s, e)$ \\
\midrule
S4             & $\forall e : \mytext{Event}. \exists !s : \mytext{Situation}. {\ } \mytext{brings-about}(e, s)$ \\
\midrule
S5             & $\forall s : \mytext{Situation}. {\ } \mytext{fact}(s) \iff \exists t : \mytext{TimePoint}. {\ } \mytext{obtainsIn}(s, t)$ \\
\midrule
S6             & $\forall e, e' : \mytext{Event}. {\ } \mytext{directly-causes}(e, e') \iff \exists s : \mytext{Situation}. {\ } \mytext{brings-about}(e, s) \wedge \mytext{triggers}(s, e')$ \\
\midrule
S7             & $\forall e, e'' : \mytext{Event}. {\ } \mytext{causes}(e, e'') \iff \mytext{directly-causes}(e, e'') \vee \left( \exists e' : \mytext{Event}. {\ } \mytext{causes}(e, e') \wedge \mytext{causes}(e', e'') \right)$ \\
\midrule
T1             & $\forall t : \mytext{TimePoint}. {\ } \neg \mytext{precedes}(t, t)$ \\
\midrule
T2             & $\forall t, t' : \mytext{TimePoint}. {\ } \mytext{precedes}(t, t') \implies \neg \mytext{precedes}(t', t)$ \\
\midrule
T3             & $\forall t, t', t'' : \mytext{TimePoint}. {\ } \mytext{precedes}(t, t') \wedge \mytext{precedes}(t', t'') \implies \mytext{precedes}(t, t'')$ \\
\midrule
t4\cite{ufo_ab}& $\forall t, t': \mytext{TimePoint}. {\ } \mytext{immediateSuccessorOf}(t', t) \iff \mytext{precedes}(t, t') \wedge \neg \exists t'' : \mytext{TimePoint}. \left( \mytext{precedes}(t'', t') \wedge  \mytext{precedes}(t, t'') \right)$ \\          
\midrule
T5             & $\forall e : \mytext{Event}. \exists !t : \mytext{TimePoint}. \exists !t' : \mytext{TimePoint}. (t = \mytext{begin-point}(e)) \wedge (t' = \mytext{end-point}(e))$ \\
\midrule
P1             & $\forall e : \mytext{AtomicEvent}. \exists !o : \mytext{Object}. {\ } \mytext{dependsOn}(e, o)$ \\
\midrule
P2             & $\forall e : \mytext{AtomicEvent}, o : \mytext{Object}. \mytext{excDepends}(e, o) \iff \mytext{dependsOn}(e, o)$ \\
\midrule
P4             & $\forall e : \mytext{Event}. \mytext{Participation}(e) \iff \exists !o : \mytext{Object}. {\ } \mytext{excDepends}(e, o)$ \\
\midrule
P5             & $\forall o : \mytext{Object}, p : \mytext{Participation}. \mytext{participationOf}(p, o) \iff \mytext{excDepends}(p, o)$ \\
\midrule
a102\cite{ufo_a_axioms} & $\mytext{manifests}(x, y) \implies \mytext{Perdurant}(x) \wedge \mytext{Endurant}(y) $\\
\midrule
D1             & $\forall d : \mytext{Disposition}. \exists !o : \mytext{Object}. {\ } \mytext{inheresIn}(d, o)$ \\
\midrule
D2             & $\forall e : \mytext{AtomicEvent}. \exists !d : \mytext{Disposition}. {\ } \mytext{manifests}(e, d)$ \\
\midrule
D3             & $\forall s : \mytext{Situation}, e : \mytext{AtomicEvent}. {\ } \mytext{triggers}(s, e) \iff \exists d : \mytext{Disposition}. {\ } \mytext{activates}(s, d) \wedge \mytext{manifests}(e ,d)$ \\
\midrule
D4             & $\forall d : \mytext{Disposition}, e : \mytext{AtomicEvent}, o : \mytext{Object}. \mytext{manifests}(e ,d) \wedge \mytext{inheresIn}(d, o) \implies \mytext{dependsOn}(e, o)$ \\
\bottomrule
\end{tabular}
\caption{A subset of axioms of UFO-B, taken from \cite{ufo_b}.}\label{tab:ufo_b_axioms}
\end{table*}

Events transform one situation into another. We call the situation that an event begins in its \textit{antecedent}, and the situation that is brought about by the event its \textit{consequent}.

An atomic event is an event that is not made up of other events (in contrast to \textit{complex} events)\cite{ufo_b}. We distinguish two types of atomic events: triggered events and spontaneous events. 

\newcommand{\event}{e}
\newcommand{\spevent}{\hat{\event}}
\begin{definition}[Triggered events]\label{def:triggered_events}
    We may denote a triggered event $\event$, the situation that triggers it (S3, \autoref{tab:ufo_b_axioms}), and the situation it brings about (S4, \autoref{tab:ufo_b_axioms}) as $\fdginstance \xrightarrow{\event} \fdginstance'$:
    \begin{multline}\label{eq:triggered_event_definition}
    \forall {\,} \fdginstance, \fdginstance' : \fdg. \forall \event : \mytext{Event}.\\
    \fdginstance \xrightarrow{\event} \fdginstance' \iff \ufo{triggers}(\fdginstance, \event) \wedge \ufo{brings-about}(\event, \fdginstance')
    \end{multline}
\end{definition}

\begin{definition}[Spontaneous events]\label{def:spontaneous_events}
    We may denote a spontaneous event $\spevent$, the situation that precedes it, and the situation it brings about (S4, \autoref{tab:ufo_b_axioms}) as $\fdginstance \overset{\spevent}\looparrowright \fdginstance'$:
    \begin{multline}\label{eq:spontaneous_event_definition}
    \forall {\,} \fdginstance, \fdginstance' : \fdg. \forall \event : \mytext{Event}. {\ } \fdginstance \overset{\spevent}\looparrowright \fdginstance' \iff\\
    \ufo{obtainsIn}(\fdginstance, \ufo{begin-point}(\spevent)) \wedge\\
    \neg \ufo{triggers}(\fdginstance, \spevent) \wedge \ufo{brings-about}(\spevent, \fdginstance')
    \end{multline}
\end{definition}

By S3, for each event there must exist a unique situation that triggers it. However, for spontaneous events we cannot describe this situation since, by \autoref{def:spontaneous_events}, we do not know it (if we did, the event would be a \textit{triggered} event).

Furthermore:
\begin{itemize}
    \item An antecedent situation that triggers an event obtains at the begin point of that event (S1, \autoref{tab:ufo_b_axioms}), and;
    \item The consequent situation that an event brings about obtains at the end point of that event (S2, \autoref{tab:ufo_b_axioms}).
\end{itemize}

\newcommand{\anyevent}{\dot{e}}
\newcommand{\generalevent}{\curvearrowright}
In certain cases, we need only consider the consequent situation that an event brings about. If we do not need to distinguish between triggered events and spontaneous events, we may write $\fdginstance \overset{\anyevent}\generalevent \fdginstance'$ in lieu of either $\fdginstance \xrightarrow{\event} \fdginstance'$ or $\fdginstance \overset{\spevent}\looparrowright \fdginstance'$:
\begin{equation}\label{eq:event_shorthand_definition}
    \forall {\,} \fdginstance, \fdginstance' : \fdg. \forall \event : \mytext{Event}. {\ } \fdginstance \overset{\event}\generalevent \fdginstance' \iff \fdginstance \xrightarrow{\event} \fdginstance' \xor \fdginstance \overset{\event}\looparrowright \fdginstance'
\end{equation}

where $\xor$ is the \texttt{XOR} connective.

\begin{definition}[Temporal order of situations]\label{def:temporal_order_situations}
    Let the strict partial order $\prec {\ } \subseteq {\ } \fdg \times \fdg$ denote the temporal order of factual situations that obtain at some time point
    (S5, T1, T2, T3, \autoref{tab:ufo_b_axioms}):
    \begin{equation}\label{eq:factual_temporal_ordering}
    \forall {\,} \fdginstance, \fdginstance' : \fdg. {\ } \fdginstance \prec \fdginstance' \iff \ufo{precedes}(\tau(\fdginstance), \tau(\fdginstance'))
    \end{equation}
\end{definition}

By axiom T5 (\autoref{tab:ufo_b_axioms}), events cannot have zero duration in UFO-B. Therefore, the time point of an event's antecedent always precedes the time point of its consequent. By \cref{eq:factual_temporal_ordering}, we obtain:
\begin{equation*}
\forall {\,} \fdginstance, \fdginstance' : \fdg, \anyevent : \mytext{Event}. {\ }
\fdginstance \overset{\anyevent}\generalevent \fdginstance' \implies \fdginstance \prec \fdginstance'
\end{equation*}

We also make use of a weaker relation $\preceq {\ } : {\ } \fdg \times \fdg$ that is entailed by $\prec$:
\begin{equation*}
    \forall {\,} \fdginstance, \fdginstance' : \fdg. {\ } \fdginstance \preceq \fdginstance' \iff \fdginstance \prec \fdginstance' \vee \tau(\fdginstance) = \tau(\fdginstance')
\end{equation*}




\begin{definition}[Channel depletion]\label{def:channel_transportation}
    A channel $c$ can only output a resource of type $r$ if that resource is present in the channel. If the resource is not (continuously) replenished by the manifestation of the channel's input capability, then after some non-zero time $\consdv$, its store of $r$ is depleted and the channel's output capability cannot be manifested. We denote this depletion event $\Delta \langle c \rangle$, which has a constant duration $\consdv$ for all channels. The begin point of $\Delta \langle c \rangle$ is the end point of the manifestation of $c$'s input of $r$, and the end point of $\Delta \langle c \rangle$ is the first time point at which $r$ is depleted.
\end{definition}


Let $c \in \channel_\fdginstance$ be a channel. Then:
\begin{multline}\label{eq:channel_transportation}
    \forall {\,} \fdginstance, \fdginstance' : \fdg. {\ } \hat\fdginstance \xrightarrow{\Delta \langle c \rangle} \fdginstance' \iff\\
    \channelcons_{\hat\fdginstance}(c) \not \in \text{dom} \varrho_{\hat\fdginstance} \wedge \channelprod_{\fdginstance'}(c) \not \in \text{dom} \varrho_{\fdginstance'}
\end{multline}

Where by $\hat\fdginstance$ in \cref{eq:channel_transportation} we mean the earliest situation in which $\channelcons_\fdginstance(c)$ is not manifested (see \autoref{def:earliest_situation}).


\begin{definition}[Earliest antecedent such that $\varphi$]\label{def:earliest_situation}
    Let $\varphi$ be a predicate. Let $\fdginstance \xrightarrow{e} \fdginstance'$ be some event such that $\varphi$ holds in $\fdginstance$ (denoted $\fdginstance \models \varphi$). Let $t$ be the time point such that $\mytext{obtainsIn}(\fdginstance, t)$. Then $\fdginstance$ is the earliest situation in which $\varphi$ holds iff there does not exist a situation $\fdginstance''$ that obtains at a time point $t''$, such that $t''$ immediately succeeded by $t$ and $\varphi$ holds in $\fdginstance''$. I.e.:
    \begin{multline*}
        \neg \exists t'' : \textit{TimePoint}. {\ } \exists \fdginstance'' : \fdg. {\ } \mytext{obtainsIn}(\fdginstance'', t'') \wedge \\ \ufo{\mytext{immediateSuccessorOf}}(t, t'') \wedge \fdginstance'' \models \varphi
    \end{multline*}
    $\fdginstance$ is then the earliest situation such that $\varphi$ holds, denoted $\hat{\fdginstance}$.
\end{definition}

\newcommand{\causalorder}{\Leftarrow}
\begin{definition}[Causality]\label{def:causation}
    Recall from \autoref{tab:ufo_b_axioms} the definitions of direct (S6) and general (S7) causality. This is captured by \cref{eq:direct_causation} and \cref{eq:indirect_causation} respectively.
    
    \begin{multline}\label{eq:direct_causation}
        \forall {\,} \fdginstance, \fdginstance', \fdginstance'' : \fdg. \forall \anyevent, \event' : \mytext{Event}.\\
        \fdginstance \overset{\anyevent}\generalevent \fdginstance' \xrightarrow{\event'} \fdginstance'' \iff \ufo{directly-causes}(\anyevent, \event')
    \end{multline}
    and
    \begin{multline}\label{eq:indirect_causation}
    \forall n \geq 2. {\ } \forall {\,} \fdginstance, \dots, \fdginstance_n : \fdg. \forall \anyevent_0, \event_1, \dots, \event_n : \mytext{Event}.\\
    \fdginstance_0 \overset{\anyevent_0}\generalevent \fdginstance_1 \xrightarrow{\event_1} \dots \xrightarrow{\event_n} \fdginstance_n \iff \ufo{causes}(\anyevent_0, \event_n)
    \end{multline}
    for any sequence of causal events of length $n + 1$.
\end{definition}

Note that by \cref{eq:spontaneous_event_definition} and S6, a spontaneous event cannot have a (direct) cause:
\begin{multline*}
    \forall {\,} \fdginstance, \fdginstance', \fdginstance'' : \fdg. \forall \anyevent, \spevent : \mytext{Event}.\\
    \fdginstance \overset{\anyevent}\generalevent \fdginstance' \overset{\spevent}\looparrowright \fdginstance'' \implies \neg \ufo{causes}(\anyevent, \spevent)
\end{multline*}

Lastly, we define a causal relationship between the event that a channel is not provided with a resource, and the event that a channel stops producing that resource. Let $c \in \channel_\fdginstance$ be a channel, and $r \in \resource$ a resource that it accommodates. Let $\fdginstance$ be the situation in which $h = \channelcons_\fdginstance(c, r)$ and $k = \channelprod_\fdginstance(c, r)$ are manifested as functions $\bar{h}$ and $\bar{k}$. Then let $\event$ be the event that terminates $\bar{h}$ in $\fdginstance'$, such that $\fdginstance \overset{\event}\generalevent \fdginstance'$. Let $\Delta \langle c \rangle$ be the event introduced in \cref{eq:channel_transportation}, and $\fdginstance''$ be the situation that obtains some time $\delta$ after $\fdginstance'$, such that $k$ is not manifested in $\fdginstance''$. Then by \cref{eq:triggered_event_definition}, \cref{eq:spontaneous_event_definition}, \cref{eq:event_shorthand_definition}, and (S6 and S7, \autoref{tab:ufo_b_axioms}):
\begin{equation*}
    \fdginstance \overset{\event}\generalevent \fdginstance' \xrightarrow{\Delta \langle c \rangle} \fdginstance'' \implies \ufo{causes}(\event, \Delta \langle c \rangle)
\end{equation*}

In the definitions of {\fdgname} events in the following sections, we define preconditions and postconditions in the antecedent and consequent situations of an event, respectively. An event is defined by these conditions.
However, we do not describe the full antecedent and consequent situations -- only the relevant portion that defines the particular event. For instance, when we write about the pre and postconditions of a channel depletion event like in \cref{eq:channel_transportation}, we do not make any commitments to other components, channels, resource types, etc. that may be part of $\fdginstance$ or $\fdginstance'$. This naturally raises the question: what happens to the rest of the {\fdgname}?

\begin{postulate}[Locality of events]\label{post:events_as_changes}
    The elements of {\fdgnamearticle} {\fdgname} that are not terms in the pre- or postconditions of an event are unaltered by the event.
\end{postulate}

\autoref{post:events_as_changes} states that the elements of a situation $\fdginstance$ are ``copied over'' by an event, only changing the (aspects of) elements described in the pre and postconditions. Do note that events are not \textit{necessarily} changes of the elements in a situation, but all non-relational changes of elements of situations \textit{are} events\cite{ufo_ab}.

\begin{postulate}[The passive nature of time]\label{post:time_is_passive}
    In this work, we assume that the passing of time effects no change in the \textit{moments} of objects.
\end{postulate}

We admit that \autoref{post:time_is_passive} is too strong a postulate for physical systems. However, a description of the mechanisms of passing time is not given in UFO and its exploration warrants intent study. We leave this to future work.

\begin{remark}[Events as Types]
    It is worth noting that we have introduced events as \textit{types}, fully characterized by their antecedent and consequent situations. When we write something of the form
    \[
        \forall {\,} \fdginstance, \fdginstance' : \fdg. \forall \event : \mytext{Event}. {\ } \fdginstance \xrightarrow{\event} \fdginstance' \iff \dots
    \]
    The universal quantification indicates that we are introducing a property that holds for all instances of a particular event nature (i.e. an event type). For instance in \cref{eq:triggered_event_definition}, the combination of universal quantification and bi-implication indicates a (complete) characterization of the event type ``triggered event''.

    When we talk about particular events, we usually do so by means of existential quantification:
    \[
        \exists \fdginstance, \fdginstance' : \fdg. {\ } \fdginstance \xrightarrow{\event} \fdginstance' \implies \dots
    \]
    Here, the implication indicates that we are deriving a property that holds for the occurrence of a particular event.

    Contrary to what \cref{eq:event_shorthand_definition} suggests, generalization sets over event types are not \textit{necessarily} disjoint.

    This position aims to align with the notion of event types as employed in COVER~\cite{sales_cover}. We leave a more rigorous integration of COVER, as well as a detailed taxonomy of the events we present in this paper, for future work.
\end{remark}

\section{On faults, errors and fault activation}\label{sec:fault_creation_error_fault_activation}
The previous section describes the modeling of dynamic aspects of {\fdgname}s in a \textit{general} sense. In this section, we specifically describe how faults come about, and how faults and errors evolve over time. We also disambiguate between the core notions of fault creation, errors and fault activation.

Our notions of faults, errors and fault activation is inspired by the taxonomy proposed in \cite{Avizienis2004}. We redefine these concepts in the context of {\fdgname}s using the ontological framework of objects and events provided by UFO\cite{ufo_a}\cite{ufo_b}.

\mypar{Example}
An example of \textit{fault creation} in the running example is the power supply unit (PSU) burning out. The created fault (a molten internal winding) is dormant as long as we do not try to use the PSU. \textit{Fault activation} commences at the instant we put a load on the PSU and expect an output current. The fault then prevents the PSU from supplying this expected current, resulting in an \textit{error}.

Note that we do not consider \textit{failures} in this work. A notion of failure pertains to a notion of \textit{service}\cite{Avizienis2004} or \textit{goals}\cite{wer201710}, and we do not consider such notions in this work. We do consider \textit{expectations}, which we describe only as an analytical artifact.

\begin{definition}[Fault creation]\label{def:fault_creation}
In this work, we consider the \textit{creation} of a set of faults $F \subseteq \fault_{\fdginstance'}$ in an object $\mytext{obj} \in \mytext{ComponentId}_{\fdginstance'} \cup \mytext{ChannelId}_{\fdginstance'}$ a spontaneous atomic event, denoted $\fdginstance \overset{\faultcreation_F}\looparrowright \fdginstance'$, such that $\forall \faultinstance{f} \in F. {\ } \beta(\faultid_{\fdginstance'}(\faultinstance{f})) = \mytext{obj}$.
\end{definition}

A spontaneous event is a fault creation event, iff the spontaneous event creates faults in a component or channel that did not exist in that object before the event.
Let $F \subseteq \fault_{\fdginstance'}$ be the set of faults that are created. Let $\fdginstance, \fdginstance' : \fdg$. Then:
\begin{equation}\label{eq:fault_creation}
    \fdginstance \overset{\faultcreation_{F}}\looparrowright \fdginstance' \iff \forall \faultinstance{f} \in F. {\ } \bigl( \neg \exists \faultinstance{f} \in \fault_{\fdginstance} \bigr) \wedge \bigl( \exists \faultinstance{f} \in \fault_{\fdginstance'} \bigr)
\end{equation}



In the consequent situation of a fault creation event, the created faults prevent the manifestation of capabilities of the particular object. By the existential dependency of interfaces (as relators) on the manifestation of such capabilities, any interface that is associated with these capabilities cannot exist. Let $\fdginstance, \fdginstance' : \fdg$. Then:
\begin{multline*}
    \fdginstance \overset{\faultcreation_F}\looparrowright \fdginstance' \implies \forall (\mytext{id}_{\faultinstance{f}}, C) \in F. {\ } \forall f \in C.\\
    \neg \exists o \in \ointerface_{\fdginstance'}. {\ } f \in \quaindcomp^o_{\fdginstance'}(o) \wedge \neg \exists i \in \iinterface_{\fdginstance'}. {\ } f \in \quaindcomp^i_{\fdginstance'}(i)
\end{multline*}

I.e.: all material $\mytext{provides}$ and $\mytext{providedBy}$ relations that are existentially dependent on the capabilities covered by $F$ are terminated \textit{synchronously} in the fault creation's consequent situation.


\begin{definition}[Errors]\label{def:error}
    
\textbf{Errors} are atomic events that are triggered by unmet expectations placed on capabilities in a particular situation. An error event $\err \langle f \rangle$ describes the period of time during which a capability $f$ remains erroneous. An error event ends as soon as the capability is no longer erroneous.

\end{definition}

An error is triggered in the antecedent situation $\fdginstance$, if and only if we have an expectation for the manifestation of a function in $\fdginstance$, and that expectation is not met. The situation that an error brings about is one where either i) there is no expectation for the manifestation of the capabilities that the error comprises, or ii) the capabilities that the error comprises are manifested again. I.e. an error \textit{persists} so long as the expectation exists and the capabilities that it comprises are not manifested. Formally:
\begin{multline}\label{eq:error}
    \forall e \in \expectation_\fdginstance. {\ } \exists \fdginstance, \fdginstance' : \fdg. {\ }
    \hat\fdginstance \xrightarrow{\err \langle e \rangle} \fdginstance' \iff\\
    \neg \mytext{meets}(\hat\fdginstance, e) \wedge \left( e \not \in \expectation_{\fdginstance'} \vee \mytext{meets}(\fdginstance', e) \right)
\end{multline}


Where by $\hat\fdginstance$ in \cref{eq:error} we mean the earliest situation $\fdginstance$ in which the expectation $e$ is not met (see \autoref{def:earliest_situation}). Note that $e \not \in \expectation_{\fdginstance'} \implies \neg \mytext{meets}(\fdginstance', e)$ -- we only include it in the formula to emphasize this special case.


In the consequent situation of $\err \langle f \rangle$, the phase of the capability that the error comprises has changed from \textit{Erroneous} back to \textit{Functional}:
\begin{multline*}
    \forall {\,} \fdginstance, \fdginstance' : \fdg. {\ } \hat\fdginstance \xrightarrow{\err \langle f \rangle} \fdginstance' \iff\\
    \mytext{ErroneousCapability}_{\hat\fdginstance}( f ) \wedge\\
    \mytext{FunctionalCapability}_{\fdginstance'}( f )
\end{multline*}

In this work, we do not describe mechanisms that change the phase of capabilities from Erroneous to Functional. As a result, erroneous capabilities remain erroneous once they have entered this phase. That is to say: components cannot be repaired, rerouted or replaced. In \autoref{sec:structure_function} we touch on the consequences of this w.r.t. the evaluation of our model.
One way to look at this is that the number of erroneous capabilities as a function of situations is monotonic (given the order $\preceq$ over the set of situations).

\begin{proposition}\label{prop:no_repairs}
    Capabilities that are erroneous in a situation $\fdginstance$ are erroneous in all situations $\fdginstance''$ where $\fdginstance \preceq \fdginstance'' \prec \hat{\fdginstance'}$, s.t. $\hat{\fdginstance'}$ is the earliest situation in which the error is resolved.
\end{proposition}


\begin{definition}[Fault activation]\label{def:fault_activation}
Fault activation is the complex event in which the situation brought about by a fault creation event triggers an error.
\end{definition}

Let $\fdginstance \overset{\faultcreation_F}\looparrowright \fdginstance'$ be a fault creation event. The consequences of the fault creation then depend on the expectations we place on the affected capabilities. If there exists any expectation $e \in \expectation_{\fdginstance'}$ such that the created fault inhibits the manifestation of that expectation (i.e. $e \in F$), then by \cref{eq:fault_creation}, \cref{eq:faults_prevent_manifestation} and \autoref{def:error}, $\fdginstance'$ triggers an error, and $\faultcreation_F$ is the direct cause for this error (\autoref{def:causation}):
    
\begin{multline}\label{eq:fault_activation_full}
\fdginstance \overset{\faultcreation_F}\looparrowright \fdginstance' \implies \Bigl( \exists e \in \expectation_{\fdginstance'}. {\ } e \in F \wedge \neg \mytext{meets}(\fdginstance', e) \implies\\
\exists \fdginstance'' : \fdg. {\ }
\fdginstance' \xrightarrow{\err \langle e \rangle} \fdginstance'' \Bigr)
\end{multline}

If no such expectation exists in $\fdginstance'$, no error is triggered. The distinct cases of \textit{fault activation} are:
\begin{enumerate}
    \item With a witness to such an unmet expectation $e \in \expectation_{\fdginstance'}$, we obtain from \cref{eq:fault_activation_full}:
    \begin{equation}\label{eq:fault_activation}
        \fdginstance \overset{\faultcreation_F}\looparrowright \fdginstance' \xrightarrow{\err \langle e \rangle} \fdginstance''
    \end{equation}
    and by \cref{eq:direct_causation}, $\faultcreation_F$ is then a direct cause for $\err \langle e \rangle$.

    \item Otherwise, the created fault is called \textit{dormant}\cite{Avizienis2004}. Since faults inhere in objects, it will remain dormant as long as the fault persists, the object exists, or the fault is activated in a future situation where an expectation exists for the affected capability and the activation of the (until then dormant) fault causes an error.
\end{enumerate}

We do not aim to encode the dynamics of expectations in this work. We therefore leave the exploration of dormant fault activation to future work. From here on, when we refer to fault activation, we always refer to its first, causal case.

\section{Functional Dependency}\label{sec:fdep}
In this section, we define the notion of functional dependency. In the next section, we show how this notion combines with the definition of errors from the previous section to form error propagation.

Functional dependency is a complex notion that describes the necessary conditions for the manifestation of a function and comprises i) the prerequisites for the manifestation of a capability as a function, and ii) the fragment of the system architecture that fulfills these prerequisites. We address each constituent notion individually:
\begin{enumerate}[label=\roman*.]
    \item Recall that, in UFO, an event is the manifestation of a disposition, created and terminated by the activation and deactivation of a disposition\cite{ufo_b}. The manifestation of a capability as a function depends on access to certain types of resources. The instant that access to these resources is realized, the capability's manifestation as a function begins. The instant that access to these resources stops, this manifestation is terminated. We call the need for access to this set of resources the \textit{prerequisites} for the manifestation of a capability as a function.
    \item The minimal configuration fragment of the interaction between components (through channels) that provides access to the prerequisites for the manifestation of a capability in a given situation.
\end{enumerate}

\begin{definition}[Capability prerequisites]
    The manifestation of a capability as a function is contingent on access to that capability's prerequisites. We define the function $\preq_\fdginstance : \capability_\fdginstance \rightarrow \mathcal{P}(\resource)$ that returns a set of resource types.
    
    Let $\fdginstance : \fdg$ be a situation, and $f \in \capability_\fdginstance$ be a capability. Then:
    \[
    \preq_\fdginstance(f) = \{ r {\ } | {\ } r \in \resource. {\ } \mytext{consumes}(\capid_\fdginstance(f), r) \}
    \]
    or equivalently, by \cref{eq:capability_tuple_constraints}:
    \[
    \preq_\fdginstance(f) = \mytext{proj}_2(f)
    \]
\end{definition}

The notion of ``access'' of a capability to a particular resource type $r$ is formally defined based on the \textit{role}\footnote{``Role'', here, in the colloquial sense.} that capability plays:
\begin{itemize}
    \item For any component capability $f$, access is defined as the taking part of that capability in an input interface from which it \mytext{consumes} $r$.
    \item For any channel \textit{output} capability $k$, it is defined as the manifestation of that channel's input capability for $r$\footnote{The matter of \textit{when} that input capability must be manifested is discussed in \autoref{sec:error_propagation}.}.
    \item For any channel \textit{input} capability $h$, it is defined as the taking part of that capability in an output interface from which it \mytext{consumes} $r$.
\end{itemize}
Note that in all three cases, the capability \mytext{consumes} that resource.

The set of capabilities that provide access to a particular resource type is called a \textit{configuration fragment}.

\begin{definition}[Component configuration fragments]
    We define $\providedby_{\fdginstance} : \capability_\fdginstance \times \resource \rightarrow \pow{\capability_\fdginstance}$ to be the function that returns the set of channel production capabilities that supply a given component capability with a given resource type. Let $\fdginstance : \fdg$ be a situation. Let $f \in \capability_\fdginstance$ be a component capability and let $r \in \resource$ be a resource type that $f$ consumes, such that $\exists \alpha \in \component_\fdginstance. {\ } f \in \capability^\alpha_\fdginstance$. Then:
    \[
        \providedby_{\fdginstance}(f, r) = \{ \quaindchannel^i_{\fdginstance}(i) {\ } | {\ } i \in \iinterface_\fdginstance. f \in \quaindcomp^i_{\fdginstance}(i) \wedge \mytext{proj}_3 (i) = r \}
    \]
\end{definition}

\begin{definition}[Channel output configuration fragments]\label{def:channel_output}
    Recall that, by \autoref{def:channel_transportation}, there is a delayed dependency between a channel's output capability and its corresponding input capability for a particular resource. By definition, this delay is precisely $\delta$ time units.
    
    We define the function $\providedby_\fdginstance : \capability_\fdginstance \times \resource_\fdginstance \rightarrow \capability_\fdginstance$ to be the function that returns the channel's unique input capability that provides its output capability with access to a particular resource type. Let $\fdginstance^-, \fdginstance : \fdg$ be situations, such that $\fdginstance^-$ precedes $\fdginstance$ by $\delta$ time units. Let $k \in \capability_\fdginstance$ be a channel output capability and let $r \in \resource$ be a resource type that $k$ consumes, such that $\exists c \in \channel_\fdginstance. {\ } k = \channelprod_\fdginstance(c, r)$.  Then:
    \begin{equation}
        \providedby_\fdginstance(k, r) = \channelcons_{\fdginstance^-}(\beta_\fdginstance(k), r)
    \end{equation}
\end{definition}

\begin{definition}[Channel input configuration fragments]
    We define $\providedby_{\fdginstance} : \capability_\fdginstance \times \resource \rightarrow \pow{\capability_\fdginstance}$ to be the function that returns the set of component production capabilities that provide a given channel input capability with access a given resource type. Let $\fdginstance : \fdg$ be a situation. Let $h \in \capability_\fdginstance$ be a channel consumption capability and let $r \in \resource$ be a resource type that $h$ consumes, such that $\exists c \in \channel_\fdginstance. {\ } h = \channelcons_\fdginstance(c, r)$. Then:
    \[
        \providedby_{\fdginstance}(h, r) = \bigcup \{ \quaindcomp^o_{\fdginstance}(o) {\ } | {\ } o \in \ointerface_\fdginstance. {\ } h = \quaindchannel^o_{\fdginstance}(o) \}
    \]
\end{definition}

A capability $f$ is said to have \textit{access} to a resource type $r$ iff $(f, r)$ is in the domain of $\providedby_{\fdginstance}$, and $\providedby_{\fdginstance}(f, r) \neq \emptyset$ (for component capabilities and channel input capabilities):
\[
    \mytext{access}(f, r) \iff (f, r) \in \mytext{dom} \providedby_\fdginstance \wedge \providedby_{\fdginstance}(f, r) \neq \emptyset
\]

Since we only consider resource \textit{types} in this work, we do not explicitly deal with quantities or qualities of individual resources. We capture this assumption in \autoref{post:any_access_to_resource_type_sufficient}.
\begin{postulate}[Sufficient access]\label{post:any_access_to_resource_type_sufficient}
    Any access to the appropriate resource types is sufficient to satisfy the prerequisites of a capability.
\end{postulate}

Access to all of a capability's prerequisites necessarily activates (or maintains, if already activated) a capability. Similarly, losing access to any prerequisite necessarily deactivates the manifestation.
\begin{postulate}[Liveliness]\label{post:vivacitas}
    A capability is necessarily activated, respectively deactivated, when it gains or loses access to all, respectively any, of its prerequisites.
\end{postulate}
This liveliness is, in a sense, an over-approximation of the behavior of a system.

Let $\fdginstance : \fdg$ be a situation and let $f \in \capability_\fdginstance$ be a capability. By \autoref{post:vivacitas}, we obtain:
\begin{equation}\label{eq:access_necessary_manifestation}
    f \in \text{dom} \varrho_\fdginstance \iff \forall r \in \preq_\fdginstance(f). {\ } \mytext{access}(f, r)
\end{equation}

The minimal sets of such interactions that satisfy a capability's prerequisites are called \textit{minimal configuration fragments}.
\newcommand{\mcs}{\mytext{MCS}}
\begin{definition}[Minimal configuration fragments]
    A configuration fragment is a minimal configuration fragment (\mcs) w.r.t. to a capability $f$ and a resource type $r$ is a configuration fragment, iff if one element of that set would cease to be manifested, then $f$ would not have sufficient access to $r$.
\end{definition}






\begin{definition}[Functional dependency A]\label{def:functional_dependency}
    The function $\fdep_\fdginstance : \capability_\fdginstance \times \resource \rightarrow \pow{\capability_\fdginstance}$ is defined to return the minimal configuration fragments that a given capability $f$ is functionally dependent on for access to the resource type $r$. I.e. if any element is removed from a set $S \in \fdep_\fdginstance(f ,r)$, the set of channel capabilities $S$ would not provide $f$ with sufficient access to the resource type $r$, as in \cref{eq:minimal_configuration_fragments}.
\end{definition}



We define the functions $\mytext{ProvidedQuantity}_\fdginstance : \resource \times \mathcal{P}(\channel_\fdginstance) \rightarrow \mathbb{N}$ and $\mytext{RequiredQuantity}_\fdginstance : \resource \times \capability_\fdginstance \rightarrow \mathbb{N}$ to return the quantity of resources that (a set of) capabilities provide and require respectively.

Let $f \in \capability_\fdginstance$ be a component capability. Let $r \in \preq_\fdginstance(f)$ be a resource type. Then:
\begin{multline}\label{eq:minimal_configuration_fragments}
\fdep_\fdginstance(f, r) =\\
\{ S {\ } | {\ } S \subseteq \providedby_{\fdginstance}(f, r). {\ } \mytext{Prov.Q.}_\fdginstance(r, S) \geq \mytext{Req.Q.}_\fdginstance(r, f) \wedge\\
\bigl( \neg \exists S' \subset S. {\ } \mytext{Prov.Q.}_\fdginstance(r, S \setminus S' ) \geq \mytext{Req.Q.}_\fdginstance(r, f) \bigr) \}
\end{multline}

Note that by \autoref{post:any_access_to_resource_type_sufficient} all non-empty sets $S \in \fdep_\fdginstance(f, r)$ provide the capability $f$ with access to a sufficient quantity of the resource $r$ in $\fdginstance$. I.e. the set $\fdep_\fdginstance(f, r)$ is the set of minimal configuration fragments. To make this explicit, we provide definitions for the functions $\mytext{ProvidedQuantity}_\fdginstance$ and $\mytext{RequiredQuantity}_\fdginstance$ that do not consider any aspects of instances of $r$:
\begin{align*}
\mytext{ProvidedQuantity}_\fdginstance(r, S) &= |S|\\ 
\mytext{RequiredQuantity}_\fdginstance(r, f) &= 1
\end{align*}

These definitions allow us to simplify \cref{eq:minimal_configuration_fragments} by substitution:
\begin{equation}\label{eq:fdepa_simplified}
\fdep_\fdginstance(f, r) = \{ S {\ } | {\ } S \subseteq \providedby_{\fdginstance}(f, r). {\ } |S| = 1\}
\end{equation}

Finally, if we cannot construct any such minimal configuration fragment, then the prerequisites for the manifestation of a capability are not satisfied and the capability is not manifested. Let $\fdginstance : \fdg$. Let $\alpha \in \component_\fdginstance$ and $f \in \capability^\alpha_\fdginstance$. Then by \cref{eq:access_necessary_manifestation}:
\begin{equation*}
    \forall r \in \preq_\fdginstance(f). {\ } \fdep_\fdginstance(f, r) = \emptyset \implies f \not \in \text{dom} \varrho_\fdginstance
\end{equation*}



\begin{definition}[Functional Dependency B]

    Recall that, by \autoref{def:channel_transportation} and \autoref{def:channel_output}, there is a delayed dependency between a channel's output capability and its corresponding input capability for a particular resource. We define the function $\fdepb_\fdginstance : \capability_\fdginstance \times \resource \rightarrow \capability_\fdginstance$ such that
    \begin{equation}\label{eq:fdepb_simplified}
        \fdepb_\fdginstance(k, r) = \channelcons_{\fdginstance^-}(\beta_\fdginstance(k), r)
    \end{equation}
    where $\fdginstance^-$ precedes $\fdginstance$ by $\delta$ time units.

    
\end{definition}


\begin{definition}[Functional Dependency C]\label{def:functional_dependency_c}
    The function $\fdepc_\fdginstance : \capability_\fdginstance \times \resource \rightarrow \pow{\capability_\fdginstance}$ is defined to return the minimal configuration fragments that a given channel consumption capability is functionally dependent on such that, if any element is removed from a set $S \in \fdep_\fdginstance(h ,r)$, the set of channel capabilities $S$ would not provide $h$ with sufficient access to the resource type $r$.
\end{definition}

    Following similar steps as for \autoref{def:functional_dependency}, we obtain:
    \begin{equation}\label{eq:fdepc_simplified}
    \fdepc_\fdginstance(h, r) = \{ S {\ } | {\ } S \subseteq \provides_\fdginstance(h, r). {\ } |S| = 1\}
    \end{equation}

    And here too, if we cannot construct any such minimal configuration fragment, then the prerequisites for the manifestation of a capability are not met and the capability is not manifested. Let $\fdginstance : \fdg$. Let $c \in \channel_\fdginstance$ and $h \in \capability^c_\fdginstance$. Then:
    \begin{equation*}
        \forall r \in \preq_\fdginstance(h). {\ } \fdepc_\fdginstance(h, r) = \emptyset \implies h \not \in \mytext{dom} \varrho_\fdginstance
    \end{equation*}

\begin{definition}[Functional dependency]
    We define the transitive relation $\dee_\fdginstance \subseteq \mytext{CapabilityId}^2$, which denotes that for any ordered pair $(\mytext{id}_f, \mytext{id}_g) \in \dee_\fdginstance$, the capability $f$ is functionally dependent on the capability $g$ in the situation $\fdginstance$.
\end{definition}


The relation $\dee_\fdginstance$ captures all three cascading notions of functional dependency (A, B and C).

    Let $\fdginstance : \fdg$ be a situation. Let $f \in \capability_\fdginstance$ be a capability, and let $\mytext{id}_f = \capid_\fdginstance(f)$. Then:
    \begin{multline*}
        \forall r \in \preq_\fdginstance(f).\\
        \left\{\begin{alignedat}{2}
            &\forall S^A \in \fdep_\fdginstance(f, r). {\ } \mu(S^A), {\ } && \text{if} {\ } \exists \alpha \in \component_\fdginstance. f \in \capability^\alpha_\fdginstance\\
            & (\mytext{id}_f, \capid_\fdginstance(\fdepb_\fdginstance(f, r))) \in \dee_\fdginstance, {\ } && \text{if} {\ } \exists c \in \channel_\fdginstance. \overset{\rightarrow}{\theta}(c, r, f)\\
            & \forall S^C \in \fdepc_\fdginstance(f, r). {\ } \mu(S^C), {\ } && \text{if} {\ } \exists c \in \channel_\fdginstance. \overset{\leftarrow}{\theta}(c, r, f)
        \end{alignedat}
        \right.
    \end{multline*}

    Where
    \begin{align*}
        \mu(S) &\iff \forall g \in S. {\ } (\mytext{id}_f, \capid_\fdginstance(g)) \in \dee_\fdginstance \\
        \overset{\rightarrow}{\theta}(c, r, f) &\iff \channelprod_\fdginstance(c, r) = f\\
        \overset{\leftarrow}{\theta}(c, r, f) &\iff \channelcons_\fdginstance(c, r) = f
    \end{align*}

This notion of functional dependency is intimately tied to that of errors: if there exist prerequisites for a component's capability, and there does not exist a minimal configuration fragment that satisfies these prerequisites, then that capability is not, and cannot be, manifested as a function. If, in any particular situation, there exists an expectation for the manifestation of such a function, then an error is triggered.

\mypar{On Dependency and Expectation}
Knowing that the manifestation of a capability has functional dependencies, we cannot expect a capability to be manifested if such functional dependencies are not satisfied. Therefore, we introduce a cascading notion of expectation, mirroring the transitive relation of functional dependency.

\begin{stipulate}[Cascading expectations]\label{def:transitivity_of_expectations}
    If we expect a capability to be manifested, we must also expect the sufficient conditions for that manifestation to hold.
\end{stipulate}

    The sufficient conditions for the manifestation of a capability are determined by access to its prerequisites. That is:
    \begin{enumerate}
        \item If we have an expectation for a component's capability $f$, then we also have an expectation for the channels' capabilities that provide access to $f$'s prerequisites, and;
        \item If we have an expectation for a channel to produce a resource, then we also have an expectation for that channel to have consumed that resource, and;
        \item If we have an expectation for a channel to consume a resource, then we also have an expectation for some component capabilities to provide that resource.
    \end{enumerate}
    
    The cascading of expectations in this manner reflect the three types of functional dependency (A, B and C) captured by the transitive relation $\dee_\fdginstance$. Formally:

    \begin{equation*}
        \forall {\,} \fdginstance : \fdg. \forall (\mytext{id}_f, \mytext{id}_g) \in \dee_\fdginstance. {\ } \mytext{id}_f \in \expectation_\fdginstance \implies \mytext{id}_g \in \expectation_\fdginstance
    \end{equation*}




\mypar{Example}
Recall the running example from \autoref{fig:print_head}. The \texttt{print head}'s function \textit{to deposit ink} is functionally dependent on the \texttt{PSU}'s function \textit{to provide electricity}, and the \texttt{ink reservoir}'s function \textit{to provide ink}.

\section{Error propagation}\label{sec:error_propagation}

In this section, we combine the definitions of errors and functional dependency to construct a notion of error propagation. This forms the theoretical backbone for {\fdgname} analysis and FT synthesis.

Error propagation is a sequence of causal events, where each situation that is the antecedent of such an event triggers an error. This is similar in structure to the causal chain shown in \cref{eq:indirect_causation}, \autoref{def:causation}.

\begin{definition}[Error propagation]\label{def:error_propagation}
Error propagation is a complex event, denoted $\fdginstance_i \overset{E}\rightsquigarrow \fdginstance_{i+n}$, that comprises a causal sequence of events
\[
    \fdginstance_i \overset{e_i}{\generalevent} \dots \xrightarrow{e_{i+n-1}} \fdginstance_{i+n}
\]
such that each consequent situation $\fdginstance_{i+j}$, $0 < j \leq n$ in the sequence triggers an error, i.e. there exists some capability $f_j$ such that
\[
    \fdginstance_{i+j} \xrightarrow{\err \langle f_{j} \rangle} \fdginstance_{n_j}
\]
Then \textit{all} these events are mereological parts of $E$, i.e.
\[
    \mytext{has-part}(E, e_i); \qquad \mytext{has-part}(E, \err \langle f_j \rangle)
\]
\end{definition}

\mypar{Errors and functional dependency}
By \autoref{def:error}, when a producer of a particular resource experiences an error over the capability that produces it, it no longer provides that resource to any other capabilities. And through functional dependency, the manifestation of a capability as a function depends on access to certain types of resources.

These facts together mean that error propagation is observed if in an antecedent situation $\fdginstance$, the following holds:
\begin{enumerate}
    \item A capability $g$ provides a channel $c$ with a resource $r$, and;
    \item A capability $f$, with $r$ as  prerequisite, is provided with $r$ by $c$, and;
    \item Let $h = \channelcons_\fdginstance(c, r)$ be the input capability of channel $c$, such that there exists no minimal configuration fragment for $h$ that does not include $g$, and;
    \item Let $k = \channelprod_\fdginstance(c, r)$ be the output capability of cannel $c$, such that there exists no minimal configuration fragment for $f$ that does not include $k$, and;
    \item There exists an expectation for the manifestation of $f$ (whereby there exist expectations for $k$, $h$ and $g$).
\end{enumerate}
Then...
\begin{enumerate}
    \item If in the situation $\fdginstance$ fault creation occurs over $g$, i.e. $\fdginstance \overset{\faultcreation_{\{ \faultinstance{g} \}}}{\looparrowright} \fdginstance'$, then;
    \item[] Fault activation occurs, such that an error $\err \langle g \rangle$ is triggered over $g$, i.e. $\fdginstance' \xrightarrow{\err \langle g \rangle} \fdginstance_n'$, and;
    \item[] Also in $\fdginstance'$ an error $\err \langle h \rangle$ is triggered over $h$, i.e. $\fdginstance' \xrightarrow{\err \langle h \rangle} \fdginstance_{n'}'$.
    \item Then in the situation $\fdginstance'$, the capability $h$ is no longer manifested, and a depletion event $\Delta$ is triggered over the bearing channel, i.e. $\fdginstance' \xrightarrow{\Delta \langle \beta(h) \rangle} \fdginstance''$, and;
    \item[] In the situation $\fdginstance''$, $k$ is no longer manifested and an error is triggered over $k$, such that $\fdginstance'' \xrightarrow{\err \langle k \rangle} \fdginstance_j''$, and;
    \item[] In that situation $\fdginstance''$, $f$ is no longer manifested and an error $\err \langle f \rangle$ is triggered over $f$, i.e. $\fdginstance'' \xrightarrow{\err \langle f \rangle} \fdginstance_{j'}''$.
\end{enumerate}

Let $\fdginstance_i \overset{E}\rightsquigarrow \fdginstance_{i+n}$ denote a causal chain of error propagation, such that
\[
\fdginstance_i \overset{E}\rightsquigarrow \fdginstance_{i+n} \iff \fdginstance_i \overset{e}{\generalevent} \fdginstance_{i+1} \xrightarrow{\Delta_{i+1}}  \dots \xrightarrow{\Delta_{i+n-1}} \fdginstance_{i+n} 
\]
where the event $e$ and the depletion events $\Delta_{i + 1}$ and $\Delta_{i+2}$ combine to form the complex error propagation event:
\begin{equation*}
    \mytext{has-part}(E, e)
\end{equation*}
and
\begin{equation*}
    \forall j. {\ } i < j < i+n. {\ } \mytext{has-part}(E, \Delta_{j})
\end{equation*}

Lastly, we call such a causal chain an \textit{error propagation} iff in the consequent situation of each event an error is triggered:
\begin{equation*}
    \forall j. {\ } i \leq j < i+n. {\ } \exists \fdginstance' : \fdg. {\ } \fdginstance_j \xrightarrow{\err_j} \fdginstance' \wedge \mytext{has-part}(E, \err_{j})
\end{equation*}

The situation this complex event brings about is determined by the temporal order of the situations brought about by the events it comprises. We find that the error propagation $E$, with $n \geq 3$ events and $n+1$ situations $\fdginstance_i, \dots, \fdginstance_{i+n}$, is a sequence over $\prec$, bringing about the latest situation $\fdginstance_{i+n}$ in that sequence.

\mypar{Example of error propagation}
In \autoref{fig:example_error_propagation} we show a smaller version of the running example introduced in \autoref{fig:print_head_formal}. Let $g$ be the symbol for the \texttt{Ink reservoir}'s capability to \textit{provide ink}. Let $f$ denote the \texttt{Print head}'s capability to \textit{deposit ink}. Let there be a channel $c$ between them that accommodates \texttt{Ink} (denoted $r$). Let $h = \channelcons_\fdginstance(c, r)$ be the channel's input capability, and let $k = \channelprod_\fdginstance(c, r)$ be the channel's output capability. Lastly, recall that $\bar{g}$ and $\bar{f}$ denote the manifestation of the capabilities $g$ and $f$ as a function respectively.

The left-hand side of \autoref{fig:example_error_propagation} labels each situation, and the edges between those labels depict events (e.g. the event $\fdginstance \overset{\faultcreation_{\{ \faultinstance{g} \}}}\looparrowright \fdginstance'$ is depicted).

\begin{figure}[hb!]
\centering
\includegraphics[]{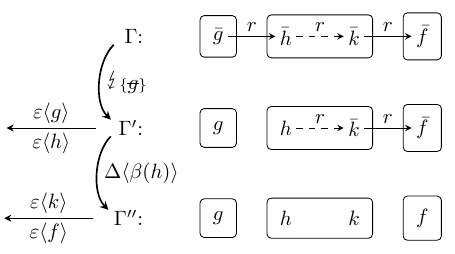}
\caption{An example of error propagation.}\label{fig:example_error_propagation}
\end{figure}

We wish to place an expectation on the manifestation of $f$. By \autoref{def:transitivity_of_expectations}, we must then also place an expectation on $k$, $h$ and $g$. \autoref{fig:example_error_propagation} then shows how an error over $g$ propagates, through $c$, to an error over $f$.

A solid arrow in the diagram denotes the presence of an interface. As the errors propagate, the interactions between capabilities are terminated and the arrows are removed. The dashed arrow signifies a satisfied functional dependency between the output and input capabilities of a channel (for a particular resource type).

Furthermore, if $\bar{f}$ provides another expected function in $\fdginstance''$, the propagation continues.

\begin{remark}[The end point of an error]
Note that in UFO, all events must have an endpoint. Events are also immutable. Foregoing the philosophical debate on determinism v.s. free will, this only makes sense if these points are necessarily in the past. Therefore, these error events must have obtained in a situation before we ``query" (or ``factually know") anything about such error events.
\end{remark}

\begin{remark}[The causality of errors]    
There are no causal relationships between error events themselves. The causal aspect of an error propagation is carried by the causality between fault creation and depletion events across functional dependency.
\end{remark}

\begin{proposition}[Direct causes for errors]\label{prop:causes_for_errors}
For any function, fault creation (\autoref{def:fault_activation}) and propagation of depletion events through unmet functional dependency (\autoref{def:error_propagation}) make up the two possible direct causes for errors in {\fdgname}s.
\end{proposition}

\section{{\fdgname} Analysis}\label{sec:structure_function}
In this section we provide an error function that builds upon the notion of error propagation from the previous section. This error function is instrumental in shaping the FT synthesis algorithm.

The error function of {\fdgname}s answers a diagnostic question: \textit{``What events could lead to an error over this capability?''}
This question is, in a way, the dual of the predictive question answered by the definitions of fault activation and error propagation: \textit{``What are the consequences of an error over this capability?''}
Since we have at present already developed that ``predictive'' framework, we use it now to answer the diagnostic question.

\autoref{tab:glossary_fdg_perdurant} provides a glossary of the perdurant {\fdgname} concepts introduced earlier.

\begin{table}[]
    \centering
    \begin{tabular}{cl}
        \toprule
        \textbf{Notation} & \textbf{Description} \\
        \midrule
        $\fdginstance \prec \fdginstance'$ & Situation $\fdginstance$ precedes the situation $\fdginstance'$ \\
        \midrule
        $\fdginstance \preceq \fdginstance'$ & Situation $\fdginstance$ precedes, or coincides with, $\fdginstance'$ \\
        \midrule
        $\fdginstance \overset{\event}\generalevent \fdginstance'$ & Event $\event$ with antecedent $\fdginstance$ and consequent $\fdginstance'$ \\
        \midrule
        $\fdginstance \xrightarrow{\event} \fdginstance'$ & A triggered event $\event$ \\
        \midrule
        $\fdginstance \overset{\spevent}\looparrowright \fdginstance'$ & A spontaneous event $\spevent$ \\
        \midrule
        $\fdginstance \xrightarrow{\Delta \langle c \rangle} \fdginstance'$ & A depletion event over the channel $c$ \\
        \midrule
        $\fdginstance \overset{\faultcreation_F}\looparrowright \fdginstance'$ & A fault creation event, creating the faults $F$ \\
        \midrule
        $\fdginstance \xrightarrow{\err \langle e \rangle} \fdginstance'$ & An error event over the expectation $e$ \\
        \midrule
        $\fdginstance \overset{\faultcreation_F}\looparrowright \fdginstance' \xrightarrow{\err \langle e \rangle} \fdginstance''$ & Fault activation over $e \in F$ \\
        \midrule
        $\fdginstance_i \overset{E}\rightsquigarrow \fdginstance_{i+n}$ & Error propagation \\
        \bottomrule
    \end{tabular}
    \caption{A glossary of perdurant {\fdgname} concepts.}
    \label{tab:glossary_fdg_perdurant}
\end{table}

\begin{definition}[{\fdgname} error function]\label{def:fdg_eval}
We use the notation $[\![ f ]\!]^{\fdginstance_{n}}_{\fdginstance_0} \in \{ \top, \bot \}$ to indicate that a capability $f$ is erroneous in the situation $\fdginstance_n$:
\begin{multline*}
    \forall \fdginstance_n : \fdg. \forall f \in \capability_{\fdginstance_n}. {\ } [\![ f ]\!]^{\fdginstance_{n}}_{\fdginstance_0} \iff \exists \fdginstance_{n-i}, \fdginstance_{n+j} : \fdg.\\
    \fdginstance_0 \preceq \fdginstance_{n-i} \preceq \fdginstance_n \prec \fdginstance_{n+j} \wedge \fdginstance_{n-i} \xrightarrow{\err \langle f \rangle} \fdginstance_{n+j}
\end{multline*}

\end{definition}

\mypar{A structure function implementation}
In its current form, the definition of the error function given in \autoref{def:fdg_eval} is still not very helpful in answering the diagnostic question. We have described in \autoref{prop:causes_for_errors} that, to find all possible causes for an error $\fdginstance_{n-i} \xrightarrow{\err \langle f \rangle} \fdginstance_{n+j}$ that spans the situation $\fdginstance_n$, we must consider both i) fault activation of a fault created in its bearer and ii) error propagation through functional dependency. I.e. given $[\![f]\!]_{\fdginstance^n} = \top$, we know that there must exist a previous situation $\fdginstance_{n-i}$ from which either fault creation or error propagation occurred.

To cite the principal NASA handbook on Fault Tree construction\cite{nasa2002}:
\begin{quote}
\textit{``In this way the analyst proceeds down the tree continually transferring the point of view from
mechanism to mode, and continually approaching finer resolution in defining mechanisms and
modes, until ultimately, the limit of resolution of the tree is reached. This limit consists of basic
component failures of one sort or another, and the tree is now complete.''}
\end{quote}

We now split on fault creation and the three cases of functional dependency to derive such a recursive structure function.

    
First, by definition of errors (\autoref{def:error}) there must exist an expectation $f \in \expectation_{\fdginstance_n}$. I.e.:
\begin{equation}\label{eq:strcuture_function_implies_expectation}
[\![ f ]\!]^{\fdginstance_{n}}_{\fdginstance_0} \implies f \in \expectation_{\fdginstance_n}
\end{equation}

Then, let $E_{\fdginstance_n}$ be a set of erroneous capabilities in $\fdginstance_n$:
\begin{equation}\label{eq:e_erroneous_capabilities}
E_{\fdginstance_n} = \left\{ f {\ } | {\ } f \in \capability_{\fdginstance_n}. {\ } \mytext{ErroneousCapability}(\capid_{\fdginstance_n}(f)) \right\}
\end{equation}

By definition of the erroneous phase of capabilities, there exists an error for each element of $E_{\fdginstance_n}$ whose begin point is in $\fdginstance_n$ or precedes $\fdginstance_n$, and whose end point is after $\fdginstance_n$:
\begin{multline}\label{eq:errors_span_f}
    \forall f \in E_{\fdginstance_n}. \exists \fdginstance_{n-i}, \fdginstance_{n+j} : \fdg.\\
    \fdginstance_{n-i} \preceq \fdginstance_n \prec \fdginstance_{n+j} \wedge \fdginstance_{n-i} \xrightarrow{\err \langle f \rangle} \fdginstance_{n+j}
\end{multline}

Let $\fdginstance_0, \fdginstance_n : \fdg$ be {\fdgname}s, such that $\fdginstance_0 \preceq \fdginstance_n$. Let $f \in \capability_{\fdginstance_n}$ be a capability. By \autoref{def:fdg_eval}, it is clear that
\[
    [\![ f ]\!]^{\fdginstance_{n}}_{\fdginstance_0} \iff f \in E_{\fdginstance_n}
\]

\mypar{Fault Creation}
Let $f \in \expectation_{\fdginstance_n}$ be an expectation for the manifestation of $f$ in $\fdginstance_n$. Then by \cref{eq:fault_activation_full}:
\begin{equation}\label{eq:fault_activation_under_structure_function}
    \left( \exists \fdginstance_{n-i}, \fdginstance_{n-i}' : \fdg. {\ } \fdginstance_{n-i} \overset{\faultcreation_{\{ \faultinstance{f} \}}}\looparrowright \fdginstance_{n-i}' \right) \implies {\ } f \in E_{\fdginstance_n}
\end{equation}

where $\fdginstance_{n-i}' \preceq \fdginstance_n$.

\noindent Furthermore, error propagation can occur;

\medskip
\noindent\mypar{Propagation case A (input propagation)}

The capability $f$ is a component capability:
\[
\exists \alpha \in \component_{\fdginstance_n}. {\ } f \in \capability^\alpha_{\fdginstance_n}
\]

Let $f \in \expectation_{\fdginstance_n}$ be an expectation for the manifestation of $f$ in $\fdginstance_n$. Then:
\begin{multline*}
    \exists r \in \preq_{\fdginstance_n}(f). \left( \forall S^A \in \fdep_{\fdginstance_0}(f, r). {\ } S^A \subseteq E_{\fdginstance_n} \right) \implies f \in E_{\fdginstance_n}
\end{multline*}

By \cref{eq:errors_span_f}, $S^A \subseteq E_{\fdginstance_n}$ implies there is an error event for all capabilities in each set of minimal configuration fragments in $\fdep_{\fdginstance_0}(f, r)$. I.e.:
\begin{equation}\label{def:structure_function_big_wedge}
    \exists r \in \preq_{\fdginstance_n}(f). \left( \bigwedge_{\mathrlap{S^A \in \fdep_{\fdginstance_0}(f, r)}} S^A \subseteq E_{\fdginstance_n} \right) \implies f \in E_{\fdginstance_n}
\end{equation}

A conjunction over the empty set is vacuously true, i.e.:
\[
\bigwedge_{S^A \in {\ } \emptyset} = \top
\]
whereby $\fdep_{\fdginstance_n}(f, r) = \emptyset$ implies that \autoref{def:structure_function_big_wedge} is (vacuously) true. This matches the semantics that ``if we cannot construct any such minimal configuration fragment, then the prerequisites for the manifestation of a capability are not met and the capability cannot be manifested'' (\autoref{def:functional_dependency}).

\medskip
\noindent \mypar{Propagation case B (depletion propagation)}

The capability $f$ is a channel production capability:
\[
\exists c \in \channel_{\fdginstance_n}. \exists r \in \resource. {\ } f = \channelprod_{\fdginstance_n}(c, r)
\]

Given that $f \in \expectation_{\fdginstance_n}$:
\begin{equation*}
    \left( \exists \fdginstance_{n-i}, \fdginstance_{n-i}' : \fdg. \fdginstance_{n-i} \xrightarrow{\Delta \langle \beta(f) \rangle} \fdginstance_{n-i}' \right) \implies f \in E_{\fdginstance_n}
\end{equation*}

where $\fdginstance_{n-i}' \preceq \fdginstance_n$.

\medskip
\noindent\mypar{Propagation case C (output propagation)}

$f$ is a channel consumption capability:
\[
\exists c \in \channel_{\fdginstance_n}. \exists r \in \resource. {\ } f = \channelcons_{\fdginstance_n}(c, r)
\]

Given that $f \in \expectation_{\fdginstance_n}$:
\begin{multline*}
    \exists r \in \preq_{\fdginstance_n}(f).\\
    \left( \forall S^C \in \fdepc_{\fdginstance_0}(f, r). {\ } S^C \subseteq E_{\fdginstance_n} \right) \implies f \in E_{\fdginstance_n}
\end{multline*}

By \cref{eq:e_erroneous_capabilities}, $S^C \subseteq E_{\fdginstance_n}$ implies there is an error event for all capabilities in each set of minimal configuration fragments in $\fdepc_{\fdginstance_0}(f, r)$. I.e.:
\begin{equation}\label{def:structure_function_big_wedge_c}
    \exists r \in \preq_{\fdginstance_n}(f). \left( \bigwedge_{\mathrlap{S^C \in \fdepc_{\fdginstance_0}(f, r)}} S^C \subseteq E_{\fdginstance_n} \right) \implies f \in E_{\fdginstance_n}
\end{equation}

Recall that we defined conjunction such that $\bigwedge_{S^C \in {\ } \emptyset} = \top$. Therefore, $\fdepc_{\fdginstance_n}(f, r) = \emptyset$ implies that \autoref{def:structure_function_big_wedge_c} is (vacuously) true. This matches the semantics that ``if we cannot construct any such minimal configuration fragment, then the prerequisites for the manifestation of a capability are not met and the capability cannot be manifested'' (\autoref{def:functional_dependency_c}).


\begin{figure*}[ht!]

    \normalsize


    \begin{align*}
    \forall {\,} \fdginstance_n : \fdg. {\ } f \in \capability_{\fdginstance_n}. {\ } [\![ f ]\!]^{\fdginstance_{n}}_{\fdginstance_0} \iff \mytext{id}_f \in \expectation_{\fdginstance_{n}} \wedge \Biggl( &\exists \fdginstance_{n-i}, \fdginstance_{n-i}' : \fdg. {\ } \fdginstance_{n-i} \overset{\faultcreation_{F}}{\looparrowright} \fdginstance_{n-i}' {\ } \vee
    \bigvee_{r \in \preq_{\fdginstance_{n}}(f)} \biggl( \\
        &\Bigl( \exists \alpha \in \component_{\fdginstance_{n}}. {\ } f \in \capability^\alpha_{\fdginstance_{n}} \implies {\ } \bigwedge_{S^A \in \fdep_{\fdginstance_0}(f, r)} \bigvee_{k \in S^A} [\![k]\!]^{\fdginstance_{n}}_{\fdginstance_0} \Bigr) \vee \\
        &\Bigl( \exists c \in \channel_{\fdginstance_{n}}. {\ } f = \channelprod_{\fdginstance_{n}}(c, r) \implies {\ } \exists \fdginstance_{n-j} : \fdg. {\ } [\![\fdepb_{\fdginstance_0}(f, r)]\!]^{\fdginstance_{n-j}}_{\fdginstance_0} \Bigr) \vee \\
        &\Bigl( \exists c \in \channel_{\fdginstance_{n}}. {\ } f = \channelcons_{\fdginstance_{n}}(c, r) \implies {\ } \bigwedge_{S^C \in \fdepc_{\fdginstance_0}} \bigvee_{g \in S^C} [\![g]\!]^{\fdginstance_{n}}_{\fdginstance_0} \Bigr) \biggr) \Biggr)
    \end{align*}

    \text{Where $\fdginstance_0 \prec \fdginstance_{n-i}$ and $\fdginstance_{n-i}' \preceq \fdginstance_{n}$, and $\mytext{id}_f = \capid_{\fdginstance_n}(f)$, and $\mytext{id}_{\faultinstance{f}} \in F$, and $\fdginstance_{n-j} \prec \fdginstance_n$.}\\
    \text{Additionally, there exists time points $t_{n-j}$ and $t_n$ such that $\mytext{obtainsIn}(\fdginstance_{n-j}, t_{n-j})$ and $\mytext{obtainsIn}(\fdginstance_n, t_n)$, and}\\
    \text{$t_n - t_{n-j} \geq \delta$.}

    \hrulefill
    \vspace*{4pt}
    
    \caption{The definition of the {\fdgname} structure function $[\![\cdot]\!]^{\fdginstance_n}_{\fdginstance_0}$.}\label{fig:structure_function}
\end{figure*}

\mypar{Final remarks}
For completeness, we must also define the case in which $\fdginstance_n$ \textit{is} the initial situation $\fdginstance_0$
. This is trivial by the definitions or errors (\autoref{def:error}) and causation (\autoref{def:causation}); let $f \in \expectation_{\fdginstance_n}$. Then:
\[
[\![ f ]\!]^{\fdginstance_{n}}_{\fdginstance_0} \iff f \not \in \text{dom} \varrho_{\fdginstance_n}
\]
and $\err \langle f \rangle$ has no cause.


Furthermore, the evaluation of a capability using this structure function is only possible if there exists a future situation $\fdginstance_{n+j}$ such that $\fdginstance_n \prec \fdginstance_{n+j}$ in or before which all relevant errors have their end point.

\mypar{Example}
In \autoref{fig:structure_function_example} we show a system with redundant cables between the power supply unit (PSU) and the print head from the running example. Let $g$ denote the \texttt{PSU}'s capability to \textit{provide electricity}. Let $h_0, k_0$, and $h_1, k_1$ respectively denote \texttt{cable harness 1} and \texttt{cable harness 2}'s capabilities to \textit{transport electricity}. Finally, let $f$ denote the \texttt{ink head}'s capability to \textit{deposit ink}.

\begin{figure}[ht!]
\centering

\includegraphics[]{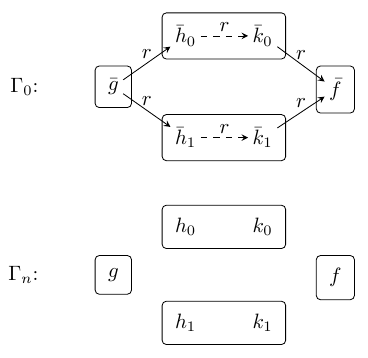}

\caption{An example of a system with redundancy.}
    \label{fig:structure_function_example}
\end{figure}

We now show that, using \cref{fig:structure_function}, we can symbolically determine all causes for an error over $\bar{f}$ in situation $\fdginstance^n$ w.r.t. the situation $\fdginstance_0$ as depicted in \autoref{fig:structure_function_example}. Let $\gamma \in \component_{\fdginstance_0}$ and $\alpha \in \component_{\fdginstance_0}$ be components, such that $\bar{g} \in \function^\gamma_{\fdginstance_0}$ and $\bar{f} \in \function^\alpha_{\fdginstance_0}$. Let $c_0 \in \channel_{\fdginstance_0}$ and $c_1 \in \channel_{\fdginstance_0}$ be channels, such that $h_0, k_0 \in \function^{c_0}_{\fdginstance_0}$ and $h_1, k_1 \in \function^{c_1}_{\fdginstance_0}$.

The arrows in \autoref{fig:structure_function_example} depict interfaces. I.e.: let $o_0 \in \ointerface_{\fdginstance_0}$ be an output interface, such that $\quaindcomp^o_{\fdginstance_0}(o_0) = \{ \capid_{\fdginstance_0}(g) \}$ and $\quaindchannel^o_{\fdginstance_0}(o_0) = \capid_{\fdginstance_0}(h_0)$ and $\mytext{proj}_3(o_0) = r$. Let $o_1 \in \ointerface_{\fdginstance_0}$ be an output interface, such that $\quaindcomp^o_{\fdginstance_0}(o_1) = \{ \capid_{\fdginstance_0}(g) \}$ and $\quaindchannel^o_{\fdginstance_0}(o_1) = \capid_{\fdginstance_0}(h_1)$ and $\mytext{proj}_3(o_1) = r$.

Let $i_0 \in \iinterface_{\fdginstance_0}$ be an input interface, such that $\quaindcomp^i_{\fdginstance_0}(i_0) = \{ \capid_{\fdginstance_0}(f) \}$ and $\quaindchannel^i_{\fdginstance_0}(i_0) = \capid_{\fdginstance_0}(k_0)$ and $\mytext{proj}_3(i_0) = r$. Let $i_1 \in \iinterface_{\fdginstance_0}$ be an input interface, such that $\quaindcomp^i_{\fdginstance_0}(i_1) = \{ \capid_{\fdginstance_0}(f) \}$ and $\quaindchannel^i_{\fdginstance_0}(i_1) = \capid_{\fdginstance_0}(k_1)$ and $\mytext{proj}_3(i_1) = r$.

From this, it follows that $\channelcons_{\fdginstance_0}(c_0, r) = h_0$ and $\channelprod_{\fdginstance_0}(c_0, r) = k_0$, and $\channelcons_{\fdginstance_0}(c_1, r) = h_1$ and $\channelprod_{\fdginstance_0}(c_1, r) = k_1$.

\begin{table}[t!]
\begin{tabular}{llll}
\textbf{Capability} & \textbf{$\fdep_{\fdginstance_0}(\cdot, r)$} & \textbf{$\fdepb_{\fdginstance_0}(\cdot, r)$} & \textbf{$\fdepc_{\fdginstance_0}(\cdot, r)$} \\
\toprule
$f$ & $\left\{ \{ k_0 \}, \{ k_1 \} \right\}$ & -- & -- \\
\midrule
$k_0$ & -- & $h_0$ & -- \\
\midrule
$k_1$ & -- & $h_1$ & -- \\
\midrule
$h_0$ & -- & -- & $\left\{ \{ g \} \right\}$ \\
\midrule
$h_1$ & -- & -- & $\left\{ \{ g \} \right\}$ \\
\midrule
$g$ & $\emptyset$ & -- & -- \\
\bottomrule
\end{tabular}
\caption{The functional dependency relationships in $\fdginstance_0$ of \autoref{fig:structure_function_example}.}\label{tab:structure_function_example_fdep}
\end{table}

\mypar{Derivation}
\autoref{tab:structure_function_example_fdep} presents the functional dependency relationships in $\fdginstance_0$ so that we can refer to them in the derivation. We also use the symbol $\faultcreation$ as shorthand notation for the existence of a fault creation event;
\[
\faultcreation^{\fdginstance_n}_{\fdginstance_0} \langle \mytext{id}_f \rangle \iff \exists \fdginstance_{n-i}, \fdginstance_{n-i}' : \fdg. {\ } \fdginstance_{n-i} \overset{\faultcreation_{F}}{\looparrowright} \fdginstance_{n-i}'
\]
where $\fdginstance_0 \prec \fdginstance_{n-i}$ and $\fdginstance_{n-i}' \preceq \fdginstance_{n}$, and $\mytext{id}_{\faultinstance{f}} \in F$.

We evaluate the structure function $[\![ f ]\!]^{\fdginstance_n}_{\fdginstance_0}$ for the presence of an error over the capability $f$;
\begin{align}\label{eq:example_structure_function_derivation}
\begin{split}
    [\![ f ]\!]^{\fdginstance_{n}}_{\fdginstance_0} &\iff \mytext{id}_f \in \expectation_{\fdginstance_n} \wedge \left( \faultcreation^{\fdginstance_n}_{\fdginstance_0} \langle \mytext{id}_f \rangle \vee \left( [\![k_0]\!]^{\fdginstance_{n}}_{\fdginstance_0} \wedge [\![k_1]\!]^{\fdginstance_{n}}_{\fdginstance_0} \right) \right)\\
    [\![ k_0 ]\!]^{\fdginstance_{n}}_{\fdginstance_0} &\iff \faultcreation^{\fdginstance_n}_{\fdginstance_0} \langle \mytext{id}_{k_0} \rangle \vee \left( \exists \fdginstance_{n-j} : \fdg. {\ } [\![h_0]\!]^{\fdginstance_{n-j}}_{\fdginstance_0} \right) \\
    [\![ k_1 ]\!]^{\fdginstance_{n}}_{\fdginstance_0} &\iff \faultcreation^{\fdginstance_n}_{\fdginstance_0} \langle \mytext{id}_{k_1} \rangle \vee \left( \exists \fdginstance_{n-j'} : \fdg. {\ } [\![h_1]\!]^{\fdginstance_{n-j'}}_{\fdginstance_0} \right) \\
    [\![h_0]\!]^{\fdginstance_{n-j}}_{\fdginstance_0} &\iff \faultcreation^{\fdginstance_{n-j}}_{\fdginstance_0} \langle \mytext{id}_{h_0} \rangle \vee [\![g]\!]^{\fdginstance_{n-j}}_{\fdginstance_0} \\
    [\![h_1]\!]^{\fdginstance_{n-j'}}_{\fdginstance_0} &\iff \faultcreation^{\fdginstance_{n-j'}}_{\fdginstance_0} \langle \mytext{id}_{h_1} \rangle \vee [\![g]\!]^{\fdginstance_{n-j'}}_{\fdginstance_0} \\
    [\![g]\!]^{\fdginstance_{n-j}}_{\fdginstance_0} &\iff \faultcreation^{\fdginstance_{n-j}}_{\fdginstance_0} \langle \mytext{id}_{g} \rangle \\
    [\![g]\!]^{\fdginstance_{n-j'}}_{\fdginstance_0} &\iff \faultcreation^{\fdginstance_{n-j'}}_{\fdginstance_0} \langle \mytext{id}_{g} \rangle
\end{split}
\end{align}

Recall that expectation is transitive over functional dependency (e.g. $\mytext{id}_f \in \expectation_{\fdginstance_n} \implies \mytext{id}_{k_0} \in \expectation_{\fdginstance_n}$).

Bottom-up substitution of the recursively defined equivalences yields a boolean formula with atomic propositions. E.g.:
\[
[\![h_0]\!]^{\fdginstance_{n-j}}_{\fdginstance_0} \iff \faultcreation^{\fdginstance_{n-j}}_{\fdginstance_0} \langle \mytext{id}_{h_0} \rangle \vee [\![g]\!]^{\fdginstance_{n-j}}_{\fdginstance_0}
\]
becomes:
\[
[\![h_0]\!]^{\fdginstance_{n-j}}_{\fdginstance_0} \iff \faultcreation^{\fdginstance_{n-j}}_{\fdginstance_0} \langle \mytext{id}_{h_0} \rangle \vee \faultcreation^{\fdginstance_{n-j}}_{\fdginstance_0} \langle \mytext{id}_{g} \rangle
\]

There are many sequences\footnote{``Sequence'' here is a loose term that also permits concurrent collections of events.} of events that satisfy \cref{eq:example_structure_function_derivation}. One example is given in \autoref{fig:structure_function_example_lts}. The set of configurations that satisfy this formula are exactly the distinct failure modes that satisfy the structure function of the corresponding fault tree.

\begin{figure}[ht!]
\centering

\includegraphics[]{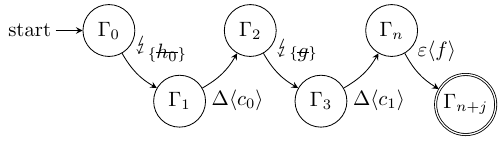}

\caption{A sequence of events that satisfies the formula $[\![ f ]\!]^{\fdginstance_{n}}_{\fdginstance_0}$.}
    \label{fig:structure_function_example_lts}
\end{figure}

\section{KG querying and {\fdgname} construction}\label{sec:kg_query_fdg_construction}
In this section, we take a side-step and show how we query KGs to instantiate {\fdgname}s.

We assume the KG is well-formed with respect to the ontological constraints. Enforcement of these data constraints in the GraphDB database (e.g. using SHACL) is a technical detail that we do not discuss in this work.


\subsection{Semantic mapping}\label{sec:gufo_semantic_mapping}
\autoref{tab:sparql_to_fdg} shows the a straightforward semantic mapping between RDF triple patterns and {\fdgname} elements.

\begin{table}[]
\centering
\begin{tabular}{ll}
\textbf{RDF triple pattern}             & \textbf{{\fdgname} construct} \\
\toprule
?x rdf:type :\texttt{Class}             & $x \in \phi( \texttt{Class})_{\fdginstance_0}$ \\
\midrule
?x :has ?f                              & $f \in \capability^x_{\fdginstance_0}$ \\
\midrule
?x :consumes ?r                         & $r \in \mytext{proj}_2 {\ } \capid^{-1}(x)$ \\
\midrule
?x :produces ?r                         & $f \in \mytext{proj}_3 {\ } \capid^{-1}(x)$ \\
\midrule
?c :accommodates ?r                     & $r \in \mytext{proj}_2 {\ } \capid^{-1}(c)$ \\
\midrule
?i :channel ?c                          & $c = \mytext{proj}_1 {\ } \capid^{-1}(i)$ \\
\midrule
?i :component ?a                        & $a = \mytext{proj}_2 {\ } \capid^{-1}(i)$ \\
\midrule
?i :component\_capabilities ?f          & $f \in \quaindcomp_{\fdginstance_0}(\interfaceid^{-1}(i))$ \\
\midrule
?i :channel\_capabilities ?f            & $f = \quaindchannel_{\fdginstance_0}(\interfaceid^{-1}(i))$ \\
\bottomrule
\end{tabular}
\caption{The semantic mapping between KGs and {\fdgname}s, with $\phi = \{ \mytext{Component} \mapsto \component, \mytext{Capability} \mapsto \capability, \mytext{ResourceType} \mapsto \resource, \mytext{Channel} \mapsto \channel, \mytext{InputInterface} \mapsto \iinterface, \mytext{OutputInterface} \mapsto \ointerface \}$.}\label{tab:sparql_to_fdg}
\end{table}

It is important to note that the state of the KG represents \textbf{only} the initial system configuration. None of the dynamic features of the domain, i.e. the changes in situations that events bring about, are encoded in the KG. The {\fdgname} constructed from the KG is the initial situation, denoted $\fdginstance_0$.

\autoref{lst:sparql} shows the SPARQL query \textit{Q} that retrieves components, functions and resource exchanges from a KG. We may interpret it, in plain English, as
\begin{quote}
    \textit{``What are all component capabilities that exchange resources through interfacing with channel capabilities?''}
\end{quote}

The formal semantics of SPARQL are given in \cite{Perez2006,perez2009semantics,salas_hogan}. The result mapping $\mu: V \rightarrow I$, is defined for all variables that we \texttt{SELECT} in \autoref{lst:sparql}.

We can evaluate \textit{Q} over an instantiated KG, which yields a result set $\Omega \in \mathcal{P}(\mu)$:
\[
[\![Q]\!]_{\mytext{KG}} = \Omega
\]

Where all selected variables are bound, and thus
\[
\forall \mu \in \Omega. {\ } \mytext{dom} {\ } \mu = \{ \qmark \alpha, \qmark f, \qmark c, \qmark k, \qmark r, \qmark h, \qmark \zeta, \qmark g, \qmark i, \qmark o \}
\]

\vspace{0.5em}
\noindent \begin{minipage}{0.98\columnwidth}
    \lstset{style=sparql}
\begin{lstlisting}[language=sparql, frame=single,  captionpos={b}, caption={The SPARQL query \textit{Q}.}, label={lst:sparql}]
PREFIX : <...>
PREFIX rdf: <...>
SELECT ?alpha ?f ?c ?k ?r ?h ?zeta ?g
 ?inputinterface ?outputinterface
WHERE{
 ?r rdf:type :ResourceType.

 ?alpha rdf:type :Component.
  ?alpha :has ?f.
   ?f rdf:type :Capability.
   ?f :consumes ?r.

 ?zeta rdf:type :Component.
  ?zeta :has ?g.
   ?g rdf:type :Capability.
   ?g :produces ?r.

 ?inputinterface rdf:type :InputInterface.
  ?inputinterface :component ?alpha.
  ?inputinterface :channel ?c.
  ?inputinterface :component_capabilities ?f.
  ?inputinterface :channel_capabilities ?k.
    
 ?c rdf:type :Channel.
  ?c :accommodates ?r.
   ?c :has ?k.
    ?k rdf:type :Capability.
    ?k :produces ?r.
   ?c :has ?h.
    ?h rdf:type :Capability.
    ?h :consumes ?r.
    
 ?outputinterface rdf:type :OutputInterface.
  ?outputinterface :channel ?c.
  ?outputinterface :component ?zeta.
  ?outputinterface :channel_capabilities ?h.
  ?outputinterface :component_capabilities ?g.
}
\end{lstlisting}
\end{minipage}

\subsection{{\fdgname} construction}

\newcommand{\partialcap}{\partial \capability}

We introduce a function $T : \mathcal{P}(\mu) \rightarrow \fdginstance_0$ that constructs {\fdgnamearticle} {\fdgname} from a result set over $\mytext{KG}$. We then build the initial {\fdgname} $\fdginstance_0$ from $\Omega$ in parts.

Let the set of partial capabilities be
\begin{multline*}
\partialcap = \bigcup_{\mathclap{\qmark x \in \{ \qmark f, \qmark h \}}} \left\{ (\mu(\qmark x), \{ \mu(\qmark r) \}, \emptyset) {\ } | {\ } \mu \in \Omega \right\} \cup\\
\bigcup_{\mathclap{\qmark x \in \{ \qmark g, \qmark k \}}} \left\{ (\mu(\qmark x), \emptyset, \{ \mu(\qmark r) \}) {\ } | {\ } \mu \in \Omega \right\}
\end{multline*}

Let $\mytext{CapabilityId}_{\fdginstance_0}$ be the set of capability IDs in $\fdginstance_0$:
\[
\mytext{CapabilityId}_{\fdginstance_0} = \{ \mytext{proj}_1 {\ } x {\ } | {\ } x \in \partialcap \}
\]

Let $R_{\mytext{in}}(\mytext{id})$ be the set of resource consumptions of the capability identified by $\mytext{id}$:
\[
R_{\mytext{in}}(\mytext{id}) = \bigcup \left\{ \mytext{proj}_2 {\ } x {\ } | {\ } x \in \partialcap. {\ } \mytext{proj}_1 {\ } x = \mytext{id} \right\}
\]

Similarly, let $R_{\mytext{out}}(\mytext{id})$ be the set of resource productions of the capability identified by $\mytext{id}$:
\[
R_{\mytext{out}}(\mytext{id}) = \bigcup \left\{ \mytext{proj}_3 {\ } x {\ } | {\ } x \in \partialcap. {\ } \mytext{proj}_1 {\ } x = \mytext{id} \right\}
\]

The set of capabilities $\capability_{\fdginstance_0}$ is then:
\[
\capability_{\fdginstance_0} = \left\{ (\mytext{id}, R_{\mytext{in}}(\mytext{id}), R_{\mytext{out}}(\mytext{id})) {\ } | {\ } \mytext{id} \in \mytext{CapabilityId}_{\fdginstance_0} \right\}
\]

By \autoref{post:vivacitas}, there exists a function for each capability such that this function manifests that unique capability. Let the relation $\vartheta : \mytext{CapabilityId}_{\fdginstance_0} \rightarrow \mytext{FunctionId}_{\fdginstance_0}$ return this unique manifestation. Let the set of functions in $\fdginstance_0$ permit a bijection over $\vartheta$ w.r.t. the set of capability IDs $\mytext{CapabilityId}_{\fdginstance_0}$:
\begin{equation*}
    \function_{\fdginstance_0} = \left\{ \vartheta(\mytext{id}) {\ } | {\ } \mytext{id} \in \mytext{CapabilityId}_{\fdginstance_0} \right\}
\end{equation*}

Let the set of faults $\fault_{\fdginstance_0} = \emptyset$ be the empty set.

\newcommand{\partialcomp}{\partial \component}
We can now construct the set of partial components $\partialcomp$:
\[
\partialcomp = \left\{ (\mu(\qmark \alpha), \{ \capid^{-1}_{\fdginstance_0}(\mu(\qmark f)) \}, \{ \vartheta(\mu(\qmark f)) \}, \emptyset) {\ } | {\ } \mu \in \Omega \right\}
\]

Let $\mytext{ComponentId}_{\fdginstance_0}$ be the set of component IDs in $\fdginstance_0$:
\[
\mytext{ComponentId}_{\fdginstance_0} = \{ \mytext{proj}_1 {\ } x {\ } | {\ } x \in \partialcomp \}
\]

Let $C(\mytext{id})$ be the set of capabilities of the component identified by $\mytext{id}$:
\[
C(\mytext{id}) = \bigcup \left\{ \mytext{proj}_2 {\ } x {\ } | {\ } x \in \partialcomp. {\ } \mytext{proj}_1 {\ } x = \mytext{id} \right\}
\]

Similarly, let $F(\mytext{id})$ be the set of component functions of the capability identified by $\mytext{id}$:
\[
F(\mytext{id}) = \bigcup \left\{ \mytext{proj}_3 {\ } x {\ } | {\ } x \in \partialcomp. {\ } \mytext{proj}_1 {\ } x = \mytext{id} \right\}
\]

And the set of components $\component_{\fdginstance_0}$ is then:
\[
\component_{\fdginstance_0} = \left\{ (\mytext{id}, C(\mytext{id}), F(\mytext{id}), \emptyset) {\ } | {\ } \mytext{id} \in \mytext{ComponentId}_{\fdginstance_0} \right\}
\]

\newcommand{\partialchannel}{\partial \channel}
Next, we construct the set of partial channels $\partialchannel$. The process is similar as for components:
\begin{multline*}
\partialchannel = \left\{ (\mu(\qmark c), \{ \capid^{-1}_{\fdginstance_0}(\mu(\qmark h)), \capid^{-1}_{\fdginstance_0}(\mu(\qmark k)) \}, \right.\\
\left. \{ \vartheta(\mu(\qmark h)), \vartheta(\mu(\qmark k)) \}, \emptyset) {\ } | {\ } \mu \in \Omega \right\}
\end{multline*}

Let $\mytext{ChannelId}_{\fdginstance_0}$ be the set of channel IDs in $\fdginstance_0$:
\[
\mytext{ChannelId}_{\fdginstance_0} = \{ \mytext{proj}_1 {\ } x {\ } | {\ } x \in \partialchannel \}
\]

Let $C'(\mytext{id})$ be the set of capabilities of the channel identified by $\mytext{id}$:
\[
C'(\mytext{id}) = \bigcup \left\{ \mytext{proj}_2 {\ } x {\ } | {\ } x \in \partialchannel. {\ } \mytext{proj}_1 {\ } x = \mytext{id} \right\}
\]

Similarly, let $F'(\mytext{id})$ be the set of channel functions of the capability identified by $\mytext{id}$:
\[
F'(\mytext{id}) = \bigcup \left\{ \mytext{proj}_3 {\ } x {\ } | {\ } x \in \partialchannel. {\ } \mytext{proj}_1 {\ } x = \mytext{id} \right\}
\]

And the set of channels $\channel_{\fdginstance_0}$ is then:
\[
\channel_{\fdginstance_0} = \left\{ (\mytext{id}, C'(\mytext{id}), F'(\mytext{id}), \emptyset) {\ } | {\ } \mytext{id} \in \mytext{ChannelId}_{\fdginstance_0} \right\}
\]

By \autoref{eq:capabilities_partition}, \autoref{eq:functions_partition} and \autoref{eq:faults_partition}, the set of capabilities $\capability_{\fdginstance_0}$, functions $\function_{\fdginstance_0}$ and faults $\fault_{\fdginstance_0} = \emptyset$ is simply the union of those that inhere in (and manifest modes of) components and capabilities:

\begin{align*}
    \capability_{\fdginstance_0} = \left( \bigcup_{\alpha \in \component_\fdginstance} \capability^\alpha_\fdginstance \right) &\cup \left( \bigcup_{c \in \channel_\fdginstance} \capability^c_\fdginstance \right)\\
    \function_{\fdginstance_0} = \left( \bigcup_{\alpha \in \component_\fdginstance} \function^\alpha_\fdginstance \right) &\cup \left( \bigcup_{c \in \channel_\fdginstance} \function^c_\fdginstance \right)
\end{align*}

The sets output ($\ointerface_{\fdginstance_0}$) and input ($\iinterface_{\fdginstance_0}$) interfaces are constructed as follows:
\begin{align*}
    \ointerface_{\fdginstance_0} &= \left\{ (\mu(\qmark \zeta), \mu(\qmark c), \mu(\qmark r)) {\ } | {\ } \mu \in \Omega \right\}\\
    \iinterface_{\fdginstance_0} &= \left\{ (\mu(\qmark c), \mu(\qmark \alpha), \mu(\qmark r)) {\ } | {\ } \mu \in \Omega \right\}
\end{align*}
Such that
\begin{align*}
    \interfaceid^o_{\fdginstance_0}((\mu(\qmark \zeta), \mu(\qmark c), \mu(\qmark r))) &= \mu(\qmark {\,} \mytext{outputinterface})\\
    \interfaceid^o_{\fdginstance_0}((\mu(\qmark c), \mu(\qmark \alpha), \mu(\qmark r))) &= \mu(\qmark {\,} \mytext{inputinterface})
\end{align*}

Let $\mytext{OutputInterfaceId}_{\fdginstance_0}$ be the set of output interface IDs in $\fdginstance_0$:
\[
\mytext{OutputInterfaceId}_{\fdginstance_0} = \left\{ \mu(\qmark {\,} \mytext{outputinterface}) {\ } | {\ } \mu \in \Omega \right\}
\]
And let $\mytext{InputInterfaceId}_{\fdginstance_0}$ be the set of input interface IDs in $\fdginstance_0$:
\[
\mytext{InputInterfaceId}_{\fdginstance_0} = \left\{ \mu(\qmark {\,} \mytext{inputinterface}) {\ } | {\ } \mu \in \Omega \right\}
\]

\newcommand{\partialquas}{\partial \quaindcomp}
The partial component capabilities that the output and input interface identified by $\mytext{id}$ aggregate are:
\begin{align*}
    \partialquas^o(\mytext{id}) &= \left\{ \mu(\qmark g) {\ } | {\ } \mu \in \Omega. {\ } \mu(\qmark {\,} \mytext{outputinterface}) = \mytext{id} \right\}\\
    \partialquas^i(\mytext{id}) &= \left\{ \mu(\qmark f) {\ } | {\ } \mu \in \Omega. {\ } \mu(\qmark {\,} \mytext{inputinterface}) = \mytext{id} \right\}
\end{align*}

Whereby the functions that return the component capabilities aggregated by an output or input interface are given by:
\begin{align*}
    \quaindcomp^o_{\fdginstance} &= \left\{ \interfaceid^{o,-1}_{\fdginstance_0}(\mytext{id}) \mapsto \partialquas^o(\mytext{id}) {\ } | {\ } \mytext{id} \in \mytext{OutputInterfaceId}_{\fdginstance_0} \right\}\\
    \quaindcomp^i_{\fdginstance} &= \left\{ \interfaceid^{i,-1}_{\fdginstance_0}(\mytext{id}) \mapsto \partialquas^i(\mytext{id}) {\ } | {\ } \mytext{id} \in \mytext{InputInterfaceId}_{\fdginstance_0} \right\}
\end{align*}

And the channel capabilities they aggregate are:
\begin{align*}
    \quaindchannel^o_{\fdginstance_0} &= \left\{ \interfaceid^{o,-1}_{\fdginstance_0}(\mu(\qmark {\,} \mytext{outputinterface})) \mapsto \mu(\qmark h) {\ } | {\ } \mu \in \Omega \right\}\\
    \quaindchannel^i_{\fdginstance_0} &= \left\{ \interfaceid^{i,-1}_{\fdginstance_0}(\mu(\qmark {\,} \mytext{inputinterface})) \mapsto \mu(\qmark k) {\ } | {\ } \mu \in \Omega \right\}
\end{align*}

Finally, let $\resource = \left\{ \mu(?r) {\ } | {\ } \mu \in \Omega \right\}$ and $\expectation_{\fdginstance_0} = \emptyset$.

The constructed {\fdgname} $\gamma_0$ is then:
\[
(\capability_{\fdginstance_0}, \function_{\fdginstance_0}, \fault_{\fdginstance_0}, \component_{\fdginstance_0}, \channel_{\fdginstance_0}, \ointerface_{\fdginstance_0}, \iinterface_{\fdginstance_0}, \expectation_{\fdginstance_0}) \in \fdg
\]
which obtains at some arbitrary time point t, such that
\[
\fdginstance_0 = (\gamma_0, t) : \fdg
\]

And finally,
\[
T(\Omega) = \fdginstance_0
\]


\section{The {\fdgname}-FT Transformation}\label{sec:transformation}

\newcommand{\target}{f}
\newcommand{\source}{g}

In this section we first introduce the formal FT semantics. We then provide a semantic mapping between the {\fdgname} concepts and FT concepts introduced in the previous sections. Lastly, we define the algorithm that performs the synthesis of FTs from {\fdgname} instances.

\subsection{Fault Tree Semantics}\label{sec:ft_def}
We present the formal structure and semantics of FTs following the summarization in \cite{Ruijters2015}.

\begin{definition}
\label{def:ft}

A static fault tree (FT) is a 4-tuple $FT = \langle BE, G, T, I \rangle$, consisting of the following components.

\begin{itemize}
    \item[] $BE$ is the set of basic events.
    \item[] $G$ is the set of gates, with $BE \cap G = \emptyset$. We write $E = BE \cup G$ for the set of elements.
    \item[] $T : G \rightarrow \mytext{GateTypes}$ is a function that describes the type of each gate.
    \item[] $I : G \rightarrow \mathcal{P}(E)$ describes the inputs of each gate.
\end{itemize}

\end{definition}

\newcommand{\ftfail}{S}
\begin{definition}\label{def:ft_semantics}

The semantics of a fault tree \mytext{FT} is a function
\[
\pi_{\mytext{FT}} : \mathcal{P}(BE) \times E \rightarrow \{0,1\}
\]
where $\pi_{\mytext{FT}}(\ftfail, e)$ indicates whether $e$ fails given the set {\ftfail } of failed BEs. It is defined as follows.

\begin{enumerate}
    \item \textbf{For $e \in BE$}, $\pi_{\mytext{FT}}(\ftfail, e) = e \in \ftfail$.
 
    \item \textbf{For $g \in G$ and $T(g) = \text{And}$}, let 
    \[
    \pi_{\mytext{FT}}(\ftfail, g) = \bigwedge_{x \in I(g)} \pi_{\mytext{FT}}(\ftfail, x).
    \]

    \item \textbf{For $g \in G$ and $T(g) = \text{Or}$}, let 
    \[
    \pi_{\mytext{FT}}(\ftfail, g) = \bigvee_{x \in I(g)} \pi_{\mytext{FT}}(\ftfail, x).
    \]


\end{enumerate}

\end{definition}

\subsection{Semantic mapping}\label{sec:fdg_ft_semantic_mapping}
We present the semantic partial mapping to show the correspondence between {\fdgname} and FT concepts. Given the greater specificity of concepts in {\fdgname}s, we can only propose this partial mapping in one direction. After all, we can construct FTs that do not pertain to systems architectures at all (e.g. an FT about social inequality). That is to say: the partial mapping between {\fdgname} concepts and FT concepts is necessarily surjective. Additionally, not all {\fdgname} concepts are mapped to an FT concept (e.g. the {\fdgname} concept \texttt{Fault} does not have a related concept in FTs). \autoref{tab:semantic_mapping_fdg_ft} lists the semantic partial mapping.

\begin{table}[]
\centering
\begin{tabular}{ll}
\textbf{{\fdgname} concept}   & \textbf{FT concept} \\
\toprule
Capability of Interest & Top-Level Event     \\
\midrule
Error Propagation      & Intermediate Event  \\
\midrule
Fault Activation       & Basic Event         \\
\bottomrule
\end{tabular}
\caption{The semantic mapping from {\fdgname} concepts to FT concepts.}\label{tab:semantic_mapping_fdg_ft}
\end{table}


There are two distinct causes for errors in {\fdgname}s. An error can either be the direct result of a fault activation in the bearing component, or the indirect result of error propagation that a capability is dependent on.

Fault trees distinguish three types of events; basic events (BEs), which occur spontaneously, intermediate events (IEs), which are caused by one or more other events\cite{Ruijters2015}, and top-level events (TLEs), which are at the root of an FT. IEs and TLEs are associated with a gate type that models the logical behavior of failure propagation.

\mypar{Top-level events}
An FT's TLE is the event of interest in FTA. This is the undesired event we wish to calculate metrics over.
Similarly, in {\fdgname} evaluation, we choose a capability of interest and place an expectation on it. This comes down to the same thing; the event that this capability is not (or no longer) manifested is the ``undesired event''. In {\fdgname} terminology, an \textit{error}.
It is therefore natural to map {\fdgname} capabilities of interest to top-level events.

Recall that by placing an expectation on a capability of interest, we transitively also place expectations on the capabilities that it is functionally dependent on. This results in a DAG of potential errors. Each vertex of this DAG (excluding the potential error over the capability of interest) gets mapped to IEs and/or BEs.

\mypar{Intermediate events}
Intermediate events are associated with a gate type. In standard FT's, the intermediate event with an \texttt{OR} gate type is observed when any of its input events occur. Intermediate events with an \texttt{AND} gate type are observed when \textit{all} of its input events occur.
In {\fdgname}s we distinguish two scenarios for which, if either of them occur, an error is observed over a capability $f$:
\begin{enumerate}[label=\roman*.]
    \item Fault activation over $f$, or;
    \item Error propagation preventing access to any of $f$'s prerequisite resources.
\end{enumerate}
This maps the error over $f$ to an \texttt{OR} type intermediate event with two inputs. Scenario (i) maps to a basic event, and scenario (ii) maps each prerequisite resource type to a distinct \texttt{AND} gate.

Scenario (ii) also pertains to the notion of minimal configuration fragments w.r.t. a prerequisite resource type $r$.
By \cref{eq:fdepa_simplified}, \cref{eq:fdepb_simplified} and \cref{eq:fdepc_simplified}, the size of a minimal configuration fragment is always one. Therefore we map such errors to child events directly, instead of mapping the minimal configuration fragment to an intermediate \texttt{OR} gate with a single child event representing the error over the minimal configuration fragment's single element. In the general case, such an \texttt{OR} gate should be generated for minimal configuration fragments.


In FTs, \texttt{AND} gates describe that all of the input events must happen for the particular intermediate event to be observed.
In {\fdgname}s, an error over at least one element (recall that in this work, there is only one) of \textit{each} minimal configuration fragment over $f$ for $r$ then means that we observe an error over $f$. We therefore map the set of minimal configuration fragments over $f$ for $r$ to a single \texttt{AND} gate.

\mypar{Basic events}
An FT's basic events are stochastically independent failure events. They are in a causal relationship w.r.t. the TLE of the fault tree.
In {\fdgname}s, fault creation is a spontaneous event (i.e. it is not caused by any endogenous part of the system).
In the focus of this work, a fault creation event is always a direct cause of a fault activation event. This induced fault activation is exclusively dependent on the same object as the fault creation that caused it. All components are susceptible to fault creation, and the fault activation that is, in this work, necessarily triggered by the fault creation over a capability is mapped to the FT concept of the basic event.

It is important to note that, since capabilities cannot be repaired in {\fdgname}s, there can occur \textbf{at most} one fault activation over each capability.

\begin{remark}[Fault Tree Events as Event Types]
    The fault tree events we synthesize are in fact based on event \textit{types}, rather than particular events. Qualitative FTA, for instance by use of the structure function, utilizes this distinction in a subtle way: membership in the set of failed basic events $S$ represents the \textit{occurrence} of a particular event of a certain type. This is a fundamental insight about the fault trees we synthesize from {\fdgname}s, and indeed fault trees analysis in general.
    Furthermore, this realization opens the door to a future ontological exploration of the quantitative analysis of FTs, e.g. in line with the notion of event type likelihood in COVER~\cite{sales_cover}.
\end{remark}

\begin{remark}
    According to Veseley et al., creating an FT with no intermediate labels between its gate events is \textit{``indicative of sloppy analysis''}~\cite{vesely1981fault}. Although the FT in \autoref{fig:transformation_illustration_ft} does not show event labels, meaningful labels can be determined as follows:
    \begin{itemize}
        \item A fault activation event over some capability $g$ has the exclusive participation of the bearing object $o$. The label that we associate with the basic event this fault activation is mapped to may take the form \textit{``Fault activation over $g$ in $o$''}.
        \item For an error over a channel $c$'s input capability $h$ that consumes a resource type $r$, a meaningful label for the synthesized \texttt{OR} gate may take the shape \textit{``$c$ fails to consume $r$''}.
        \item For an error over that channel's output capability $k$ that produces the resource type $r$, the label of the corresponding \texttt{OR} may be \textit{``$c$ fails to provide $r$''}.
        \item For an error over a component capability $f$'s access to a prerequisite resource $r$, one may choose to label its \texttt{AND} gate \textit{``$f$ loses sufficient access to $r$''}.
        \item For an error over a component $\alpha$'s output capability (or a capability of interest) $\ell$, a label like \textit{``$\alpha$ fails to $\ell$''} is appropriate for the generated \texttt{OR} gate.
    \end{itemize}
\end{remark}




\newcommand{\transform}{\tau}
\subsection{Transformation Definition}

We introduce a transformation function $\transform$, which synthesizes an FT from a given {\fdgname}. This transformation assumes that the given {\fdgname} is valid.



\newcommand{\ft}{\text{FT}}
\newcommand{\parent}{\gamma}
\newcommand{\getbes}{\mytext{BEs}}
\newcommand{\getgates}{\mytext{Gates}}
\newcommand{\mbe}{\text{be}}
\newcommand{\mor}{\text{or}}
\newcommand{\mand}{\text{and}}
\begin{definition}[The {\fdgname} to FT transformation function]\label{def:transformation}
    The function $\transform: \fdg \times \capability_\fdginstance \rightarrow \text{FT}$ transforms {\fdgnamearticle} {\fdgname} to an FT.
\end{definition}

\autoref{alg:fdg_ft_transformation} shows the precise definition of $\transform$ in pseudo code.

\begin{algorithm*}[t]
\small

\caption{The {\fdgname} to FT transformation function $\transform$}\label{alg:fdg_ft_transformation}

\begin{algorithmic}[1]
\footnotesize

\State $\mathcal{L} : \mytext{Gate} \cup \{\varepsilon\}$
\State $P \subseteq \mathcal{P}(\mathcal{L})$ \Comment{Encountered FT gates (recurs. guard)}

\vspace{0.5em}

\Function{\getbes}{$\langle \mytext{BE}, G, T, I \rangle :\!\mytext{FaultTree}$}
    \State \Return $\mytext{BE}$
\EndFunction

\vspace{0.5em}

\Function{\getgates}{$\langle \mytext{BE}, G, T, I \rangle :\!\mytext{FaultTree}$}
    \State \Return $G$
\EndFunction

\vspace{0.5em}

\Function{I}{$\langle \mytext{BE}, G, T, I \rangle :\!\mytext{FaultTree}$}
    \State \Return $I$
\EndFunction

\vspace{0.5em}

\Function{$\transform$}{$\fdginstance : \fdg, \target : \function$}
    \State $\ft \gets \mytext{new } \mytext{FaultTree}$
    \State $\transform'(\fdginstance, \emptyset, \target, \varepsilon, \ft)$
    \State \Return \ft
\EndFunction

\vspace{0.5em}

\Function{$\transform'$}{$\fdginstance :\!\fdg, P, \target :\!\capability, \parent :\!\mathcal{L}, \ft :\!\mytext{FaultTree}$}


    \If{$\mor_{\target} \in P$}
        \State Generate basic event $\mbe_{\target}$ and place it under $\parent$ such that 
        \[
            \mbe_{\target} \in \getbes(\ft); {\ } \mbe_{\target} \in I(\parent)
        \]
        \State \Return
    \EndIf

    \State Generate OR gate $\mor_{\target}$ such that
    \[
        \mor_{\target} \in \getgates(\ft)
    \]
    \State $P' \gets P \cup \{\mor_{\target}\}$ 

    \If{$\parent = \varepsilon$}
        \State $\mor_{\target}$ becomes the top-level event, such that
        \[
            \getgates(\ft) = \{ \mor_{\target} \}
        \]
    \Else
        \State Place $\mor_{\target}$ under $\parent$ such that
        \[
            \mor_{\target} \in I(\parent)
        \]
    \EndIf

    \State Generate basic event $\mbe_{\target}$ and place it under $\mor_{\target}$ such that 
    \[
        \mbe_{\target} \in \getbes(\ft); {\ } \mbe_{\target} \in I(\mor_{\target})
    \]


    \ForAll{$r \in \preq_\fdginstance(\target)$} \Comment{Prerequisite resources for $\target$}
        
        \State Generate AND gate ${}_{r}\mand_{\target}$ and place it under $\mor_{\target}$ such that
        \[
            {}_{r}\mand_{\target} \in \getgates(\ft); {\ } {}_{r}\mand_{\target} \in I(\mor_{\target})
        \]

        \ForAll{$k \in \fdep({\target}, r)$} \Comment{Channel output capabilities that provide ${\target}$ with $r$}

            \State Generate OR gate ${}_{k}\mor_{f}$ and place it under ${}_{r}\mand_{\target}$ such that
            \[
                {}_{k}\mor_{f} \in \getgates(\ft); {\ } {}_{k}\mor_{f} \in I({}_{r}\mand_{\target})
            \]
            
            \State Generate basic event $\mbe_{k}$ and place it under ${}_{k}\mor_{f}$ such that 
            \[
                \mbe_{k} \in \getbes(\ft); {\ } \mbe_{k} \in I({}_{k}\mor_{f})
            \]

            \State Obtain the channel's input capability $h = \fdepb_\fdginstance(f, r)$ \Comment{Channel input capability that provides $k$ with $r$}

            \State Generate OR gate ${}_{h}\mor_{k}$ and place it under ${}_{k}\mor_{f}$ such that
            \[
                {}_{h}\mor_{k} \in \getgates(\ft); {\ } {}_{h}\mor_{k} \in I({}_{k}\mor_{f})
            \]
            
            \State Generate basic event $\mbe_{h}$ and place it under ${}_{h}\mor_{k}$ such that 
            \[
                \mbe_{h} \in \getbes(\ft); {\ } \mbe_{h} \in I({}_{h}\mor_{k})
            \]

            \ForAll{$g \in \fdepc(h, r)$} \Comment{Component capabilities that provide $h$ with $r$}
                
                \State Generate OR gate ${}_{g}\mor_{h}$ and place it under ${}_{h}\mor_{k}$ such that
                \[
                    {}_{g}\mor_{h} \in \getgates(\ft); {\ } {}_{g}\mor_{h} \in I({}_{h}\mor_{k})
                \]
                
                \State Generate basic event $\mbe_{g}$ and place it under ${}_{g}\mor_{h}$ such that 
                \[
                    \mbe_{g} \in \getbes(\ft); {\ } \mbe_{g} \in I({}_{g}\mor_{h})
                \]

                \State $\transform'(\fdginstance, P', g, {}_{g}\mor_{h}, \ft)$
            \EndFor

        \EndFor
    \EndFor

\EndFunction

\end{algorithmic}
\end{algorithm*}


Recall the print head example from \autoref{fig:structure_function_example}. \autoref{fig:transformation_illustration} illustrates the transformation of that system into a fault tree $\ft = \transform(\fdginstance_0, f)$.

\begin{figure}[t!]
    \centering
    \begin{subfigure}[t]{0.55\linewidth}
        \centering
        
        \includegraphics[]{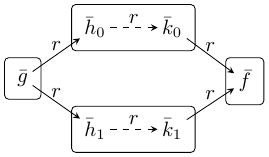}
                
        \caption{Example of the print head in its ground state $\fdginstance_0$.}\label{fig:transformation_example_ground_state}
    \end{subfigure}%
    ~
    \begin{subfigure}[t]{0.45\linewidth}
        \centering
        \includegraphics[width=\linewidth]{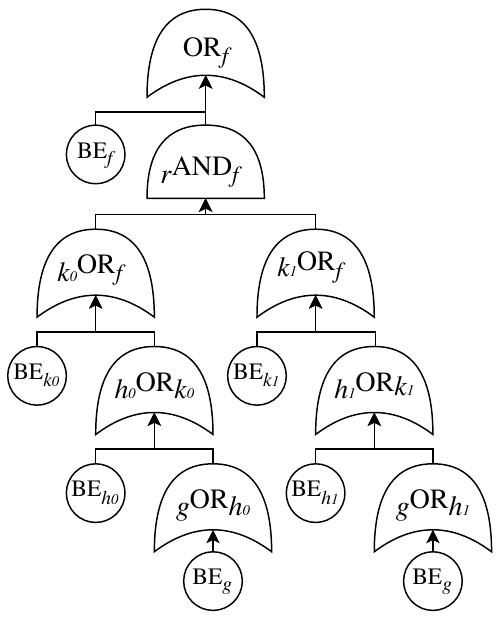}
        \caption{The FT synthesized from $\fdginstance_0$.}\label{fig:transformation_illustration_ft}
    \end{subfigure}
    \caption{Fault tree synthesized from the print head example.}\label{fig:transformation_illustration}
\end{figure}

\begin{figure}[ht!]

\end{figure}

\section{Transformation properties}\label{sec:transformation_properties}

In this section we prove the soundness and completeness of the synthesis algorithm presented in the previous section, with respect to error semantics.

Intuitively, the transformation is sound iff the failure semantics of the synthesized fault tree are not spurious. I.e. all positive evaluations of the structure function of the synthesized FT represent a positive evaluation of the {\fdgname}'s structure function.

Conversely, our transformation is complete iff it preserves all failure semantics of the original system. All failure modes that can occur in the {\fdgname} model must be represented by a corresponding failure mode in the generated fault tree.

The transformation is sound \textit{and} complete iff the synthesized fault tree represents the exact error semantics of the {\fdgname} it was generated from.

\begin{theorem}[The soundness and completeness of {$\transform$}]\label{thm:completeness}

Let $\target \in \capability_{\fdginstance_0}$ be a capability. Let $\fdginstance_0, \fdginstance_n : \fdg$ be valid {\fdgname}s. Then the transformation $\transform$ is sound and complete iff
\[
\forall \target \in \capability. {\ } [\![ \target ]\!]^{\fdginstance_n}_{\fdginstance_0} \iff \pi_{\transform(\fdginstance_0, \target)}(S, \target)
\]
where by \autoref{lem:s_from_e}, $S$ is derived from the set of erroneous components $E_{\fdginstance_n}$.
\end{theorem}

\begin{lemma}[Derivation of $S$ from $E_{\fdginstance_n}$]\label{lem:s_from_e}
    In FTA, $S$ denotes the set of failed basic events. The set $S$ can be derived from $E' \subseteq E_{\fdginstance_n}$ s.t.
    \[
        \forall x \in E'. {\ } \exists ! b_x \in S. {\ } b_x \leftarrow x
    \]
    Note that basic event generation is idempotent in \autoref{alg:fdg_ft_transformation}. Furthermore, $\transform$ does not generate spurious basic events:
    \[
        \forall b_x \in S. {\ } \exists ! x \in E'. {\ } b_x \leftarrow x
    \]
    I.e. there exists a partial bijection from $E_{\fdginstance_n}$ to $S$.
\end{lemma}

\begin{proof}[Proof sketch for \autoref{lem:s_from_e}]
    Suppose $\fdginstance_0, \fdginstance_n : \fdg$ are valid {\fdgname}s, such that $\fdginstance_0 \preceq \fdginstance_n$. Let $f \in \capability_{\fdginstance_n}$ be a capability.

    We use the symbol $\faultcreation$ as shorthand notation for the existence of a fault creation event;
    \[
    \faultcreation^{\fdginstance_n}_{\fdginstance_0} \langle \mytext{id}_f \rangle \iff \exists \fdginstance_{n-i}, \fdginstance_{n-i}' : \fdg. {\ } \fdginstance_{n-i} \overset{\faultcreation_{F}}{\looparrowright} \fdginstance_{n-i}'
    \]
    where $\fdginstance_0 \prec \fdginstance_{n-i}$ and $\fdginstance_{n-i}' \preceq \fdginstance_{n}$, and $\mytext{id}_{\faultinstance{f}} \in F$.

    By \cref{eq:fault_activation_under_structure_function}, we have that
    \[
        \faultcreation^{\fdginstance_n}_{\fdginstance_0} \langle \mytext{id}_{x} \rangle \implies x \in E_{\fdginstance_n}
    \]

    It is obvious that the set of all capabilities whose fault activation is a term in $[\![ f ]\!]^{\fdginstance_n}_{\fdginstance_0}$ is then a subset $E' \subseteq E_{\fdginstance_n}$.

    Furthermore, from \autoref{alg:fdg_ft_transformation}, we have that $\transform$ maps the fault activation over such a capability $x \in E'$ to a unique basic event $b_x \in \mytext{BE}$\footnote{Note that the generation of basic events in $\transform$ is idempotent.}. Such a derived basic event is denoted $b_x \leftarrow x$. I.e.:
    \[
        \exists ! b_x \in S. {\ } b_x \leftarrow \faultcreation^{\fdginstance_n}_{\fdginstance_0} \langle \mytext{id}_{x} \rangle
    \]

    By simple composition of these two results, we obtain
    \[
        \exists E' \subseteq E_{\fdginstance_n}. {\ } \forall b_x \in S. {\ } \exists !x \in E'. {\ } b_x \leftarrow x
    \]


\end{proof}


\subsection{Inductive proof}
We proceed by induction over the size of the {\fdgname}. We split the induction in three cases. Case 1 is the addition of a new prerequisite resource type through existing connected channels and components. Case 2 is the addition of a channel with an existing resource type between already connected components. Case 3 is the addition of a new component that provides of an existing resource type to an already connected channel.

Let $\fdginstance_0 : \fdg$ be a valid {\fdgname}. Let $\ft = \transform(\fdginstance_0, \target)$. Let $I = |\resource_{\fdginstance_0}|$, $M = |\channel_{\fdginstance_0}|$, and $N = |\capability^\component_{\fdginstance_0}|$ respectively denote the number of resources, channels and component capabilities\footnote{We define $\capability^\component_{\fdginstance_0} = \{ f {\ } | {\ } f \in \capability_{\fdginstance_0}. {\ } \exists \alpha \in \component_{\fdginstance_0}. {\ } f \in \capability^\alpha_{\fdginstance_0} \}$.} in the {\fdgname} $\fdginstance_0$. Let $K = I + M + N$ be the finite size of $\fdginstance_0$.

To prove that $[\![ \target ]\!]^{\fdginstance_n}_{\fdginstance_0} \iff \pi_{\ft}(S, \target)$ holds for {\fdgnamearticle} {\fdgname} of any size $K$.

\begin{figure}
    \centering
    \includegraphics[width=0.9\linewidth]{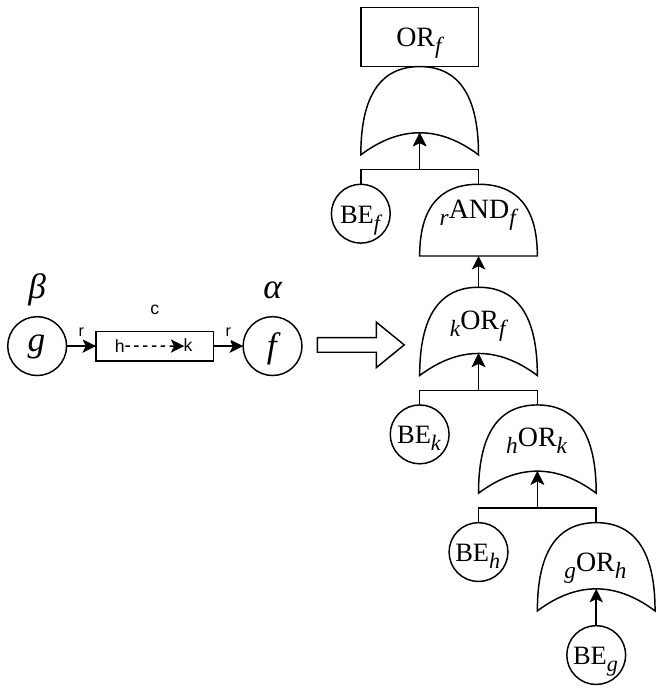}
    \caption{An illustration of the base case.}\label{fig:ip_base_case}
\end{figure}

\begin{lemma}[Base case]\label{lem:ip_base_case}
    Observe the {\fdgname} of size $I + M + N = 4$ in \autoref{fig:ip_base_case}, showing two component capabilities and a single channel between them. We argue that this interaction is the minimal {\fdgname} that is meaningful for analysis (i.e. the analysis of a single component with no inputs or outputs is trivial). The capability $\target \in \capability_{\fdginstance_n}$ receives resource type $r$ from $g \in \capability_{\fdginstance_n}$ through channel $c \in \channel_{\fdginstance_n}$.
    
    Let $\ft = \transform(\fdginstance_0, f)$ be the fault tree synthesized from the {\fdgname} $\fdginstance_0$. We then have a logical equivalence under the one-to-one correspondence between fault creation events and basic events:
    \[
        [\![ f ]\!]^{\fdginstance_{n}}_{\fdginstance_0} \iff \pi_\ft(S, f)
    \]
\end{lemma}

\begin{proof}[Base case]
    \mypar{{\fdgname} Derivation}
    We evaluate the structure function $[\![ f ]\!]^{\fdginstance_n}_{\fdginstance_0}$ for the existence of an error over the capability $f$:
    \begin{align*}\label{eq:base_case_structure_function_derivation}
        [\![ f ]\!]^{\fdginstance_{n}}_{\fdginstance_0} &\iff \faultcreation^{\fdginstance_n}_{\fdginstance_0} \langle \mytext{id}_f \rangle \vee [\![k_0]\!]^{\fdginstance_{n}}_{\fdginstance_0} \\
        [\![ k_0 ]\!]^{\fdginstance_{n}}_{\fdginstance_0} &\iff \faultcreation^{\fdginstance_n}_{\fdginstance_0} \langle \mytext{id}_{k_0} \rangle \vee \left( \exists \fdginstance_{n-j} : \fdg. {\ } [\![h_0]\!]^{\fdginstance_{n-j}}_{\fdginstance_0} \right) \\
        [\![h_0]\!]^{\fdginstance_{n-j}}_{\fdginstance_0} &\iff \faultcreation^{\fdginstance_{n-j}}_{\fdginstance_0} \langle \mytext{id}_{h_0} \rangle \vee [\![g]\!]^{\fdginstance_{n-j}}_{\fdginstance_0} \\
        [\![g]\!]^{\fdginstance_{n-j}}_{\fdginstance_0} &\iff \faultcreation^{\fdginstance_{n-j}}_{\fdginstance_0} \langle \mytext{id}_{g} \rangle
    \end{align*}
    
    Recall that by \autoref{eq:strcuture_function_implies_expectation}, expectation is vacuous for any capability that we evaluate through the structure function. For readability, we therefore discard the conjunctive term $\mytext{id}_f \in \expectation_{\fdginstance_n}$. Furthermore, recall expectation is transitive over functional dependency (a mechanism reflected by the structure function). E.g. $\mytext{id}_f \in \expectation_{\fdginstance_n} \implies \mytext{id}_{k_0} \in \expectation_{\fdginstance_n}$, also omitted for readability.
    
    Bottom-up substitution of the recursively defined equivalences yields a boolean formula with atomic propositions. E.g.:
    \[
    [\![h_0]\!]^{\fdginstance_{n-j}}_{\fdginstance_0} \iff \faultcreation^{\fdginstance_{n-j}}_{\fdginstance_0} \langle \mytext{id}_{h_0} \rangle \vee [\![g]\!]^{\fdginstance_{n-j}}_{\fdginstance_0}
    \]
    becomes:
    \[
    [\![h_0]\!]^{\fdginstance_{n-j}}_{\fdginstance_0} \iff \faultcreation^{\fdginstance_{n-j}}_{\fdginstance_0} \langle \mytext{id}_{h_0} \rangle \vee \faultcreation^{\fdginstance_{n-j}}_{\fdginstance_0} \langle \mytext{id}_{g} \rangle
    \]
    
    \mypar{FT derivation}
    Assume a set $S$ of failed basic events, derived from the set of erroneous components as defined in \autoref{thm:completeness}. Recursively applying the fault tree structure function $\pi$ over the fault tree $\ft = \transform(\fdginstance_0, \target)$ yields:
    \begin{align*}\label{eq:base_case_ft_structure_function_derivation}
        \pi_\ft(S, \mytext{OR}_f) &\iff \mytext{BE}_f \in S \vee \pi_\ft(S, {}_{r}\mytext{AND}_f) \\
        \pi_\ft(S, {}_{r}\mytext{AND}_f) &\iff \pi_\ft(S, {}_{k}\mytext{OR}_f) \\
        \pi_\ft(S, {}_{k}\mytext{OR}_f) &\iff \mytext{BE}_k \in S \vee \pi_\ft(S, {}_{h}\mytext{OR}_k) \\
        \pi_\ft(S, {}_{h}\mytext{OR}_k) &\iff \mytext{BE}_h \in S \vee \pi_\ft(S, {}_{g}\mytext{OR}_h) \\
        \pi_\ft(S, {}_{g}\mytext{OR}_h) &\iff \mytext{BE}_g \in S
    \end{align*}

    \mypar{To conclude}
    By \autoref{lem:s_from_e}, for any capability $x \in \capability_{\fdginstance_n}$ such that $\faultcreation \langle \capid_{\fdginstance_n}(x) \rangle$ is a term in $[\![ f ]\!]^{\fdginstance_{n}}_{\fdginstance_0}$, we have a one-to-one correspondence to a derived term $b_x \in S$ s.t. ${\ } b_x \leftarrow x$ in the FT structure function $\pi_\ft(S, b_f)$:
    \[
        \faultcreation \langle \capid_{\fdginstance_n}(x) \rangle \mapsto b_x
    \]
\end{proof}

\begin{figure}
    \centering
    \includegraphics[width=\linewidth]{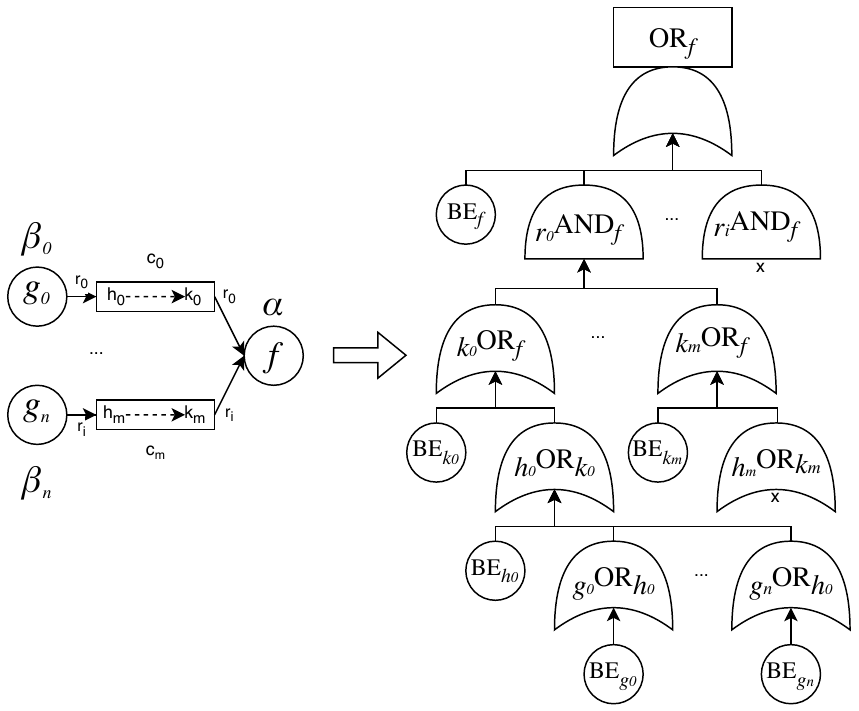}
    \caption{An illustration of the general case of the transformation.}\label{fig:ip_general_case}
\end{figure}

\begin{theorem}[IH]\label{thm:ip_ih}
    Assume a valid {\fdgname} $\fdginstance_0$ of size $K = I + M + N$, as depicted in \autoref{fig:ip_general_case}. Let $f \in \capability_{\fdginstance_0}$ be a capability, and let $\ft = \transform(\fdginstance_0, f)$ be the synthesized fault tree. The induction hypothesis is
    \[
        [\![ f ]\!]^{\fdginstance_{n}}_{\fdginstance_0} \iff \pi_\ft(S, f)
    \]
\end{theorem}

We derive the structure function for a generic {\fdgname} of size $K$ for reusability. We do the same for the structure function of the FT synthesized from this {\fdgname}.
\mypar{{\fdgname} Derivation}
    We evaluate the structure function $[\![ f ]\!]^{\fdginstance_n}_{\fdginstance_0}$ for the existence of an error over the capability $f$:
    \begin{align}\label{eq:ih_structure_function_derivation}
    \begin{split}
        [\![ f ]\!]^{\fdginstance_{n}}_{\fdginstance_0} &\iff \faultcreation^{\fdginstance_n}_{\fdginstance_0} \langle \mytext{id}_f \rangle \vee
            \\ &\biggl( \Bigl(
            \dots \wedge [\![k^{r_0}_m]\!]^{\fdginstance_{n}}_{\fdginstance_0} \Bigr) \vee \dots \vee \Bigl(
            \dots \wedge [\![k^{r_i}_m]\!]^{\fdginstance_{n}}_{\fdginstance_0} \Bigr) \biggr) \\
        [\![ k^{r_0}_0 ]\!]^{\fdginstance_{n}}_{\fdginstance_0} &\iff \faultcreation^{\fdginstance_n}_{\fdginstance_0} \langle \mytext{id}_{k^{r_0}_0} \rangle \vee \left( \exists \fdginstance_{n-j} : \fdg. {\ } [\![h^{r_0}_0]\!]^{\fdginstance_{n-j}}_{\fdginstance_0} \right) \\
        &\cdots\\
        [\![ k^{r_i}_m ]\!]^{\fdginstance_{n}}_{\fdginstance_0} &\iff \faultcreation^{\fdginstance_n}_{\fdginstance_0} \langle \mytext{id}_{k^{r_i}_m} \rangle \vee \left( \exists \fdginstance_{n-j} : \fdg. {\ } [\![h^{r_i}_m]\!]^{\fdginstance_{n-j}}_{\fdginstance_0} \right) \\
        [\![h^{r_0}_0]\!]^{\fdginstance_{n-j}}_{\fdginstance_0} &\iff \faultcreation^{\fdginstance_{n-j}}_{\fdginstance_0} \langle \mytext{id}_{h^{r_0}_0} \rangle \vee [\![g^{h^{r_0}_0}_0]\!]^{\fdginstance_{n-j}}_{\fdginstance_0} \vee \dots \vee [\![g^{h^{r_i}_m}_n]\!]^{\fdginstance_{n-j}}_{\fdginstance_0} \\
        &\cdots\\
        [\![h^{r_i}_m]\!]^{\fdginstance_{n-j}}_{\fdginstance_0} &\iff \faultcreation^{\fdginstance_{n-j}}_{\fdginstance_0} \langle \mytext{id}_{h^{r_i}_m} \rangle \vee [\![g^{h^{r_0}_0}_0]\!]^{\fdginstance_{n-j}}_{\fdginstance_0} \vee \dots \vee [\![g^{h^{r_i}_m}_n]\!]^{\fdginstance_{n-j}}_{\fdginstance_0} \\
        [\![g^{h^{r_0}_0}_0]\!]^{\fdginstance_{n-j}}_{\fdginstance_0} &\iff \faultcreation^{\fdginstance_{n-j}}_{\fdginstance_0} \langle \mytext{id}_{g^{h^{r_0}_0}_0} \rangle \\
        &\cdots\\
        [\![g^{h^{r_i}_m}_n]\!]^{\fdginstance_{n-j}}_{\fdginstance_0} &\iff \faultcreation^{\fdginstance_{n-j}}_{\fdginstance_0} \langle \mytext{id}_{g^{h^{r_i}_m}_n} \rangle \\
    \end{split}
    \end{align}

    Note that $\fdginstance_{n-j}$ is fresh in each term $[\![h^{r_x}_y]\!]^{\fdginstance_{n-j}}_{\fdginstance_0}$, though we have simplified its index here for readability.

\mypar{FT derivation}
    Assume a set $S$ of failed basic events, derived from the set of erroneous components as defined in \autoref{thm:completeness}. Recursively applying the fault tree structure function $\pi$ over the fault tree $\ft = \transform(\fdginstance_0, \target)$ yields:
    \begin{align}\label{eq:ih_ft_structure_function_derivation}
    \begin{split}
        \pi_\ft(S, \mytext{OR}_f) &\iff \mytext{BE}_f \in S \vee\\
            \bigl( \pi_\ft(S, &{}_{r_0}\mytext{AND}_f) \vee \dots \vee \pi_\ft(S, {}_{r_i}\mytext{AND}_f) \bigr) \\
        \pi_\ft(S, {}_{r0}\mytext{AND}_f) &\iff\\
            \pi_\ft(S, &{}_{k^{r_0}_0}\mytext{OR}_f) \wedge \dots \wedge \pi_\ft(S, {}_{k^{r_0}_m}\mytext{OR}_f) \\
        \cdots\\
        \pi_\ft(S, {}_{r_i}\mytext{AND}_f) &\iff\\
            \pi_\ft(S, &{}_{k^{r_i}_0}\mytext{OR}_f) \wedge \dots \wedge \pi_\ft(S, {}_{k^{r_i}_m}\mytext{OR}_f) \\
        \pi_\ft(S, {}_{k^{r_0}_0}\mytext{OR}_f)) &\iff \mytext{BE}_{k^{r_0}_0} \in S \vee \pi_\ft(S, {}_{h^{r_0}_0}\mytext{OR}_{k^{r_0}_0}) \\
        \cdots\\
        \pi_\ft(S, {}_{k^{r_i}_m}\mytext{OR}_f)) &\iff \mytext{BE}_{k^{r_i}_m} \in S \vee \pi_\ft(S, {}_{h^{r_i}_m}\mytext{OR}_{k^{r_i}_m}) \\
        \pi_\ft(S, {}_{h^{r_0}_0}\mytext{OR}_k) &\iff \mytext{BE}_{h^{r_0}_0} \in S \vee\\
            \pi_\ft(S, &{}_{g^{h^{r_0}_0}_0}\mytext{OR}_{h^{r_0}_0}) \vee \dots \vee \pi_\ft(S, {}_{g^{h^{r_i}_m}_n}\mytext{OR}_{h^{r_i}_m}) \\
        \cdots\\
        \pi_\ft(S, {}_{h^{r_i}_m}\mytext{OR}_k) &\iff \mytext{BE}_{h^{r_i}_m} \in S \vee\\
            \pi_\ft(S, &{}_{g^{h^{r_0}_0}_0}\mytext{OR}_{h^{r_0}_0}) \vee \dots \vee \pi_\ft(S, {}_{g^{h^{r_i}_m}_n}\mytext{OR}_{h^{r_i}_m}) \\
        \pi_\ft(S, {}_{g^{h^{r_0}_0}_0}\mytext{OR}_{h^{r_0}_0}) &\iff \mytext{BE}_{g^{h^{r_0}_0}_0} \in S\\
        \cdots\\
        \pi_\ft(S, {}_{g^{h^{r_i}_m}_n}\mytext{OR}_{h^{r_i}_m}) &\iff \mytext{BE}_{g^{h^{r_i}_m}_n} \in S
    \end{split}
    \end{align}

\begin{proof}[IH proof sketch]
    To prove \autoref{thm:ip_ih}, we induct over the size of $\fdginstance_0$ to show that $K + 1$ holds. We do so in three dimensions.

    \begin{enumerate}
        \item By \autoref{lem:ip_case_i}, we have that the IH holds for
        \[K + 1 = (I + 1) + M + N\]
        \item By \autoref{lem:ip_case_m}, we have that the IH holds for
        \[K + 1 = I + (M + 1) + N\]
        \item By \autoref{lem:ip_case_n}, we have that the IH holds for
        \[K + 1 = I + M + (N + 1)\]
    \end{enumerate}

    Furthermore, if $K$ increases with any elements that do not affect the terms of the structure function $[\![ f ]\!]^{\fdginstance_{n}}_{\fdginstance_0}$, then the IH holds trivially. We have thus shown that for any valid {\fdgname} of size $K + 1$, the IH holds.
\end{proof}

\begin{lemma}[Case $K + 1 = (I + 1) + M + N$]\label{lem:ip_case_i}

    Let $\fdginstance_0'$ be the extension of $\fdginstance_0$ by a fresh prerequisite resource $r_{i + 1}$. Let $f \in \capability_{\fdginstance_0'}$ be a capability, and let $\ft' = \transform(\fdginstance_0', f)$ be the synthesized fault tree. Then
    \[
        [\![ f ]\!]^{\fdginstance_{n}}_{\fdginstance_0'} \iff \pi_{\ft'}(S, f)
    \]
    
\end{lemma}

\begin{proof}[Case $I+1$]
    Recall the formulae \cref{eq:ih_structure_function_derivation} and \cref{eq:ih_ft_structure_function_derivation}. By adding an additional resource $r_{i+1}$, we obtain the following extension for the {\fdgname} structure function:
    \begin{align*}\label{eq:ih_i_structure_function_derivation}
        [\![ f ]\!]^{\fdginstance_{n}}_{\fdginstance_0} &\iff \faultcreation^{\fdginstance_n}_{\fdginstance_0} \langle \mytext{id}_f \rangle \vee
            \\ &\biggl( 
            \dots \vee
            \Bigl( \dots \wedge [\![k^{r_{i}}_m]\!]^{\fdginstance_{n}}_{\fdginstance_0} \Bigr) \bm{\vee \Bigl(} \mathbf{\dots} \boldsymbol{\wedge} \bm{ [\![k^{r_{i+1}}_m]\!]^{\fdginstance_{n}}_{\fdginstance_0} } \bm{\Bigr)} \biggr) \\
        %
        &\cdots\\
        [\![ k^{r_{i}}_m ]\!]^{\fdginstance_{n}}_{\fdginstance_0} &\iff \faultcreation^{\fdginstance_n}_{\fdginstance_0} \langle \mytext{id}_{k^{r_i}_m} \rangle \vee \left( \exists \fdginstance_{n-j} : \fdg. {\ } [\![h^{r_i}_m]\!]^{\fdginstance_{n-j}}_{\fdginstance_0} \right) \\
        \bm{[\![ k^{r_{i+1}}_m ]\!]^{\fdginstance_{n}}_{\fdginstance_0}} & \bm{\iff \faultcreation^{\fdginstance_n}_{\fdginstance_0} \langle \textbf{id}_{k^{r_{i+1}}_m} \rangle}\\
            &\bm{\vee \left( \exists \fdginstance_{n-j} : \fdg. {\ } [\![h^{r_{i+1}}_m]\!]^{\fdginstance_{n-j}}_{\fdginstance_0} \right)} \\
        %
        &\cdots\\
        [\![h^{r_i}_m]\!]^{\fdginstance_{n-j}}_{\fdginstance_0} &\iff \faultcreation^{\fdginstance_{n-j}}_{\fdginstance_0} \langle \mytext{id}_{h^{r_i}_m} \rangle  \vee \dots \vee [\![g^{h^{r_i}_m}_n]\!]^{\fdginstance_{n-j}}_{\fdginstance_0} \\
        \bm{[\![h^{r_{i+1}}_m]\!]^{\fdginstance_{n-j}}_{\fdginstance_0}} & \bm{\iff \faultcreation^{\fdginstance_{n-j}}_{\fdginstance_0} \langle \textbf{id}_{h^{r_{i+1}}_m} \rangle \vee} \mathbf{\dots} \bm{\vee [\![g^{h^{r_{i+1}}_m}_n]\!]^{\fdginstance_{n-j}}_{\fdginstance_0}} \\
        %
        &\cdots\\
        [\![g^{h^{r_i}_m}_n]\!]^{\fdginstance_{n-j}}_{\fdginstance_0} &\iff \faultcreation^{\fdginstance_{n-j}}_{\fdginstance_0} \langle \mytext{id}_{g^{h^{r_i}_m}_n} \rangle \\
        \bm{[\![g^{h^{r_{i+1}}_m}_n]\!]^{\fdginstance_{n-j}}_{\fdginstance_0}} &\bm{\iff \faultcreation^{\fdginstance_{n-j}}_{\fdginstance_0} \langle \textbf{id}_{g^{h^{r_{i+1}}_m}_n} \rangle} \\
    \end{align*}

    For the FT structure function, we obtain a similar extension:
    \begin{align*}
        \pi_\ft(S, \mytext{OR}_f) &\iff \mytext{BE}_f \in S \vee\\
            & \bigl( \dots \vee \pi_\ft(S, {}_{r_i}\mytext{AND}_f)\\
                &\qquad \bm{\vee \pi_\ft(S, {}_{r_{i+1}}\textbf{AND}_f)} \bigr) \\
        \cdots\\
        \pi_\ft(S, {}_{r_i}\mytext{AND}_f) &\iff \pi_\ft(S, {}_{k^{r_i}_0}\mytext{OR}_f) \wedge \dots\\
            &\wedge \pi_\ft(S, {}_{k^{r_i}_m}\mytext{OR}_f) \\
        \bm{\pi_\ft(S, {}_{r_{i+1}}\textbf{AND}_f)} & \bm{\iff \pi_\ft(S, {}_{k^{r_{i+1}}_0}\textbf{OR}_f)} \boldsymbol{\wedge} \mathbf{\dots}\\
            &\boldsymbol{\wedge} \bm{\pi_\ft(S, {}_{k^{r_{i+1}}_m}\textbf{OR}_f)} \\
        %
        \cdots\\
        \pi_\ft(S, {}_{k^{r_i}_m}\mytext{OR}_f)) &\iff \mytext{BE}_{k^{r_i}_m} \in S \vee \pi_\ft(S, {}_{h^{r_i}_m}\mytext{OR}_{k^{r_i}_m}) \\
        \bm{\pi_\ft(S, {}_{k^{r_{i+1}}_m}\textbf{OR}_f))} & \bm{\iff \textbf{BE}_{k^{r_{i+1}}_m} \in S}\\
            &\bm{\vee \pi_\ft(S, {}_{h^{r_{i+1}}_m}\textbf{OR}_{k^{r_{i+1}}_m})} \\
        %
        \cdots\\
        \pi_\ft(S, {}_{h^{r_i}_m}\mytext{OR}_k) &\iff \mytext{BE}_{h^{r_i}_m} \in S \vee \dots\\
            &\vee \pi_\ft(S, {}_{g^{h^{r_i}_m}_n}\mytext{OR}_{h^{r_i}_m}) \\
        \bm{\pi_\ft(S, {}_{h^{r_{i+1}}_m}\textbf{OR}_k)} & \bm{\iff \textbf{BE}_{h^{r_{i+1}}_m} \in S \vee} \mathbf{\dots}\\
            &\bm{\vee \pi_\ft(S, {}_{g^{h^{r_{i+1}}_m}_n}\textbf{OR}_{h^{r_{i+1}}_m})} \\
        %
        \cdots\\
        \pi_\ft(S, {}_{g^{h^{r_i}_m}_n}\mytext{OR}_{h^{r_i}_m}) &\iff \mytext{BE}_{g^{h^{r_i}_m}_n} \in S \\
        \bm{\pi_\ft(S, {}_{g^{h^{r_{i+1}}_m}_n}\textbf{OR}_{h^{r_{i+1}}_m})} & \bm{\iff \textbf{BE}_{g^{h^{r_{i+1}}_m}_n} \in S}
    \end{align*}

    Apply bottom-up substitution to observe that both resulting boolean formulae are extended with clauses that are disjunctive w.r.t. \cref{eq:ih_structure_function_derivation} and \cref{eq:ih_ft_structure_function_derivation} and between whose terms there is a one-to-one mapping. I.e.:
    \begin{multline*}
        [\![k^{r_{i+1}}_0]\!]^{\fdginstance_{n}}_{\fdginstance_0} \wedge \dots \wedge [\![k^{r_{i+1}}_m]\!]^{\fdginstance_{n}}_{\fdginstance_0} \iff\\
        \pi_\ft(S, {}_{k^{r_{i+1}}_0}\mytext{OR}_f) \wedge \dots \wedge \pi_\ft(S, {}_{k^{r_{i+1}}_m}\mytext{OR}_f)
    \end{multline*}

    And by the IH:
    \[
        [\![ f ]\!]^{\fdginstance_{n}}_{\fdginstance_0'} \iff \pi_{\ft'}(S, f)
    \]
    
\end{proof}

\begin{lemma}[Case $K + 1 = I + (M + 1) + N$]\label{lem:ip_case_m}

    Let $\fdginstance_0'$ be the extension of $\fdginstance_0$ by a fresh channel $c_{m + 1}$. Let $f \in \capability_{\fdginstance_0'}$ be a capability, and let $\ft' = \transform(\fdginstance_0', f)$ be the synthesized fault tree. Then
    \[
        [\![ f ]\!]^{\fdginstance_{n}}_{\fdginstance_0'} \iff \pi_{\ft'}(S, f)
    \]
    
\end{lemma}

\begin{proof}[Case $M + 1$]
    Recall the formulae \cref{eq:ih_structure_function_derivation} and \cref{eq:ih_ft_structure_function_derivation}. By adding an additional channel $c_{m+1}$, we obtain the following extension for the {\fdgname} structure function:

    \begin{align*}\label{eq:ih_i_structure_function_derivation}
        [\![ f ]\!]^{\fdginstance_{n}}_{\fdginstance_0} &\iff \faultcreation^{\fdginstance_n}_{\fdginstance_0} \langle \mytext{id}_f \rangle \vee
            \\ &\biggl( 
            \dots \vee
            \Bigl( \dots \wedge [\![k^{r_{i}}_{m}]\!]^{\fdginstance_{n}}_{\fdginstance_0} \boldsymbol{\wedge} \bm{ [\![k^{r_{i}}_{m+1}]\!]^{\fdginstance_{n}}_{\fdginstance_0} } \bm{\Bigr)} \biggr) \\
        %
        &\cdots\\
        [\![ k^{r_{i}}_m ]\!]^{\fdginstance_{n}}_{\fdginstance_0} &\iff \faultcreation^{\fdginstance_n}_{\fdginstance_0} \langle \mytext{id}_{k^{r_i}_m} \rangle \vee \left( \exists \fdginstance_{n-j} : \fdg. {\ } [\![h^{r_i}_m]\!]^{\fdginstance_{n-j}}_{\fdginstance_0} \right) \\
        \bm{[\![ k^{r_{i}}_{m+1} ]\!]^{\fdginstance_{n}}_{\fdginstance_0}} & \bm{\iff \faultcreation^{\fdginstance_n}_{\fdginstance_0} \langle \textbf{id}_{k^{r_{i}}_{m+1}} \rangle}\\
            &\qquad \qquad \bm{\vee \left( \exists \fdginstance_{n-j} : \fdg. {\ } [\![h^{r_{i}}_{m+1}]\!]^{\fdginstance_{n-j}}_{\fdginstance_0} \right)} \\
        %
        &\cdots\\
        [\![h^{r_i}_m]\!]^{\fdginstance_{n-j}}_{\fdginstance_0} &\iff \faultcreation^{\fdginstance_{n-j}}_{\fdginstance_0} \langle \mytext{id}_{h^{r_i}_m} \rangle \vee \dots \vee [\![g^{h^{r_i}_m}_n]\!]^{\fdginstance_{n-j}}_{\fdginstance_0} \\
        \bm{[\![h^{r_{i}}_{m+1}]\!]^{\fdginstance_{n-j}}_{\fdginstance_0}} & \bm{\iff \faultcreation^{\fdginstance_{n-j}}_{\fdginstance_0} \langle \textbf{id}_{h^{r_{i}}_{m+1}} \rangle \vee} \mathbf{\dots} \bm{\vee [\![g^{h^{r_{i}}_{m+1}}_n]\!]^{\fdginstance_{n-j}}_{\fdginstance_0}} \\
        %
        &\cdots\\
        [\![g^{h^{r_i}_m}_n]\!]^{\fdginstance_{n-j}}_{\fdginstance_0} &\iff \faultcreation^{\fdginstance_{n-j}}_{\fdginstance_0} \langle \mytext{id}_{g^{h^{r_i}_m}_n} \rangle \\
        \bm{[\![g^{h^{r_{i}}_{m+1}}_n]\!]^{\fdginstance_{n-j}}_{\fdginstance_0}} &\bm{\iff \faultcreation^{\fdginstance_{n-j}}_{\fdginstance_0} \langle \textbf{id}_{g^{h^{r_{i}}_{m+1}}_n} \rangle} \\
    \end{align*}

    \noindent For the FT structure function, we obtain a similar result:
    \begin{align*}
        \pi_\ft(S, \mytext{OR}_f) &\iff \mytext{BE}_f \in S \vee\\
            & \bigl( \dots \vee \pi_\ft(S, {}_{r_i}\mytext{AND}_f) \bigr) \\
        \cdots\\
        \pi_\ft(S, {}_{r_i}\mytext{AND}_f) &\iff \dots \wedge \pi_\ft(S, {}_{k^{r_i}_m}\mytext{OR}_f)\\
            &\boldsymbol{\wedge} \bm{\pi_\ft(S, {}_{k^{r_i}_{m+1}}\textbf{OR}_f)} \\
        %
        \cdots\\
        \pi_\ft(S, {}_{k^{r_i}_m}\mytext{OR}_f)) &\iff \mytext{BE}_{k^{r_i}_m} \in S\\
            &\vee \pi_\ft(S, {}_{h^{r_i}_m}\mytext{OR}_{k^{r_i}_m}) \\
        \bm{\pi_\ft(S, {}_{k^{r_{i}}_{m+1}}\textbf{OR}_f))} & \bm{\iff \textbf{BE}_{k^{r_{i}}_{m+1}} \in S}\\
            &\bm{\vee \pi_\ft(S, {}_{h^{r_{i}}_{m+1}}\textbf{OR}_{k^{r_{i}}_{m+1}})} \\
        %
        \cdots\\
        \pi_\ft(S, {}_{h^{r_i}_m}\mytext{OR}_k) &\iff \mytext{BE}_{h^{r_i}_m} \in S\\
            &\vee \pi_\ft(S, {}_{g^{h^{r_i}_m}_n}\mytext{OR}_{h^{r_i}_m}) \\
        \bm{\pi_\ft(S, {}_{h^{r_{i}}_{m+1}}\textbf{OR}_k)} & \bm{\iff \textbf{BE}_{h^{r_{i}}_{m+1}} \in S \vee} \mathbf{\dots}\\
            &\bm{\vee \pi_\ft(S, {}_{g^{h^{r_{i}}_{m+1}}_n}\textbf{OR}_{h^{r_{i}}_{m+1}})} \\
        %
        \cdots\\
        \pi_\ft(S, {}_{g^{h^{r_i}_m}_n}\mytext{OR}_{h^{r_i}_m}) &\iff \mytext{BE}_{g^{h^{r_i}_m}_n} \in S \\
        \bm{\pi_\ft(S, {}_{g^{h^{r_{i}}_{m+1}}_n}\textbf{OR}_{h^{r_{i}}_{m+1}})} & \bm{\iff \textbf{BE}_{g^{h^{r_{i}}_{m+1}}_n} \in S}
    \end{align*}

    Apply bottom-up substitution to observe that both resulting boolean formulae are extended with clauses that are conjunctive w.r.t. \cref{eq:ih_structure_function_derivation} and \cref{eq:ih_ft_structure_function_derivation} and between whose terms there is a one-to-one mapping. I.e.:
    \begin{equation*}
        [\![k^{r_{i}}_{m+1}]\!]^{\fdginstance_{n}}_{\fdginstance_0} \iff \pi_\ft(S, {}_{k^{r_i}_{m+1}}\mytext{OR}_f)
    \end{equation*}

    And by the IH:
    \[
        [\![ f ]\!]^{\fdginstance_{n}}_{\fdginstance_0'} \iff \pi_{\ft'}(S, f)
    \]
    
\end{proof}

\begin{lemma}[Case $K + 1 = I + M + (N + 1)$]\label{lem:ip_case_n}

    Let $\fdginstance_0'$ be the extension of $\fdginstance_0$ by a fresh component capability $g_{n + 1}$. Let $f \in \capability_{\fdginstance_0'}$ be a capability, and let $\ft' = \transform(\fdginstance_0', f)$ be the synthesized fault tree. Then
    \[
        [\![ f ]\!]^{\fdginstance_{n}}_{\fdginstance_0'} \iff \pi_{\ft'}(S, f)
    \]
    
\end{lemma}

\begin{proof}[Case $N + 1$]
    Recall the formulae \cref{eq:ih_structure_function_derivation} and \cref{eq:ih_ft_structure_function_derivation}. By adding an additional component capability $g_{n+1}$, we obtain the following extension for the {\fdgname} structure function:

    \begin{align*}\label{eq:ih_i_structure_function_derivation}
        [\![ f ]\!]^{\fdginstance_{n}}_{\fdginstance_0} &\iff \faultcreation^{\fdginstance_n}_{\fdginstance_0} \langle \mytext{id}_f \rangle \vee
            \\ &\biggl( \Bigl(
            \dots \wedge [\![k^{r_0}_m]\!]^{\fdginstance_{n}}_{\fdginstance_0} \Bigr) \vee \dots \vee \Bigl(
            \dots \wedge [\![k^{r_i}_m]\!]^{\fdginstance_{n}}_{\fdginstance_0} \Bigr) \biggr) \\
        %
        &\cdots\\
        [\![ k^{r_{i}}_m ]\!]^{\fdginstance_{n}}_{\fdginstance_0} &\iff \faultcreation^{\fdginstance_n}_{\fdginstance_0} \langle \mytext{id}_{k^{r_i}_m} \rangle \vee \left( \exists \fdginstance_{n-j} : \fdg. {\ } [\![h^{r_i}_m]\!]^{\fdginstance_{n-j}}_{\fdginstance_0} \right) \\
        %
        &\cdots\\
        [\![h^{r_i}_m]\!]^{\fdginstance_{n-j}}_{\fdginstance_0} &\iff \faultcreation^{\fdginstance_{n-j}}_{\fdginstance_0} \langle \mytext{id}_{h^{r_i}_m} \rangle \vee \dots \vee [\![g^{h^{r_i}_m}_n]\!]^{\fdginstance_{n-j}}_{\fdginstance_0}\\
            &\bm{\vee [\![g^{h^{r_i}_m}_{n+1}]\!]^{\fdginstance_{n-j}}_{\fdginstance_0}} \\
        %
        &\cdots\\
        [\![g^{h^{r_i}_m}_n]\!]^{\fdginstance_{n-j}}_{\fdginstance_0} &\iff \faultcreation^{\fdginstance_{n-j}}_{\fdginstance_0} \langle \mytext{id}_{g^{h^{r_i}_m}_n} \rangle \\
        \bm{[\![g^{h^{r_{i}}_{m}}_{n+1}]\!]^{\fdginstance_{n-j}}_{\fdginstance_0}} &\bm{\iff \faultcreation^{\fdginstance_{n-j}}_{\fdginstance_0} \langle \textbf{id}_{g^{h^{r_{i}}_{m}}_{n+1}} \rangle} \\
    \end{align*}

    For the FT structure function, we obtain a similar extension:
    \begin{align*}
        \pi_\ft(S, \mytext{OR}_f) &\iff \mytext{BE}_f \in S \vee\\
            & \bigl( \dots \vee \pi_\ft(S, {}_{r_i}\mytext{AND}_f) \bigr) \\
        \cdots\\
        \pi_\ft(S, {}_{r_i}\mytext{AND}_f) &\iff \dots \wedge \pi_\ft(S, {}_{k^{r_i}_m}\mytext{OR}_f) \\
        %
        \cdots\\
        \pi_\ft(S, {}_{k^{r_i}_m}\mytext{OR}_f)) &\iff \mytext{BE}_{k^{r_i}_m} \in S \vee \pi_\ft(S, {}_{h^{r_i}_m}\mytext{OR}_{k^{r_i}_m}) \\
        %
        \cdots\\
        \pi_\ft(S, {}_{h^{r_i}_m}\mytext{OR}_k) &\iff \mytext{BE}_{h^{r_i}_m} \in S \vee \dots\\
            &\qquad \vee \pi_\ft(S, {}_{g^{h^{r_i}_m}_n}\mytext{OR}_{h^{r_i}_m})\\
            &\qquad \bm{\vee \pi_\ft(S, {}_{g^{h^{r_i}_m}_{n+1}}\mytext{OR}_{h^{r_i}_m})} \\
        %
        \cdots\\
        \pi_\ft(S, {}_{g^{h^{r_i}_m}_n}\mytext{OR}_{h^{r_i}_m}) &\iff \mytext{BE}_{g^{h^{r_i}_m}_n} \in S \\
        \bm{\pi_\ft(S, {}_{g^{h^{r_{i}}_{m}}_{n+1}}\textbf{OR}_{h^{r_{i}}_{m}})} & \bm{\iff \textbf{BE}_{g^{h^{r_{i}}_{m}}_{n+1}} \in S}
    \end{align*}

    Apply bottom-up substitution to observe that both resulting boolean formulae are extended with terms that are disjunctive w.r.t. \cref{eq:ih_structure_function_derivation} and \cref{eq:ih_ft_structure_function_derivation} and between who there is a one-to-one mapping. I.e.:
    \begin{equation*}
        [\![g^{h^{r_i}_m}_{n+1}]\!]^{\fdginstance_{n-j}}_{\fdginstance_0} \iff \pi_\ft(S, {}_{g^{h^{r_i}_m}_{n+1}}\mytext{OR}_{h^{r_i}_m})
    \end{equation*}

    And by the IH:
    \[
        [\![ f ]\!]^{\fdginstance_{n}}_{\fdginstance_0'} \iff \pi_{\ft'}(S, f)
    \]
    
\end{proof}

\section{Related work}\label{sec:related_work}

\subsection{Failure Mode Effects Analysis}

Failure Mode Effects Analysis (FMEA) is an inductive reliability analysis method in which experts and analysts consider the potential effects of envisioned failure modes of a system. This knowledge is captured in FMEA tables -- a human-readable, textual format, listing a subset of a system's potential failure modes. Due to their natural-language descriptions, FMEA tables cannot be evaluated with classical algorithms, though some progress in mechanized natural-language interpretation has been made\cite{xu2020data}.

Since manual FMEA is a labor-intensive process, it is usually carried out in the later design phases\sidenote{This is something PK could write about!}. By that point, changes to the system design are much more difficult to accommodate and consequently much more costly\cite{incose2023incose}. As a result, the emphasis often lies on the mitigation of known design flaws, rather than improving the system design\sidenote{Again, something PK can corroborate.}.

FTA, in contrast to FMEA is a deductive approach. FT construction starts from the undesired top-level event, and potential causes are explored in a recursive top-down manner. FTs provides an intuitive visualization of how failure modes and their interactions result in system-level failure, making them easier to interpret than FMEA tables -- especially for large systems. Furthermore, fault trees are formally defined and they afford many metrics that can be computed algorithmically (i.e. minimal cut sets and minimal path sets\cite{Ruijters2015}). The strength of FTA lies in the combination of visual failure mode interactions and automatic reliability metrics.

If FTA reliability metrics can be obtained with low effort in the early design phases, design changes based on new insights are less costly to implement. Engineers and decision-makers have the ability to make systems more robust by design, and simultaneously to reduce costly ``just in case'' design decisions because of uncertainty in this early stage. However, even though analysis is fully mechanized, the manual \textit{construction} of fault trees typically suffers from the same labor-intensiveness and consistency challenges associated with manual FMEA construction.

\subsection{Fault Tree and FTA Extensions}
Dynamic Fault Trees (DFTs) extend classical FTs by introducing temporal and state-dependent failure semantics~\cite{Ruijters2015}. In addition to standard Boolean gates, DFTs incorporate constructs such as priority-AND gates, functional dependencies\footnote{The semantics of the term ``functional dependency'' in DFTs is not the same as how we define it in this work.}, and spare management, enabling the representation of dynamic system behavior and sequencing constraints. DFTs have been widely applied in the analysis of safety-critical systems. Several approaches have investigated the automatic generation of DFTs from architectural or behavioral models. However, the detailed level of knowledge required to create a DFT is unavailable at early design stages, which makes the DFT an inappropriate formalism for the focus of this work.

A more domain-specific extension of fault trees, called Component Fault Trees\cite{component_fault_trees}(CFTs), makes the role of components as the bearers of faults explicit by allowing modelers to encapsulate internal component failure modes. However, the lack of an ontological foundation allows for the construction of CFTs that are unintended by the modeler, e.g. those that have no consistent notion of (component) identity, time or causality. Furthermore, their manual construction confronts stakeholders with the same challenges as in the case of standard FTs.

Nicoletti et al. introduce Probabilistic Fault tree Logic (PFL) -- a query language for fault trees, allowing analysts to obtain user-specific and complex metrics of interest \cite{Nicoletti2023PFL}. PFL also allows for locally setting evidence in queries without changing the values of the underlying FT. In the context of this work, analysts can use PFL to obtain quantitative FT metrics despite the fact that the FTs synthesized from {\fdgname}s do no ``contain'' numerical failure probabilities.

\subsection{Fault Tree Generation from MBSE}


There is extensive work on (semi-)automatically transforming failure-decorated MBSE models to (dynamic) FTs.

Baklouti et al. generate dynamic FTs from decorated SysML models by introducing a novel redundancy profile\cite{baklouti_dynamic_2020}. However, this requires explicit safety assessment from the modeler (i.e. redundancy does not \textit{follow} from the architecture, it is \textit{assumed} of the architecture) and its tooling is proprietary to SysML (version 1) extended with the authors' profile. Furthermore, despite its practical applicability, the work lacks an ontological foundation.

Castet et al. introduce a fault management ontology, providing engineers at JPL with a common vocabulary of problematic system behavior\cite{castet2016fault}. A tool that integrates this ontology with SysML tooling is presented in \cite{castet2018failure}. The fault management ontology describes concepts like constraints, violations, mitigation and explanations. These concepts are complementary to the ontological concepts (e.g. components, capability, faults and errors) that we introduce in this work.

Mhenni et al. \cite{Mhenni2014} describe a method to synthesize FTs from SysML's Internal Block Diagrams (IBD). IBDs represent a system's internal structure and the interactions between its components through ports. The FT synthesis algorithm is done by a combination of standard graph traversal algorithms and generating sub-FTs for known IBD patterns. This is similar in spirit to the transformation we propose in this work. However, the existing approach relies on specific SysML semantics. In contrast, the conceptual model proposed in our work is generally applicable to any KG that contain structural and functional system knowledge, regardless of the modeling language in which such system knowledge is originally expressed. This allows us to integrate different models (in different modeling languages) that describe the same system.

\mypar{Alternative System Models}
The work by Bozzano et al. in \cite{Bozzano2011} allows for the automatic generation of FTs from SLIM models (an extension of AADL). While these FTs are expressive and accurate, the synthesis methodology requires detailed nominal models and explicit error models of a system.

Majdara and Wakabayashi proposee a method to generate FTs from a system model using so-called function tables for components~\cite{majdara2009new}. Such a table describes the output of a component given its current failure status. Additionally, there are state transition tables that describe the state changes that a component can go through. Lastly, they employ a notion of \textit{flow}, which is similar in spirit to our own notion of \textit{resource types}. The inter-connectivity of components that exchange flows is the basis for the FT synthesis algorithm. Though these works address some of the challenges that we also touch on in this paper, the ontological foundation is minimal. Furthermore, the need for function tables and state transition tables, which encode failure states and their consequences, incurs the same modeling burden that many of the MBSE approaches suffer from.

What all these approaches have in common, is that FT synthesis relies on manually decorating system architecture with failure information. In doing so, they allow for the \textit{recording} of expert knowledge about specific failure modes. In contrast, {\fdgname}s do not require (and indeed do not facilitate) the manual modeling of failure behavior. Therefore, the ontology we present must encode sufficiently rich semantics to enable the \textit{inference} of failure modes -- and this by the very same \textit{mechanisms} that experts may employ.

\subsection{Embedding Fault Knowledge in KGs}
When fault semantics are directly encoded in an ontology, like in that proposed by Meng et al. \cite{Meng2024}, generating FTs from KGs that instantiate such an ontology is relatively straight-forward. The ontology that Shen et al. present in \cite{Shen2023} even explicitly encodes the basic FT concepts. However, both these ontologies describe KGs where explicit fault knowledge is already present. The synthesis of FTs from an ontology about the structural and functional aspects of system architectures, \textit{without} specified fault knowledge, is a key value proposition of our approach.

There are many automated methods to extract knowledge graphs from unstructured data~\cite{zhong2023}. In this work, we assume that a knowledge graph representing the system architecture is available.

\subsection{Ontology for Risk}
The Common Ontology of Value and Risk (COVER)~\cite{sales_cover} proposes an ontological foundation for reasoning about risk, causality, and events in complex socio-technical systems. COVER emphasizes the distinction between risk and value perception, event types, and situations, enabling a semantically rigorous treatment of social risk and value experience. COVER provides explicit conceptual grounding for what objects participate in risk and value events, and how such events are causally connected. Our work aligns with this perspective by treating fault creation, fault activation, and error propagation as ontologically distinct event types grounded in UFO. However, whereas COVER primarily focuses on high-level conceptual modeling of events and risk, our approach operationalizes these ideas for the automated synthesis of FTs. Furthermore, the domain-specific notions relevant to CPSs are not captured by COVER (though the possibility for aligned of {\fdgname}s with COVER has fruitful potential).

Nicoletti et al. explicitly highlight the importance of objects, e.g. components, in FTA and Attack Tree Analysis explicit. The WATCHDOG \cite{nicoletti_watchdog} approach, based on COVER, proposes object-oriented disruption graphs (DOGs). COVER describes how Threat Objects and Objects at Risk can participate in Threat Events and Loss Events. WATCHDOG draws a parallel between the preconditions of such events and the objects that participate in them. However, these conditions are, at the leafs of a DOG, \textit{enablers} of a basic events. The work does not explicitly consider the mechanisms of non-synchronic situations (i.e. the world at different time points), and as such leaves a notion of causality between events underspecified. Because of the gap between ``enabling conditions'' and causality, this work is incompatible with the type of functional dependency analysis we focus on here. Furthermore, by allowing \textit{any kind of object} in the model, the approach relies on manual, iterative expert modeling for the correct construction of DOGs (e.g. the correct modeling of causal relationships between events is not guided or enforced by the model). In contrast, since {\fdgname}s are domain-specific to CPS architecture, they are sufficiently constrained to enable fully automated synthesis of FTs.

Venceslau et al.~\cite{Venceslau2014} later extended the work of Majdara and Wakabayashi~\cite{majdara2009new} with an ontology (in a much leaner sense of the word than what we commit to here) and applied it to an industrial plant case study. Though there are works in the direction of an ontology for risk models~\cite{oliveira2025toward}, a comprehensive ontological analysis of CPSs for reliability analysis is lacking.

Another general observation is that most of the state of the art focuses on vertical failure propagation, i.e. propagation of failure along the mereological composition of a system. This in contrast to error propagation along its horizontal integration, which is what we focus on in this work.

\section{Conclusion}

This paper introduced the {\fdgnamefull} ({\fdgname}), a semantically and ontologically grounded model for representing functional dependencies and fault propagation in cyber-physical systems (CPSs). By combining principles from Model-Based Systems Engineering, knowledge graphs, and the Unified Foundational Ontology (UFO), the {\fdgname} provides a lightweight yet expressive intermediate representation that bridges CPS architectural knowledge and reliability analysis.

Building upon this model, we presented an automated synthesis pipeline that transforms knowledge-graph representations of CPS architectures into fault trees suitable for classical Fault Tree Analysis. In contrast to existing approaches that rely heavily on manually decorated architectural models and expert-specified failure behavior, our approach infers fault propagation structure directly from functional dependencies encoded in the {\fdgname}. The model explicitly distinguishes between faults, fault activations, errors, events, and situations, thereby addressing several semantic ambiguities that are commonly conflated in traditional fault tree methodologies.

The proposed ontology and formalization enable the unified representation of heterogeneous engineering knowledge while remaining sufficiently lightweight for early design stages, where detailed behavioral information is often unavailable. Furthermore, grounding the framework in UFO provides rigorous semantics for causality, object participation, event manifestation, and dependency, improving consistency and traceability across synthesized reliability models.

Overall, this work demonstrates that meaningful and semantically rich fault trees can be synthesized automatically from minimally structured CPS architectural knowledge. This opens opportunities for earlier and more scalable reliability analysis, reducing the effort associated with manual fault modeling and supporting more informed architectural decision-making during system development.

Future work includes extending the framework with quantitative reliability information, integrating behavioral and temporal semantics, incorporating uncertainty and probabilistic reasoning, and validating the approach on large-scale industrial CPS case studies. In addition, future research may explore tighter integration with MBSE toolchains and automated extraction of {\fdgname}s from heterogeneous engineering artifacts and natural-language system documentation.

\subsection{Acknowledgments}
\noindent The authors would like to thank dr. Claudenir Morais Fonseca for the insightful discussions and valuable advice on OntoUML and the Unified Foundational Ontology (UFO), which greatly contributed to the development of this work.


\medskip

\noindent \begin{minipage}{3.8em}
    \includegraphics[width=3.5em]{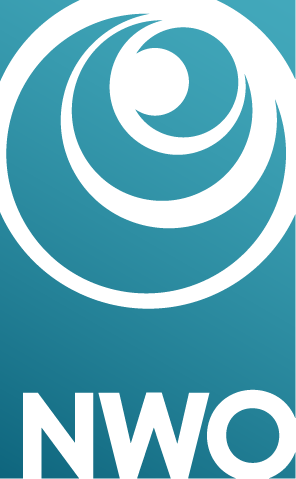}
  \end{minipage}
  \begin{minipage}{.4\textwidth}
    This publication is part of the project \href{https://zorro-project.nl}{ZORRO} with project number KICH1.ST02.21.003 of the research programme Key Enabling Technologies (KIC) which is (partly) financed by the Dutch Research Council (NWO). 

  \end{minipage}

\bibliographystyle{plain}
\bibliography{bibliography}



\end{document}